\documentclass[12pt,english]{article}
\usepackage[round,sort&compress]{natbib}
\usepackage{bm}
\usepackage{amssymb}
\usepackage{amsbsy}
\usepackage{mathrsfs}
\usepackage{graphicx}
\usepackage{amsmath}
\usepackage{amsthm}
\usepackage{multirow}
\usepackage[T1]{fontenc}
\usepackage[utf8]{inputenc}
\usepackage{babel}
\usepackage[margin=1.1in]{geometry}
\usepackage[official,right]{eurosym}
\usepackage{multirow}
\usepackage{array}
\usepackage{mdwlist}
\usepackage{lscape}
\usepackage{pifont}
\usepackage{arydshln}

\usepackage[para, flushleft]{threeparttable}
\usepackage{amsthm}
\usepackage{mathrsfs}
\usepackage{amssymb}
\usepackage{tikz}
\usepackage{mathrsfs}
\usepackage{amssymb}

\usepackage{float}
\restylefloat{table}

\usepackage{arydshln}
\usepackage{subcaption}
\usepackage{float}
\restylefloat{table}
\usepackage{hyperref}

\usetikzlibrary{calc}
\usetikzlibrary{patterns, arrows.meta}
\usetikzlibrary{patterns,math}
\usetikzlibrary{arrows}
\usetikzlibrary{shapes.geometric}
\usetikzlibrary{arrows,decorations.markings}
\usetikzlibrary{decorations.pathreplacing}

\newtheorem{theorem}{Theorem}
\newtheorem{proposition}{Proposition}
\newtheorem{lemma}{Lemma}
\newtheorem{corollary}{Corollary}

\theoremstyle{definition}

\newtheorem{claim}{Claim}
\newtheorem{example}{Example}
\newtheorem{remark}{Remark}
\newtheorem{fact}{Fact}
\newtheorem{observation}{Observation}
\newtheorem{result}{Result}

\newcommand*{\cproofname}{Proof}

\usepackage[normalem]{ulem} 

\begin{document}
\title{
\Large{Efficiency in Multiple-Type Housing Markets
\thanks{I gratefully acknowledge financial support from the Swiss National Science Foundation (SNSF) through Project 100018$\_$192583. Additionally, I would like to acknowledge the University of Lausanne for their generous financial support as well. Moreover, this article is based upon work supported by the National Science Foundation under Grant No.~DMS-1928930 and by the Alfred P.~Sloan Foundation under grant G-2021-16778, while I was in residence at the Simons Laufer Mathematical Sciences Institute (formerly MSRI) in Berkeley, California, during the Fall 2023 semester. Furthermore, I am grateful to Bettina Klaus for her invaluable advice, support, and encouragement, and for her patient reviews and detailed suggestions. I would like to thank Takuma Wakayama, William Thomson, Flip Klijn, and Ning Neil Yu, and the audience at  several seminars and conferences for their very helpful comments and discussions. 
All remaining errors are my own. 
This paper was awarded the SING Best Student Paper Prize at the 18th European Meeting on Game Theory, and it subsumes part of the results contained in my job market paper ``Endowments-swapping-proofness and Efficiency in Multiple-type Housing Markets,'' which was awarded the Kanematsu Prize from the Research Institute for Economics and Business Administration, Kobe University, in 2022.}
}
}

\author{Di Feng\thanks{Department of Finance, Dongbei University of Finance and Economics, Jianshan Jie 217, 116025, Shahekou, Dalian, P.R.\ China; \textit{e-mail}: \href{mailto:dfeng@dufe.edu.cn}{\tt dfeng@dufe.edu.cn}.}}

\date{\today}
\maketitle

\begin{abstract}
\linespread{1.1}\selectfont	
\noindent We consider multiple-type housing markets \citep{moulin1995}, which extend Shapley-Scarf housing markets \citep{shapley1974} from one dimension to higher dimensions. 
In this model, \textsl{Pareto efficiency} is incompatible with \textsl{individual rationality} and \textsl{strategy-proofness} \citep{konishi2001}. Therefore, we consider two weaker efficiency properties:
	\textsl{coordinatewise efficiency} and \textsl{pairwise efficiency}. 
	
	We show that these two properties both  (i) are compatible with \textsl{individual rationality} and \textsl{strategy-proofness}, and (ii) help us to identify two specific mechanisms. 
	To be more precise, on various domains of preference profiles, together with other well-studied properties (\textsl{individual rationality}, \textsl{strategy-proofness}, and \textsl{non-bossiness}), \textsl{coordinatewise efficiency} and \textsl{pairwise efficiency} respectively characterize two extensions of the top-trading-cycles mechanism (TTC): the coordinatewise top-trading-cycles mechanism (cTTC) and the bundle top-trading-cycles mechanism (bTTC). For multiple-type housing markets with strict preferences, our characterization of bTTC constitutes the first characterization of an extension of the prominent TTC mechanism.

    
Our proof is nonstandard and its novelty has independent methodological interest. Specifically, the absence of \textsl{non-bossiness} in the cTTC characterization and its presence in the bTTC characterization highlight both the uniqueness of our proof approach and the differences between our results and those in the existing literature.
\medskip

\medskip



\end{abstract}


\section{Introduction}\label{sec:intro}
\subsection{Motivating examples and main results}
We consider two multiple-type objects reallocation problems represented below, which motivates our research.\smallskip

\noindent\textbf{Task reallocation.}
Consider a scenario where we are organizing paper presentations for a class of ten students, we have ten papers and ten dates that need to be distributed among the students. Initially, the assignments might be made randomly. However, students have different preferences over paper-date pairs. Some students may prefer to present earlier and have a preference for presenting a paper on a specific topic of personal interest, while others may prefer to do it later and prefer another topic.
Therefore, reallocating papers and reallocating dates among the students, may be Pareto improving.  \medskip


\noindent\textbf{Job rotation.}
(1) Time schedule.
Every doctor enrolled in a hospital residency program is obligated to take turns being on an emergency duty schedule every month.
The schedule includes all doctors in the hospital and covers an entire year. The schedule is automatically carried over to the next year, but can be modified at the beginning of each year. When making changes, the preferences of each doctor regarding their duty schedule in the new year can be taken into consideration. 
(2) Location arrangement. Similarly, we can consider scheduling issues based on different locations. For instance, one anesthesiologist may be on duty at the North zone for the entire day, while another is at the South zone. Clearly, sometimes they could exchange their locations.\medskip

The common characteristics of the above examples are as follows: there are several objects that are labeled by different types and assigned to a group of agents (in the first example, there are two types, undergraduate course and graduate course, and in the second example, types refer to time periods). Each agent owns one object of each type and consumes exactly one object of each type.\footnote{Although doctors may not officially own certain time slots, they may unofficially inherit their time slot from the previous year, which can be regarded as their ownership of the slot. Furthermore, some people may officially own certain duties in certain situations, see \citet{klaus2008} in detail.} 
Additionally, there are no monetary transfers during the reallocation process. Therefore, agents have preferences over bundles, each consisting of one object of each type.\smallskip

We investigate a multiple-type reallocation model without monetary transfers that has these characteristics: the multiple-type housing markets model, as introduced by \citet{moulin1995}. In this model, a (re)allocation specifies how the objects are assigned to agents, and a mechanism is a mapping from a class of agents' preference profiles to the set of allocations.\smallskip

A common preference domain is the strict preference domain, in which all bundles are linearly ranked. In this domain, agents' preferences may exhibit complementarity. We also consider a subdomain of strict preferences which captures the substitution effect: (strictly) separable preferences.
Separability means that each agent has a list of preferences over objects of each type, and if two bundles differ only in one type, then the comparison between the two bundles is based on the agent's preference for that type. In a related paper, \cite{echenique2022} consider dichotomous preferences, which allow indifference between bundles.
\smallskip

In this paper, 
we consider the following properties for mechanisms.
\textsl{Individual rationality} is a voluntary participation condition that states that no agent will be worse off after the reallocation.
\textsl{Strategy-proofness} states that no agent can benefit from misreporting his preferences, and \textsl{group strategy-proofness} means that no group of agents can benefit from joint misreporting. 
\textsl{Non-bossiness} says that if a change in an agent's preferences does not affect his allocation, then it should not result in a change in the allocation of others. Consequently, the entire allocation would remain unchanged.
An allocation can be \textsl{Pareto improved} if an agent can be made better off without reducing the welfare of others.
\textsl{Pareto efficiency} states that the selected allocation cannot be \textsl{Pareto improved}. 
Similarly, an allocation can be \textsl{coordinatewise improved}
if there exists a Pareto improvement that reallocates exactly one object for each agent,
and \textsl{coordinatewise efficiency} means that the selected allocation cannot be \textsl{coordinatewise improved}.\footnote{The idea of \textsl{coordinatewise efficiency} is brorrowed from \citet{biro2022}, , which originally refers to it as \textsl{ig-Pareto-efficiency}. 
This property and its variants are also well-studied in the literature, e.g., \citet{aziz2019efficient}, \citet{caspari2020}, and \citet{coreno2022}.} 
Finally, an allocation can be \textsl{pairwise improved} if two agents can be both better off by exchanging their received entire bundles, and \textsl{pairwise efficiency} says that the selected allocation cannot be  \textsl{pairwise improved}.\smallskip\footnote{ 
This property is introduced by \citet{ekici2022}, which originally refers to it as \textsl{pair efficiency}.} 

When there is only one type, our model reduces to the Shapley-Scarf housing market model \citep{shapley1974}, and it is known that for such markets, \textsl{Pareto efficiency} is compatible with \textsl{individual rationality} and \textsl{strategy-proofness}. Furthermore, these three properties together uniquely identify the prominent top-trading-cycles mechanism (TTC) \citep{ma1994,svensson1999}. 
However, when there are multiple types, \textsl{Pareto efficiency} is incompatible with \textsl{individual rationality} and \textsl{strategy-proofness} \citep{konishi2001}. 
Therefore, we aim to determine the level of efficiency that remains feasible by preserving \textsl{individual rationality} and \textsl{strategy-proofness}, i.e., we investigate which type of efficiency property is compatible with these two criteria.

To address this question, we explore relaxations of \textsl{Pareto efficiency} by restricting the set of efficiency improvements that agents might employ. 
Since, in our model, implementing a full Pareto improvement is often impossible due to its complexity and the difficulty of coordinating intricate trades involving multiple objects among many agents, we focus on coordinatewise improvement and pairwise improvement: Because these types of improvements can be easily coordinated due to their relatively simple structures.
Consequently, we consider two weaker efficiency properties, \textsl{coordinatewise efficiency} and \textsl{pairwise efficiency}. 
First, we show that for separable preferences, \textsl{coordinatewise efficiency} is compatible with \textsl{individual rationality} and \textsl{strategy-proofness}.
More precisely, we show that these three properties uniquely identify one mechanism, the coordinatewise top-trading-cycles mechanism (cTTC), which is an extension of TTC (Theorem~\ref{thm:cTTC}).
However, for strict preferences, \textsl{coordinatewise efficiency} is also incompatible with \textsl{individual rationality} and \textsl{strategy-proofness} (Theorem~\ref{thm:imposs}).
We then turn to our second efficiency property, \textsl{pairwise efficiency}. We show that for both separable preferences and strict preferences, only another TTC extension,
the bundle top-trading-cycles mechanism (bTTC), 
satisfies \textsl{individual rationality}, \textsl{group strategy-proofness} (or the combination of \textsl{strategy-proofness} and \textsl{non-bossiness}), and \textsl{pairwise efficiency} (Theorems~\ref{thm:bttc} and \ref{thm:bttc2}). 
Finally, we propose several variations of our efficiency properties. We find that each of them is either satisfied by cTTC or bTTC, or leads to an impossibility result (together with \textsl{individual rationality} and \textsl{strategy-proofness}).
Therefore, our characterizations can be primarily interpreted as a compatibility test, where we determine whether certain efficiency properties are compatible with \textsl{individual rationality} and \textsl{strategy-proofness}. Loosely speaking, any ``reasonable'' efficiency property (defined by efficiency improvements) that is not satisfied by cTTC or bTTC is incompatible with \textsl{individual rationality} and \textsl{strategy-proofness}.\footnote{\label{myfootnote}It is worth noting that there might be a few restrictive efficiency properties compatible with these requirements that are not satisfied by cTTC or bTTC; however, such exceptions may not be particularly interesting to focus on. In Appendix~\ref{footnoteexample}, we provide an example of one exception and demonstrate that it is uninteresting. } 
Overall, our results can be divided into two main components: (1) the characterization results of our TTC extensions (cTTC and bTTC) and (2) the compatibility results between our efficiency properties with \textsl{individual rationality} and \textsl{strategy-proofness}.

In the discussion section (Section~\ref{discussion}), we first discuss the impact of our results (Subsection~\ref{discussion1}). Furthermore, in Subsection~\ref{discussion2} we look at general environments, such as multiple reallocation problems, in which objects are not categorized by different types, and agents' endowments may be unevenly distributed.
In such general environments, we still have some TTC extensions, e.g., segmented trading cycles \citep{papai2003} and generalized trading cycles \citep{tamura2021}. 
We discuss the relation between our characterization results with characterizations of these TTC extensions. 
Subsequently, in Subsection~\ref{discussion3}, we discuss the external validity of our efficiency properties and their compatibilities with \textsl{individual rationality} and \textsl{strategy-proofness} in these general environments.

\subsection{Our contributions}
Our results contribute to the following strands of literature.\medskip

\noindent\textbf{Objects allocation problems with multi-unit demands.}  
Assignments of scarce resources are attracting a lot of attention in economics design.\footnote{See the Nobel prize lectures in economic sciences in 2012 and 2020 for examples.}
Many studies, both theoretical and empirical, focus on indivisible resources, e.g., in auctions \citep{myerson1981,hortaccsu2010}, school choice \citep{abdulkadirouglu2003,kapor2020},  
and medical resource allocation \citep{pathak2021}. 
Existing literature has often investigated such indivisible resource allocation problems with ``unit-demand,'' in the sense that each agent only consumes one object.
Although in some important cases (e.g.,
kidney exchange, school choice, etc.), unit-demand is an appropriate assumption, in many other cases, agents may receive more than one object. Many studies analyze such situations \citep{papai2001,papai2007,anno2016,kazumura2020,manjunath2021,biro2022,biro2021,echenique2022}.
We contribute to this strand of literature by studying the multiple-type housing markets model, which allows agents to consume one object of each type. We believe that the analysis of multiple-type housing markets is relevant since this model is applicable to many problems in reality: a familiar example for most readers would be students' enrollments at many universities, where courses are taught in small groups and multiple sessions \citep{klaus2008}. Additionally, various other scenarios exist where individuals may wish to exchange their assigned resources, including term paper topics and dates during a course \citep{mackin2016}, staff, operating rooms, and dates in hospitals to improve surgery schedules \citep{huh2013}, and shifts in work settings due to personal reasons \citep{manjunath2021}. Furthermore, recent technological developments have enabled the allocation of several resource types together, such as cloud computing \citep{ghodsi2011,ghodsi2012} and 5G network slicing \citep{peng2015recent,bag2019,han2019}. Such situations can be modeled as multiple-type housing markets. Thus, the analysis of multiple-type housing markets may impact the real world.
Moreover, the analysis of multiple-type housing markets, as the first step, is potentially useful for addressing issues for other multi-unit demands models as mentioned above. \medskip

\noindent\textbf{Strategic robustness and efficiency.} 
In mechanism design, the characterization of strategically robust and efficient mechanisms is an important issue \citep{holmstrom1979,barbera1997,moulin1980,ma1994,pycia2017,shinozaki2023}. Standard properties of strategic robustness and efficiency are \textsl{strategy-proofness} and \textsl{Pareto efficiency}, respectively. However, when agents consume more than one object, the combination of \textsl{strategy-proofness} and \textsl{Pareto efficiency} essentially results in serially dictatorial mechanisms \citep{klaus2002,monte2015}. However, such mechanisms ignore property rights driven by the endowments, which violates \textsl{individual rationality}. In other words, in the presence of \textsl{individual rationality}, \textsl{strategy-proofness} and \textsl{Pareto efficiency} are incompatible. This impossibility leads to a common trade-off between strategic robustness and efficiency in the literature \citep{arrow1950,satterthwaite1975,myerson1983,Alva2020}. 
We want to keep \textsl{individual rationality} since it induces voluntary participation. That is, agents may lose interest in participating in a \textsl{individually irrational} mechanism. 
To ensure strategic robustness, it is important for the social planner to have knowledge of agents' true preferences in order to efficiently allocate resources. For example, if agents report some fake preferences, then, some allocations that are efficient with respect to reported preferences, may not be efficient for their true preferences. 
To avoid such situations, a mechanism that is both \textsl{individually rational} and \textsl{strategy-proof} should be used.\footnote{One may argue that perhaps we could follow another approach: maintain \textsl{individual rationality} and \textsl{Pareto efficiency}, while disregarding \textsl{strategy-proofness}. However, this approach also proves to be challenging. To illustrate this point, suppose that two agents want to efficiently reallocate~/~trade their endowments by themselves in a decentralized manner, rather than participating in a centralized mechanism. However, to obtain a better individual outcome, agents might have incentives to untruthfully share their private information (which represents their preferences, in our model). Consequently, after their negotiation, the resulting outcome may not be efficient with respect to their true preferences. In fact, it could be worse than the allocation they would have received if they had attended the mechanism.} 
Therefore, we search for a plausible mechanism by relaxing \textsl{Pareto efficiency} while keeping \textsl{individual rationality} and \textsl{strategy-proofness}.
There are various weakened efficiency properties that could be used, and some of them are satisfied by several mechanisms \citep{klaus2008,anno2016,feng2022}. However, \textsl{coordinatewise efficiency} and \textsl{pairwise efficiency}, are unique in that (i)  they are compatible with \textsl{individual rationality} and \textsl{strategy-proofness}, and (ii) they are only satisfied by two specific mechanisms, respectively. 
Moreover, by considering some other efficiency properties that are derived from
\textsl{coordinatewise efficiency} and \textsl{pairwise efficiency}, we discuss how to find a reasonable efficiency that is compatible with \textsl{individual rationality} and \textsl{strategy-proofness}.  
In this respect, our results are closely related to \citet{papai2007}, \citet{klaus2008}, \citet{anno2016}, \citet{nesterov2017}, \citet{manjunath2021}, \citet{shinozaki2022}, and \citet{shinozaki2022sp}.\medskip

\noindent\textbf{TTC based mechanisms.} The top-trading-cycles (TTC) algorithm (due to David Gale) is commonly used for object allocation problems with unit demand. In particular, as we mentioned earlier, for Shapley-Scarf housing markets with strict preferences, only TTC satisfies \textsl{individual rationality}, \textsl{strategy-proofness}, and \textsl{Pareto efficiency}. Thus, one could conjecture that for multiple-type housing markets, some extensions of TTC would still satisfy some desirable, although perhaps not all of the three properties. We confirm this conjecture by proving several characterizations of two TTC extensions: cTTC and bTTC. Our characterizations successfully extend characterizations of TTC from one dimensional Shapley-Scarf housing markets to higher dimensional multiple-type housing markets. 
Additionally, our methodology for obtaining these results is new and may be useful for other studies seeking more efficient ways to analyze higher-dimensional models.
Moreover, such characterizations give strong support for the use of TTC extensions. 
We provide a detailed discussion in Section~\ref{discussion}.
\medskip

\noindent\textbf{Complementary preferences.} As we mentioned earlier, our characterization for bTTC is also valid for strict preferences. 
The analysis on the domain of strict preference is demanding, because it allows agents' preferences to exhibit complementarity. Thus, our results also contribute to the literature on allocation problems with complements \citep{sun2006,che2019,rostek2020,jagadeesan2021,huang2023}.
\medskip

\noindent\textbf{The role of \textsl{non-bossiness}.}
A property of mechanisms that has played a significant role in the development of the axiomatic framework of allocation problems is the so-called \textsl{non-bossiness} \citep{satterthwaite1981}. 
It says that if a change in an agent's preferences does not affect his own allotment, it should not result in a change in anyone else's allotment. It is used to characterize various mechanisms in different contexts.\footnote{See \citet{thomson2016} for an excellent survey. } 
Mainly, \textsl{non-bossiness} is used as a complement to \textsl{strategy-proofness}, as their combination eqauls to \textsl{group strategy-proofness}. However, this is not the case in our model. By comparing our results with previous findings, we explicitly highlight the role of \textsl{non-bossiness} from both a normative standpoint and its mathematical significance in the proof.

\subsection{Organization}
Our paper is organized as follows. In the following section, Section~\ref{sec:model}, we introduce multiple-type housing markets, mechanisms and their properties, and two TTC extensions: bTTC and cTTC. 
We state our results in Section~\ref{sec:results}.

\begin{itemize}
    \item In subsection~\ref{subs:Coordinatewise efficiency}, we show that for separable preferences, a mechanism is \textsl{individually rational}, \textsl{strategy-proof}, and \textsl{coordinatewise efficient} if and only if it is cTTC (Theorem~\ref{thm:cTTC}).However, for strict preferences, these three properties are incompatible (Theorem~\ref{thm:imposs}). Additionally, to emphasize the novelty of our results, we provide two remarks (Remakrs~
\ref{remark:secondbest}~and~\ref{remark:ir}) to discuss the differences between ours and previous result (Result~\ref{result1}). 

    \item In subsection~\ref{subs:Pairwise efficiency}, we show that for separable preferences and for strict preferences, a mechanism is \textsl{individually rational}, \textsl{group strategy-proof}, and \textsl{pairwise efficient} if and only if it is bTTC (Theorems~\ref{thm:bttc} and \ref{thm:bttc2}). Furthermore, to distinguish our results from previous result (Result~\ref{result3}), we provide a remark (Remark~\ref{remark:nb}) to discuss why \textsl{non-bossiness} is also necessary for our characterization.

    \item In subsection~\ref{subs:other efficiency}, we consider other efficiency properties that are derived from \textsl{coordinatewise efficiency} and \textsl{pairwise efficiency}. By providing several characterizations and impossibilities (Theorems~\ref{thm:pce}, \ref{thm:coale}, and \ref{thm:mtpe}), we discuss how to obtain a (constrained) efficiency property that is compatible with \textsl{individual rationality} and \textsl{strategy-proofness}.
\end{itemize}

In Section~\ref{discussion}, we provide a discussion of our results and how they relate to the literature. Section~\ref{sec:conclusion} concludes.
In Appendix~\ref{section:Appendix1}, we provide the proofs of our results that are not included in the main text. In Appendix~\ref{section:Appendix2}, we provide several examples to establish the logical independence of the properties in our characterizations.

\section{Preliminaries}\label{sec:model}
\subsection{Multiple-type housing markets}
We consider a barter economy without monetary transfers formed by $n$ agents and $n\times m$ indivisible objects. Let $N=\{1,\ldots,n\}$ be a finite \textit{set of agents}, where $n\geq 2$. 
A nonempty subset of agents $S\subseteq N$ is a \textit{coalition}. There exist $m\geq 1 $ \textit{(distinct) types of indivisible objects} and $n$ \textit{(distinct) indivisible objects of each type}. We denote the \textit{set of object types} by $T=\{1,...,m\}$. For each $t\in T$, the set of type-$t$ objects is $O^t=\{o^t_i\}_{i\in N}$, and the \textit{set of all objects} is $O=\{O^t\}_{t\in T}$, where $|O|= n\times m$. Each agent owns exactly one object of each type. Without loss of generality, let  $o_i^t$ be agent $i$'s endowment of type-$t$. Thus, each agent $i$'s endowment is a list
$e_i=(o_i^1,\ldots,o_i^m)$. Moreover, each agent exactly consumes one object of each type, and hence, each agent's \textit{(feasible) consumption set} is $\Pi_{t\in T} O^t$. An element in $\Pi_{t\in T} O^t$ is a \textit{(consumption) bundle}. If $m=1$, then our model is the classical Shapley-Scarf housing market model \citep{shapley1974}. \medskip

An \textit{allocation $x$} partitions the set of all objects $O$ into $n$ bundles assigned to agents. Formally, $x=\{x_1,\ldots,x_n\}$ is such that for each $t\in T$, $\bigcup_{i\in N}x_i^t=O^t$ and for each pair $i\neq j$, $x_i^t\neq x_j^t$. The \textit{set of all allocations} is denoted by $X$, and the \textit{endowment allocation} is denoted by $e=\{e_1,\ldots,e_n\}$. Given an allocation $x\in X$, for each agent $i\in N$, we say that $x_i$ is agent $i$'s \textit{allotment} at $x$ and for each type $t\in T$, $x_i^t$ is agent $i$'s \textit{type-$t$ allotment} at $x$. For simplicity, sometimes we will restate an allocation as a list $x=(x_1,\ldots,x_n)\in (\Pi_{t\in T} O^t)^N$. Given $x$, let $x^t=(x_1^t,\ldots,x_n^t)$ be the allocation of type-$t$; $x_{-i}=(x_1,\ldots,x_{i-1},x_{i+1},\ldots,x_n)$ be the list of all agents' allotments, except for agent $i$'s allotment; and $x_S=(x_i)_{i\in S}$ to be the list of allotments of the members of coalition $S$.\medskip

Each agent $i$ has \textit{complete}, \textit{antisymmetric}, and \textit{transitive} \textit{preferences $R_i$} over all bundles (allotments), i.e., $R_i$ is a linear order over $\Pi_{t\in T}O^t$.\footnote{Preferences $R_i$ are \textit{complete} if for any two allotments $x_i,y_i$, $x_i\mathbin{R_i} y_i$ or $y_i\mathbin{R_i} x_i$; they are \textit{antisymmetric} if $x_i\mathbin{R_i} y_i$ and $y_i\mathbin{R_i} x_i$ imply $x_i=y_i$; and they are \textit{transitive} if for any three allotments $x_i,y_i,z_i$, $x_i\mathbin{R_i} y_i$ and $y_i\mathbin{R_i} z_i$ imply $x_i\mathbin{R_i} z_i$.} 
For two allotments $x_i$ and $y_i$, $x_i$ is \textit{weakly better than} $y_i$ if $x_i\mathbin{R_i} y_i$, and $x_i$ is \textit{strictly better than} $y_i$ if [$x_i \mathbin{R_i} y_i$ and not $y_i \mathbin{R_i} x_i$], denoted $x_i \mathbin{P_i} y_i$. Finally, since preferences over allotments are strict, $x_i$ is indifferent to $y_i$ only if $x_i=y_i$. We denote preferences as ordered lists, e.g., $R_i: x_i,\ y_i,\ z_i$ instead of $x_i\mathbin{P_i}y_i\mathbin{P_i}z_i$. The \textit{set of all preferences} is denoted by $\mathcal{R}$, which we will also refer to as the \textit{strict preference domain}.
Throughout the paper, we focus on strict preferences.
\medskip


A \textit{preference profile} is a list $R=(R_1,\ldots,R_n)\in\mathcal{R}^N$. We use the standard notation $R_{-i}=(R_1,\ldots,R_{i-1},R_{i+1},\ldots,R_n)$ to denote the list of all agents' preferences, except for agent $i$'s preferences. For each coalition $S\subseteq N$ we define $R_S=(R_i)_{i \in S}$ and $R_{-S}=(R_i)_{i\in N\setminus S}$ to be the lists of preferences of the members of $S$ and $N\setminus S$, respectively.\medskip

In addition to the domain of strict preferences, we consider a preference subdomains based on agents' ``marginal preferences'': assume that for each agent $i\in N$ and for each type $t\in T$, $i$ has complete, antisymmetric, and transitive preferences $R_i^t$ over the set of type-$t$ objects $O^t$. We refer to $R_i^t$ as \textit{agent $i$'s type-$t$ marginal preferences},  and denote by $\mathcal{R}^t$ the \textit{set of all type-$t$ marginal preferences}. Then, we can define the following preference domain.\medskip

\noindent\textbf{(Strictly) Separable preferences.}
Agent~$i$'s preferences $R_i\in\mathcal{R}$ are \textit{separable} if for each $t\in T$ there exist type-$t$ marginal preferences $R_i^t\in \mathcal{R}^t$ such that for any two allotments $x_i$ and $y_i$, $$\mbox{if for all } t\in T,\ x_i^t\mathbin{R_i^t} y_i^t,\mbox{ then }x_i \mathbin{R_i} y_i.$$
$\mathcal{R}_s$ denotes the \textit{domain of separable preferences}.
\medskip

We use the standard notation $R^t=(R^t_1,\ldots,R^t_n)$ to denote the list of all agents' marginal preferences of type-$t$, and $R^{-t}=(R^1,\ldots,R^{t-1},R^{t+1},\ldots,R^m)$ to denote the list of all agents' marginal preferences of all types except for type-$t$.\medskip

A \textit{(multiple-type housing) market} is a triple $(N,e,R)$.
When no confusion is possible about the set of agents $N$ and the endowment allocation $e$, we denote market $(N,e,R)$ by $R$.
Thus, the strict preference profile domain $\mathcal{R}^N$ also denotes the set of all markets with strict preferences. Similarly, $\mathcal{R}_s^N$ is also the set of all markets with separable preferences.

\subsection{Mechanisms and their properties}
Note that all following definitions for the set of markets with strict preferences can be formulated for markets with separable preferences. \medskip

A \textit{mechanism} is a function $f:\mathcal{R}^N\to X$ that selects for each market $R$ an allocation $f(R)\in X$, and 
\begin{itemize}
\item for each $i\in N$, $f_i(R)$ denotes agent $i$'s allotment
\item for each $i\in N$ and each $t\in T$, $f^t_i(R)$ denotes agent $i$'s type-$t$ allotment. Moreover, $f^t(R)$ denotes the allocation of type-$t$, i.e., $f^t(R)=(f^t_1(R),\ldots,f^t_n(R))$.
\end{itemize}
\medskip

We next introduce and discuss some well-known properties for allocations and mechanisms.  \medskip

First we consider a voluntary participation condition for an allocation $x$ to be implementable without causing agents any harm: no agent will be worse off than at his endowment.
Let $R\in \mathcal{R}^N$. An allocation $x\in X$ is \textit{individually rational} at $R$ if for each agent $i\in N$, $x_i \mathbin{R_i}e_i$.\medskip

\noindent \textbf{Individual rationality}: For each $R\in \mathcal{R}^N$, $f(R)$ is 
\textsl{individually rational} at $R$.\medskip

Next, we consider two well-known efficiency criteria. Let $R\in \mathcal{R}^N$. An allocation $y\in X$ is a \textit{Pareto improvement} over allocation $x\in X$ at $R$ if for each agent $i\in N$, $y_i \mathbin{R_i} x_i$, and for at least one agent $j\in N$, $y_j \mathbin{P_j} x_j$. An allocation is \textit{Pareto efficient} at $R$ if there is no \textsl{Pareto improvement}.\medskip 

\noindent \textbf{Pareto efficiency}: For each $R\in \mathcal{R}^N$, $f(R)$ is 
\textsl{Pareto efficient} at $R$.\medskip

An allocation $x\in X$ is \textit{unanimously best} at $R$ if for each agent $i\in N$ and each allocation $y\in X$, we have $x\mathbin{R_i} y$.\footnote{Since all preferences are strict, the set of unanimously best allocations is empty or single-valued.} \medskip

\noindent \textbf{Unanimity}: For each $R\in \mathcal{R}^N$, $f(R)$ is \textsl{unanimously best} whenever it exists.\medskip

If a \textsl{unanimously best} allocation exists at $R$, then that allocation is the only \textsl{Pareto efficient} allocation at $R$. Hence, \textsl{Pareto efficiency} implies \textsl{unanimity}.\medskip

The next two properties, \textsl{strategy-proofness} and \textsl{group strategy-proofness}, are two of the incentive properties that are most frequently used in the literature on mechanism design. They model that no agent~/~coalition can benefit from misrepresenting his~/~their preferences.	\medskip

\noindent \textbf{Strategy-proofness}: For each $R\in \mathcal{R}^N$, each $i\in N$, and each $R'_i\in \mathcal{R}$, $f_i(   R_i,R_{-i} ) \mathbin{R_i} f_i(   R'_i,R_{-i})$.\medskip

\noindent \textbf{Group strategy-proofness}: For each $R\in \mathcal{R}^N$, there are no coalition $S\subseteq N$ and no preference list $R'_S=(R'_i)_{i\in S}\in \mathcal{R}^S$ such that for each $i\in S$, $f_i(R'_S,R_{-S}) \mathbin{R_i} f_i(R)$, and for some $j\in S$, $f_j(R'_S,R_{-S}) \mathbin{P_j} f_j(R)$.\medskip

Next, we consider a well-known property for mechanisms that restricts each agent's influence: no agent can change other agents' allotments without changing his own allotment by changing his reported preference.\medskip

\noindent \textbf{Non-bossiness}: For each $R\in \mathcal{R}^N$, each $i\in N$, and each $R'_i\in \mathcal{R}$, $f_i(R_i,R_{-i}) =f_i(R'_i,R_{-i})$ implies $f(R_i,R_{-i}) =f(R'_i,R_{-i})$.
\medskip

By the definition, it is easy to verify that \textsl{group strategy-proofness} implies \textsl{strategy-proofness} and \textsl{non-bossiness}. However, the  converse is not true. To be more precise, (i) for strict preferences, \textsl{group strategy-proofness} coincides with \textsl{strategy-proofness} and \textsl{non-bossiness} \citep{alva2017}, but (ii) for separable preferences, the former is stronger than the latter \citep{feng2022}.\footnote{Also see Example~\ref{example:notPE} for details.}

\subsection{TTC extensions}\label{subsection:ttc}

We next focus on the domain of separable preferences $\mathcal{R}_s$ and extend Gale's famous top trading cycles (TTC) algorithm to multiple-type housing markets.

\noindent \textbf{The type-\textit{t} TTC algorithm} 

    Consider a market $(N,e,R)$ such that $R\in\mathcal{R}_s$. For each type $t\in T$, let $(N,e^t,R^t)= \newline
    (N,(o^t_1,\ldots,o_n^t),(R_1^t,\ldots,R_n^t))$ be its \textit{associated type-$t$ submarket}.\smallskip
    
    We define the top trading cycles (TTC) allocation for each type-$t$ submarket as follows.\smallskip
    
    \noindent \textbf{Input.} A type-$t$ submarket $(N,e^t,R^t)$.\smallskip

    \noindent\textbf{Step~$k(\geq 1)$.} Each agent points to his most-preferred remaining object given $R^t$. Each remaining object points to its owner. There exists at least one \textit{trading cycle}. \textit{Execute} all cycles by assigning each agent involved in a cycle the object to which he points. Remove all agents and objects involved in a cycle. If some agents or objects remain, then proceed to step~$k+1$. \smallskip

    \noindent \textbf{Output.} The type-$t$ TTC algorithm terminates when each agent in $N$ is assigned an object in $O^t$, which takes at most $n$ steps. We denote the object in $O^t$ that agent $i\in N$ obtains in the type-$t$ TTC algorithm by $TTC_i^t(e^t,R^t)$ and the final type-$t$ allocation by $TTC^t(e^t,R^t)$.

\noindent \textbf{The cTTC allocation~/~mechanism}: The \textit{coordinatewise top trading cycles (cTTC) allocation}, $cTTC(R)$, is the collection of all type-$t$ TTC allocations, i.e., for each $R\in \mathcal{R}_s^N$,
    \[cTTC(R)=\big(\left(TTC^1_1(R^1),\ldots, TTC^m_1(R^m)\right),\ldots,\left(TTC^1_n(R^1),\ldots, TTC^m_n(R^m)\right)\big).\]
    The $cTTC$ mechanism \citep[introduced by][]{wako2005} selects each market its cTTC allocation.\medskip

Next, we consider another TTC extension, which only allows agents to trade their endowments completely.

\noindent \textbf{The bundle top trading cycles (bTTC) algorithm / mechanism} 
    The \textit{bundle top trading cycles mechanism (bTTC)} assigns to each market $R$ the unique top-trading allocation that results from the TTC algorithm if agents are only allowed to trade their whole endowments among each other.
    
    Formally, for each market $R$ and $i\in N$, let $R_i|^e$ be the restriction\footnote{That is, for each $i\in N$, $R_i|^{e}$ are preferences over $\{e_1,\ldots,e_n\}$ such that for each $e_j,e_k\in \{e_1,\ldots,e_n\}$, $e_j \mathbin{R_i|^{e}}e_k  $ if and only if $e_j \mathbin{R_i} e_k$.} 
    of $R_i$ to endowments $\{e_1,\ldots,e_n\}$ and $R |^e\equiv (R_i| ^e)_{i\in N}$ be the restriction profile. 
    We then use the TTC algorithm to compute the bTTC allocation for $R|^e$. Note that the difference with the classical TTC algorithm (for Shapley-Scarf housing markets) is that instead of an object, each agent can only point to a whole endowment. The $bTTC$ mechanism assigns the bTTC allocation above to each market.

    \begin{remark} {\textbf{Comparison of cTTC and bTTC}}\ \\
\noindent i) Generalizability: bTTC is well-defined on the domain of strict preference profiles as well as separable preference profiles. Meanwhile, cTTC is only well-defined for separable preferences. Also note that for $m=1$, cTTC and bTTC reduce to TTC.\smallskip 

\noindent ii) Strategic robustness: one can verify that cTTC and bTTC inherit \textsl{strategy-proofness} and \textsl{non-bossiness} from the underlying top trading cycles algorithm. However, bTTC is also \textsl{group strategy-proof} \citep{fkkgsp} while cTTC is not \citep{feng2022}.\smallskip 

\noindent iii) Efficiency: it is easy to see that cTTC is more flexible than bTTC in terms of trading, however, the example below shows that neither of them can Pareto improve the other.
\hfill~$\diamond$ \label{remark:domain}
    \end{remark}

    \begin{example}
        Consider a market with two agents and two types, i.e., $N=\{1,2\}$, $T=\{H(ouse),C(ar)\}$, $O=\{H_1,H_2,C_1,C_2\}$, and for each $i\in N$, $e_i=(H_i,C_i)$. $R\in \mathcal{R}_s^N$ is as follows:
       $$R_1^H:H_2,\bm{H_1}; R_1^C:\bm{C_1},C_2;\text{ and }R_2^H:\bm{H_2},H_1; R_2^C:C_1,\bm{C_2},$$ 
       $$R_1^e:e_2,\bm{e_1};\text{ and }R_2^e:e_1,\bm{e_2}.$$
       
       Thus, agent $1$ would like to trade houses but not cars and agent $2$ would like to trade cars but not houses. One easily verifies that $cTTC(R)=e=(e_1,e_2)=( (H_1,C_1),(H_2,C_2) )$ and $bTTC(R)=x=(e_2,e_1)=( (H_2,C_2),(H_1,C_1) )$. Since $e_2\mathbin{R_1}e_1$ and $e_1\mathbin{R_2}e_2$, both agents prefer the $bTTC$ allocation to the $cTTC$ allocation at $R$.
Furthermore, consider $\hat{R}\in \mathcal{R}_s^N$ as follows: 
       $$\hat{R}_1^H:\bm{H_1},H_2;\hat{R}_1^C:C_2,\bm{C_1};\text{ and }\hat{R}_2^H:\bm{H_2},H_1; \hat{R}_2^C:C_1,\bm{C_2},$$ 
       $$\hat{R}_1^e:\bm{e_1},e_2;\text{ and }\hat{R}_2^e:\bm{e_2},e_1.$$
       
       Thus, both agents would like to trade cars but not houses.
       One easily verifies that $cTTC(\hat{R})=y=(  (H_1,C_2),(H_2,C_1)  )$, $bTTC(\hat{R})=e$, and both agents prefer the $cTTC$ allocation to the $bTTC$ allocation at $\hat{R}$.
       
Finally, we additionally show that cTTC is not \textsl{group strategy-proof}.
Assume that both agents misreport their preferences as follows:
       $$\bar{R}_1^H:H_2,\bm{H_1}; \bar{R}_1^C:C_2,\bm{C_1};\text{ and }\bar{R}_2^H:H_1,\bm{H_2}; \bar{R}_2^C:C_1,\bm{C_2}.$$ 

       Then, $cTTC(\bar{R})=x$, making both agents better off compared to $cTTC(R)=e$. It means that agents $1$ and $2$ have incentives to jointly misreport $\bar{R}$ at $R$.
       \hfill~$\diamond$ \label{example:notPE}
       \end{example}

\section{Results}\label{sec:results}
As mentioned before, for $|T|=m=1$ our model is the classical Shapley-Scarf housing markets model. The Shapley-Scarf housing markets (with strict preferences) results that are pertinent for our analysis of multiple-type housing markets are the following.
\begin{result}[A positive result]\ \\
For Shapley-Scarf housing markets, only TTC is \textsl{individually rational}, \textsl{strategy-proof}, and \textsl{Pareto efficient}  \citep{ma1994,svensson1999}. \label{result1}
\end{result}

However, Result~\ref{result1} is not valid for multiple-type housing markets.

\begin{result}[Two impossibilities]\ \\
    For multiple-type housing markets 
\begin{itemize}
    \item[(a)] with separable preferences, no mechanism is \textsl{individually rational}, \textsl{strategy-proof}, and \newline\textsl{Pareto efficient} \citep{konishi2001}.
    \item[(b)] with strict preferences, no mechanism is \textsl{individually rational}, \textsl{strategy-proof}, and\newline \textsl{unanimous} \citep{feng2022}.
\end{itemize}
\label{result2}
\end{result}

Given this incompatibility, we explore relaxations of \textsl{Pareto efficiency} by restricting the set of efficiency improvements that agents might employ. In particular, we focus on ``coordinatewise improvement'' and ``pairwise improvement.''

The motivation for considering such restrictions on improvements is as follows. In object (re)allocation problems, if an allocation is not \textsl{Pareto efficient}, it admits an efficiency improvement. Therefore, after its implementation, agents could trade their allotments and destabilize the allocation by executing that improvement. However, when agents own more than one object, such destabilizations are generally complex and difficult to coordinate, potentially involving intricate trades of multiple objects among many agents. Indeed, doing so is computationally hard even when agents' preferences are additive (a subdomain of separable preferences) \citep{de2009complexity}. Since computational complexity is a suitable proxy for the plausibility of agents' behavior, these negative results suggest that full \textsl{Pareto efficiency} is unnecessarily strong.

Therefore, we consider two efficiency properties that are substantially weaker than \textsl{Pareto efficiency}. These properties rule out destabilizing efficiency improvements that can be easily coordinated due to their relatively simple structures.


\subsection{Coordinatewise efficiency} \label{subs:Coordinatewise efficiency}

Here, borrowing the clever idea from \citet{biro2022},
we consider a natural modification of \textsl{Pareto efficiency} for multiple-type housing markets, \textsl{coordinatewise efficiency}, which rules out Pareto-improvements upon the allocation by executing some ``single-object exchange.'' That is, an exchange involving a group of agents each of whom relinquishes one object and receives one object. 
Moreover, note that since in our model, each agent can receive only one object of each type, this restriction on Pareto improvements equivalently rules out improvements within a single type.

Let $R\in \mathcal{R}^N$. An allocation $y\in X$ is a \textit{coordinatewise improvement} of allocation $x\in X$ at $R$ if (i) $y$ is a \textsl{Pareto improvement} of $x$, and (ii) $y$ and $x$ only differ in one type $t\in T$, i.e., $y^t\neq x^t$ and for each $\tau\in T\setminus\{ t\}$, $y^\tau=x^\tau$. An allocation is \textit{coordinatewise efficient} at $R$ if there is no \textsl{coordinatewise improvement}. \medskip

\noindent \textbf{Coordinatewise efficiency}: For each $R\in \mathcal{R}^N$, $f(R)$ is \textsl{coordinatewise efficient} at $R$.\medskip

One easily verifies that \textsl{Pareto efficiency} implies \textsl{coordinatewise efficiency}.


\begin{remark}{\textbf{Coordinatewise efficiency for separable preferences and strict preferences}}\ \\
 For a multiple-type housing market with separable preferences, \textsl{coordinatewise efficiency} simply means that the selected allocation of each type is \textsl{Pareto-efficient} for agents' marginal preferences for the type. Formally, $f:\mathcal{R}^N_s\to X$ is \textit{coordinatewise efficient} if for each $R\in \mathcal{R}^N_s$ and each $t\in T$, $f^t(R)$ is \textsl{Pareto efficient} at $R^t$. Also, one easily verifies that for multiple-type housing markets with separable preferences, \textsl{coordinatewise efficiency} implies \textsl{unanimity}; however, it is not true for strict preferences.\footnote{We thank an anonymous referee for pointing this out, and the proof is in Appendix~\ref{footnoteceexample}.\label{footnotece}}

 Moreover, since TTC is \textsl{Pareto efficient} for Shapley-Scarf housing markets, it is easy to see that for each $R\in \mathcal{R}^N_s$ and each $t\in T$, $cTTC^t(R)$ is \textsl{Pareto efficient} at $R^t$, and hence cTTC is \textsl{coordinatewise efficient}. \hfill~$\diamond$ \label{remark:cef}
\end{remark}

We first characterize cTTC for separable preferences with respect to \textsl{coordinatewise efficiency}.

\begin{theorem}
    For multiple-type housing markets with separable preferences, only cTTC satisfies 
    \begin{itemize}
        \item \textsl{individual rationality},
        \item \textsl{strategy-proofness}, and
        \item \textsl{coordinatewise efficiency}. 
    \end{itemize}
 \label{thm:cTTC}
\end{theorem}

We prove Theorem~\ref{thm:cTTC} in Appendix~\ref{appendix:cTTC}. It is known that cTTC satisfies (i) \textsl{individual rationality} and \textsl{strategy-proofness} \citep{feng2022}, and (ii) \textsl{coordinatewise efficiency} (see Remark~\ref{remark:cef}). For uniqueness, the proof consists of two steps. 
Let us consider a mechanism that satisfies all three properties mentioned above.
We first consider a restricted domain of preference profiles, where all agents exhibit lexicographic preferences, favoring type-$1$ over type-$2$, type-$2$ over type-$3$, and so forth. 
We show that on the restricted domain, this mechanism always selects the cTTC allocation (Proposition~\ref{proposition:cttc1}).
Then by replacing agents' preferences, one by one, from our restricted domain to separable preference domain, we extend this result to the domain of separable preference profiles. We call this the \textit{preference replacement approach}.

The preference replacement approach is a commonly used method in the existing literature, e.g., \citet{papai2001, papai2003}, and \citet{feng2022}. However, our approach here significantly differs from theirs. In their papers, \textsl{non-bossiness} is a key property, which simplifies the replacement process because if we replace an agent's preference and his allotment remains unchanged, then the entire allocation remains unchanged as well.
However, in the proof of Theorem~\ref{thm:cTTC}, we do not have this specific invariance since we do not have \textsl{non-bossiness}. Instead, every time when we replace one agent's preference, we meticulously verify that the entire allocation remains unchanged. This verification process adds considerable complexity to the proof, making our preference replacement approach novel. We believe that in future studies, working with a such an approach can also be useful while studying other higher dimensional models.

The absence of \textsl{non-bossiness} introduces another challenge compared to \citet{feng2022}. Specifically, a key property used in their proof, \textsl{marginal individual rationality}, cannot be implied from the three properties we have.\footnote{A mechanism is \textit{marginally individually rational} if for each $R$, each $i$, and each $t$, $f_i^t(R)\mathbin{R_i^t} o_i^t$.  } As a result, we cannot simply decompose a multiple-type market $(N,e,R)$ into coordinate-wise submarkets $(N,e^t,R^t)_{t\in T}$ and apply Ma's result to obtain our cTTC characterization. This further highlights (1) the novelty of our new proof approach; and (2) the differences between \citet{feng2022}'s result and ours.
\medskip

We would like to make two additional remarks to emphasize the significance of our result.

\begin{remark}{\textbf{Second-best incentive compatibility} }\ \\
When considering various types of markets simultaneously, it is common to observe the independent operation of submarkets in real-life scenarios  \citep{anno2016}. In some cases, implementing separate ``locally optimal'' mechanisms can even result in a ``globally optimal'' outcome \citep{ashlagi2023}. However, this does not hold true for multiple-type housing markets. Therefore, there is a need to find support for the independent operation of markets, such as cTTC, which poses a significant challenge.

For multiple-type housing markets with separable preferences, \citet{klaus2008} weakens \textsl{Pareto efficiency} to another efficiency property, \textsl{second-best incentive compatibility}.\footnote{A \textsl{strategy-proof} mechanism $f:\mathcal{R}^N_s\to X$ is \textit{second-best incentive compatible} if there does not exist another \textsl{strategy-proof} mechanism $g:\mathcal{R}^N_s\to X$ such that (i) for each $R\in \mathcal{R}^N_s$ and each $i\in N$, $g_i(R)\mathbin{R_i}f_i(R)$, and (ii) for some $R'\in \mathcal{R}^N_s$ and  $j\in N$, $g_j(R')\mathbin{P'_j}f_j(R')$.}
She shows that cTTC is \textsl{second-best incentive compatible}, i.e., there exists no other strategy-proof mechanism that Pareto dominates cTTC. However, she also shows that there exists another mechanism that is \textsl{individually rational}, \textsl{strategy-proofness}, and \textsl{second-best incentive compatible}. In a follow-up work, \citet{anno2016} investigate \textsl{second-best incentive compatibility} for \textsl{independent} mechanisms, which treats each submarket independently and separately. In other words, under an \textsl{independent} mechanism, the selected allocation of each type only depends on agents' marginal preferences for each type. 
    They also show that cTTC is not the unique \textsl{independent} mechanism that satisfies these properties. Thus, Theorem~\ref{thm:cTTC} complements \citet{klaus2008} and \citet{anno2016}: by strengthening \textsl{second-best incentive compatibility} to \textsl{coordinatewise efficiency}, we find that cTTC is the only mechanism that satisfies \textsl{individual rationality}, \textsl{strategy-proofness}, and \textsl{coordinatewise efficiency}. 
    In other words, Theorem~\ref{thm:cTTC} demonstrates that cTTC is not only good, but it is good enough: no other mechanism can outperform it if the three properties are indispensable.
\hfill~$\diamond$   \label{remark:secondbest}
\end{remark}

    \begin{remark}{\textbf{Individual rationality}}\ \\
        Although one might view Theorem~\ref{thm:cTTC} as a trivial extension of Result~\ref{result1}, we want to stress that our finding actually adds novelty to the field. In particular, a major challenge with multiple-type housing markets is that \textsl{individual rationality} is weakened considerably. For instance, when agents lexicographically prefer one type over others, if an agent receives a better object than his endowment for his most important object type, then \textsl{individual rationality} of the allotment is respected even if we ignore the endowments of the other object types.
        Let us consider a small domain where there are two types of objects, houses and cars, and all agents lexicographically prefer houses over cars.\footnote{Lexicographic preferences are formally defined in Appendix~\ref{appendix:unclear}.}
        In this domain, an alternative mechanism exists that differs from cTTC and still satisfies \textsl{individual rationality}, \textsl{coordinatewise efficiency}, and \textsl{strategy-proofness}:
        
        Step~1: First apply TTC to houses. 
        If agent $1$ received a new house at Step~1 (and hence improves upon his own house), then move his endowed car to the bottom of his marginal preferences for cars (in terms of cTTC, agent $1$ is not allowed to point at his own car until the end of the algorithm).
        
        Step~2: Apply TTC to cars with the adjusted preferences.\hfill~$\diamond$   \label{remark:ir}  
        \end{remark}

Since cTTC is not well-defined for strict preferences, a natural question is whether there exists an extension of cTTC for strict preferences that satisfies our desired properties. Our answer is no. 

\begin{theorem}
    For multiple-type housing markets with strict preferences, no mechanism satisfies 
    \begin{itemize}
        \item \textsl{individual rationality},
        \item \textsl{strategy-proofness}, and
        \item \textsl{coordinatewise efficiency}. 
    \end{itemize}
    \label{thm:imposs}
\end{theorem}
\begin{proof}[\textbf{Proof}]
For simplicity, consider a market with two agents and two types, i.e., $N=\{1,2\}$, $T=\{H(ouse),C(ar)\}$, $O=\{H_1,H_2,C_1,C_2\}$, and for each $i\in N$, $e_i=(H_i,C_i)$. The preference profile $R\in \mathcal{R}^N$ is as follows:
$$R_1:(H_1,C_2),(H_2,C_1),\bm{(H_1,C_1)},(H_2,C_2)\text{; and }R_2:(H_1,C_2),(H_2,C_1),\bm{(H_2,C_2)},(H_1,C_1).$$

Let $x\equiv ( (H_1,C_2),(H_2,C_1) )$ and $y\equiv ((H_2,C_1), (H_1,C_2) )$.
Suppose that there is a mechanism $f:\mathcal{R}^N \to X $ that is \textsl{individually rational}, \textsl{strategy-proof}, and \textsl{coordinatewise efficient}. 

By \textsl{coordinatewise efficiency}, $f(R)\in \{x,y\}$. 
If $f(R)=x$. Let $R'_2: (H_1,C_2),\bm{(H_2,C_2)},\ldots $ By \textsl{individual rationality} and \textsl{coordinatewise efficiency}, $f(R_1,R'_2)=y$, which implies that agent $2$ has an incentive to misreport $R'_2$ at $R$; contradicting of \textsl{strategy-proofness}.
If $f(R)=y$. Let $R'_1: (H_1,C_2),\bm{(H_1,C_1)},\ldots $ By \textsl{individual rationality} and \textsl{coordinatewise efficiency}, $f(R'_1,R_2)=x$, which implies that agent $1$ has an incentive to misreport $R'_1$ at $R$; contradicting of \textsl{strategy-proofness}. 
Extending this example to include more than two agents or two types is straightforward, so we omit it.
\end{proof}

We would like to emphasize that the example we used in the proof of Theorem~\ref{thm:imposs} provides valuable insights. When agents have multi-unit demands and we consider a rich domain of preferences, e.g., rich domain in \citet{dasgupta1979}, it becomes apparent that we can easily construct scenarios in which two agents have a conflict of interests. In the presence of \textsl{individual rationality}, this conflict of interest leads to a trade-off between \textsl{strategy-proofness} and many efficiency notions, including \textsl{coordinatewise efficiency} and \textsl{Pareto efficiency}. Furthermore, it is worth noting that this impossibility result can be easily extended to other models that allow multi-unit demands.\medskip

Theorem~\ref{thm:imposs} also leads to a new question: whether there is an efficiency property that is compatible with \textsl{individual rationality} and \textsl{strategy-proofness} for strict preferences. We will address this question in the next subsection.

\subsection{Pairwise efficiency}\label{subs:Pairwise efficiency}
\textsl{Unanimity} is a weak efficiency property. However, for strict preferences, it is still incompatible with \textsl{individual rationality} and \textsl{strategy-proofness}
(Result~\ref{result2}~(b)). 
Therefore, for strict preferences, it seems difficult to find an efficiency property that is compatible with \textsl{individual rationality} and \textsl{strategy-proofness}.
To establish a suitable efficiency property, 
we consider efficiency improvements that only involve a small number of agents \citep{Goldman1982}. To be precise, here we consider \textsl{pairwise efficiency} that rules out efficiency improvements by pairwise reallocation \citep{ekici2022}. 
Let $R\in \mathcal{R}^N$. An allocation $x\in X$ is \textit{pairwise efficient} at $R$ if there is no pair of agents $\{i,j\}\subseteq N$ such that $x_j\mathbin{P_i}x_i$ and $x_i\mathbin{P_j}x_j$. \medskip

\noindent \textbf{Pairwise efficiency}: For each $R\in \mathcal{R}^N$, $f(R)$ is \textsl{pairwise efficient} at $R$.\medskip

For Shapley-Scarf housing markets, the result related to \textsl{pairwise efficiency}, that is pertinent for our analysis of multiple-type housing markets is the following.

\begin{result}{\citep{ekici2022}}\ \\
    For Shapley-Scarf housing markets, only TTC is \textsl{individually rational}, \textsl{strategy-proof}, and \textsl{pairwise efficient}. \label{result3}
    \end{result}

In our context, \textsl{pairwise efficiency} may seem strong at first glance. However, it may be suitable for some real-world applications. Consider refugee resettlement as an example. In this case, there are physical constraints that are given a priori: families cannot be assigned to one locality while working in another place \citep{delacretaz2023}.
In such scenarios, it may be appropriate to impose constraints involving all object-types together, making \textsl{pairwise efficiency} a suitable concept to consider.\medskip

By using arguments similar to arguments in Result~\ref{result3}, we also obtain that bTTC inherits \textsl{pairwise efficiency} from the underlying top trading cycles algorithm for the restricted market $R|^e$. Also, it is known that bTTC is \textsl{individually rational} and \textsl{strategy-proof} \citep{fkkgsp}. 
Hence, based on Result~\ref{result3}, one could now conjecture that for
multiple-type housing markets, bTTC is identified by these three properties. That conjecture is nearly correct, but to fully support it, we need to strengthen \textsl{strategy-proofness} to \textsl{group strategy-proofness} (or the combination of \textsl{strategy-proofness} and \textsl{non-bossiness}: recall that for strict preferences, \textsl{group strategy-proofness} coincides with the combination of \textsl{strategy-proofness} and \textsl{non-bossiness}).

\begin{theorem}
    For multiple-type housing markets with strict preferences, only bTTC satisfies 
    \begin{itemize}
        \item \textsl{individual rationality},
        \item \textsl{group strategy-proofness} (or \textsl{strategy-proofness} and \textsl{non-bossiness}), and
        \item \textsl{pairwise efficiency}. 
    \end{itemize}
 \label{thm:bttc}
\end{theorem}

\begin{remark}{\textbf{Non-bossiness}}\ \\
 It is surprising that in our characterization of bTTC, \textsl{non-bossiness} is also involved, in comparison with Result~\ref{result3}. In other words, \textsl{non-bossiness} cannot be redundant when we have more than one type.
 Why is this the case? For Shapley-Scarf housing markets, each agent only demands one object. Thus, each agent will trade within only one coalition. Therefore, \textsl{non-bossiness} is redundant since each agent can only influence the coalition in which he is involved. However, the same is not true for multiple-type housing markets because each agent may trade different objects with different coalitions. To see this point, refer to Example~\ref{example:notNB} in Appendix~\ref{section:Appendix2}.\hfill~$\diamond$  \label{remark:nb}   
\end{remark}

Quite interestingly, this characterization is also valid for separable preferences, even if we weaken \textsl{group strategy-proofness} to the combination of \textsl{strategy-proofness} and \textsl{non-bossiness} (recall that for separable preferences, \textsl{group strategy-proofness} is stronger than the combination of \textsl{strategy-proofness} and \textsl{non-bossiness}). This result is noteworthy as it contradicts the common intuition that ``mechanisms operating in such domains (separable preferences) are typically inefficient and not weak group strategy-proof,'' as noted by \citet{barbera2016}.

\begin{theorem}
    For multiple-type housing markets with separable preferences, only bTTC satisfies 
    \begin{itemize}
        \item \textsl{individual rationality},
        \item \textsl{strategy-proofness},
        \item \textsl{non-bossiness}, and
        \item \textsl{pairwise efficiency}. 
    \end{itemize}
 \label{thm:bttc2}
\end{theorem}


We prove Theorems~\ref{thm:bttc} and \ref{thm:bttc2} in Appendix~\ref{appendix:bttc}. Here we only explain the intuition of the uniqueness part of the proof. We first consider a restricted domain of preference profiles. On this restricted domain, consider a top trading cycle that forms at the first step of bTTC. We first show that, by \textsl{individual rationality}, \textsl{strategy-proofness}, \textsl{non-bossiness}, and \textsl{pairwise efficiency}, agents in this top trading cycle receive their bTTC allotments. We can then consider agents who trade at the second step of bTTC by following the same arguments as for first step trading cycles, and so on. Thus, we find that on this restricted domain, only bTTC satisfies \textsl{individual rationality}, \textsl{strategy-proofness}, \textsl{non-bossiness}, and \textsl{pairwise efficiency} (Theorem~\ref{thm:bttc3}). Then, following a similar approach as in the proof of Theorem~\ref{thm:cTTC}, we extend this result to the domain of separable preference profiles and strict preference profiles.\smallskip

The extension approach we used in Theorem~\ref{thm:bttc2} also has independent methodological interest. In particular, by this approach we confirm that the characterization of bTTC is also true for any domain between separable and strict preferences.

\begin{corollary}
    For multiple-type housing markets with any preference domain $\hat{\mathcal{R}}$ such that $ \mathcal{R}_s\subseteq \hat{\mathcal{R}}\subseteq \mathcal{R}$, only bTTC satisfies \textsl{individual rationality}, \textsl{group strategy-proofness} (or \textsl{strategy-proofness} and \textsl{non-bossiness}), and \textsl{pairwise efficiency}.   
\end{corollary}

We provide three additional remarks to facilitate the reader's understanding.

\begin{remark}{\textbf{Interpretation of Theorems~\ref{thm:bttc} and \ref{thm:bttc2}}}\ \\
How to interpret our characterization of bTTC? There are three ways to explain it.
First, as a positive result, Theorems~\ref{thm:bttc} and \ref{thm:bttc2} demonstrate that bTTC is identified by a list of properties. Consequently, the social planner should select bTTC if he cares about these properties.
Second, Theorems~\ref{thm:bttc} and \ref{thm:bttc2} reveal the trade-off between efficiency and strategic robustness (\textsl{group strategy-proofness}) in the presence of \textsl{individual rationality}: if the social planner wishes to achieve stronger efficiency, he needs to weaken \textsl{group strategy-proofness}.
Third, Theorems~\ref{thm:bttc} and \ref{thm:bttc2} suggest that bTTC can be used as a benchmark for the reallocation of multiple-type housing markets, in the sense that no mechanism should perform worse than bTTC. \hfill~$\diamond$  \label{remark:interpretation}
    \end{remark}

\begin{remark}{\textbf{Independence of Theorems~\ref{thm:bttc} and \ref{thm:bttc2}}}\ \\
Theorem~\ref{thm:bttc} is not a more general result or a trivial extension of Theorem~\ref{thm:bttc2}, and Theorem~\ref{thm:bttc2} is not a direct implication of Theorem~\ref{thm:bttc}. There is a logical independence between proving a characterization on different domains.\footnote{For a detailed discussion of the role of domains in characterization, see \citet[Section 11.3]{thomson2022}.  } 
On the one hand, we may have a characterization on some domain but not on a subdomain. For instance, for Shapley-Scarf housing markets, the characterization of TTC for strict preferences (Result~\ref{result1}) is not necessarily valid on some subdomain \citep{bade2019}. 
Conversely, a mechanism satisfying a given list of properties may not exist on a larger domain, even if it does exist on a smaller domain. For instance, while our characterization of cTTC (Theorem~\ref{thm:cTTC}) is valid on the domain of separable preference profiles, it may not hold true in certain superdomains (Theorem~\ref{thm:imposs}).  \label{remark:indep}
    \hfill~$\diamond$  
    \end{remark}

\begin{remark}{\textbf{Constraints and efficiency}}\ \\
    Trading constraints frequently occur in reality \citep{shinozaki2022}. 
    On the one hand, constraints may exclude some desirable outcomes, on the other hand, they may help us to guarantee positive results \citep{madhav2015}.
    For instance, to ensure the existence of the core, 
    \citet{kalai1978} impose restrictions on trades among certain agents and \citet{papai2007} restricts the set of feasible trades. 
    However, this also raises a new question for the mechanism designer: which constraint should be enforced? In other words, if any constraint is admissible, which constraint is the most suitable? 
   Theorems~\ref{thm:bttc} and \ref{thm:bttc2} partially answer this question: if we still want to achieve some efficiency, then, without loss of other properties, allowing agents to trade their endowments completely is sufficient and necessary to achieve \textsl{pairwise efficiency}.\hfill~$\diamond$    \label{remark:coneff}
\end{remark}

    Moreover, based on our results (Theorems~\ref{thm:cTTC},~\ref{thm:bttc}, and \ref{thm:bttc2}), we also have the following observation, which essentially shows a trade-off between our two efficiency properties in the presence of \textsl{individual rationality} and \textsl{strategy-proofness}. 

\begin{observation}
    For multiple-type housing markets, even with separable preferences, an \textsl{individually rational} and \textsl{strategy-proof} mechanism cannot satisfy both \textsl{coordinatewise efficiency} and \textsl{pairwise efficiency}.\label{observation}
\end{observation}

\subsection{Other efficiency properties}
\label{subs:other efficiency}
We will now discuss other efficiency properties that are derived from \textsl{coordinatewise efficiency} and \textsl{pairwise efficiency}.\smallskip

First, we consider a weaker version of \textsl{coordinatewise efficiency} that only involves two agents.\smallskip

Let $R\in \mathcal{R}^N$. An allocation $y\in X$ is a \textit{pairwise coordinatewise improvement} of allocation $x\in X$ at $R$ if (i) $y$ is a \textsl{coordinatewise improvement} of $x$, and (ii) $y$ and $x$ only differ in two agents $i,j\in N$ at one type $t\in T$, i.e., for some $t\in T$ and for some distinct $i,j\in N$,
 $y_i^t=x_j^t$, $y_j^t=x_i^t$, for each $k\in N\setminus\{i,j\}$, $y_k^t=x_k^t$, and for each $\tau\in T\setminus\{ t\}$, $y^\tau=x^\tau$. An allocation is \textit{pairwise coordinatewise efficient} at $R$ if there is no \textsl{pairwise coordinatewise improvement}. \medskip

\noindent \textbf{Pairwise coordinatewise efficiency}: For each $R\in \mathcal{R}^N$, $f(R)$ is \textsl{pairwise coordinatewise efficient} at $R$.\medskip

Our results in Theorems~\ref{thm:cTTC} and \ref{thm:imposs}
are still true if we replace \textsl{coordinatewise efficiency} with \textsl{pairwise coordinatewise efficiency}.\medskip 

\begin{theorem}
    For multiple-type housing markets 
    \begin{itemize}
        \item with separable preferences, only cTTC satisfies \textsl{individual rationality}, \textsl{strategy-proofness}, and \textsl{pairwise coordinatewise efficiency}. 
        \item with strict preferences, no mechanism satisfies \textsl{individual rationality}, \textsl{strategy-proofness}, and \textsl{pairwise coordinatewise efficiency}. \label{thm:pce}
    \end{itemize}
\end{theorem}
The proofs are the same as in Theorems~\ref{thm:cTTC} and \ref{thm:imposs} and hence we omit them. Note that Theorem~\ref{thm:pce} implies that \textsl{pairwise coordinatewise efficiency} and \textsl{pairwise efficiency} are logically independent: cTTC satisfies the former but violates the latter, and bTTC satisfies the latter but violates the former.\medskip

Second, we consider a stronger version of \textsl{pairwise efficiency} that involves larger coalitions: \textsl{coalitional efficiency} \citep{tierney2022}.\footnote{\citet{tierney2022} originally refers to it as \textsl{conditional optimality}.} 
This property says that the selected allocation cannot be improved by the reallocation of allotments, keeping bundled allotments intact.\smallskip

Let $R\in \mathcal{R}^N$. An allocation $x\in X$ is \textit{coalitionally efficient} at $R$ if there is no coalition $S\equiv\{i_1,i_2,\ldots,i_K\}\subseteq N$ such that for each $i_\ell\in S$, $x_{i_\ell}\mathbin{P_{i_\ell}}x_{i_{\ell+1}}$ (mod $K$).\footnote{For Shapley-Scarf housing markets with strict preferences, \textsl{coalitional efficiency} is equivalent to \textsl{Pareto efficiency}.}  \medskip

\noindent \textbf{Coalitional efficiency}: For each $R\in \mathcal{R}^N$, $f(R)$ is \textsl{coalitionally efficient} at $R$.\medskip

It is easy to verify that bTTC satisfies \textsl{coalitional efficiency} from the underlying TTC algorithm for the restricted market $R|^e$. Thus, our characterization of bTTC is still valid if we replace \textsl{pairwise efficiency} with \textsl{coalitional efficiency}. \medskip 

\begin{theorem}
    For multiple-type housing markets (i) with separable preferences and (ii) with strict preferences, only bTTC satisfies \textsl{individual rationality}, \textsl{group strategy-proofness} (or the combination of \textsl{strategy-proofness} and \textsl{non-bossiness}), and \textsl{coalitional efficiency}. \label{thm:coale}
\end{theorem}

All the efficiency properties discussed above have certain constraints on efficiency improvements. For example, \textsl{coordinatewise efficiency} and \textsl{pairwise coordinatewise efficiency} only consider efficiency improvements within one type ($\{o_1^t,\ldots,o_n^t\}$), while \textsl{pairwise efficiency} and \textsl{coalitional efficiency} only consider efficiency improvements within the full endowments of a coalition ($S\subseteq N$ with $\{e_i\}_{i\in S}$). A natural question is whether it is possible to consider something in between, such as efficiency improvements for more than one type but less than all types.

Let $R\in \mathcal{R}^N$. An allocation $x\in X$ is \textit{$T'$-types pairwise efficient} at $R$ if there is no pair of agents $ \{i,j\}\subseteq N$ and a strict subset of types $T'\subsetneq T$ such that $y_i\mathbin{P_i}x_i$ and $y_j\mathbin{P_j}x_j$, where $y_i=( (x_j^t)_{t\in T'},   (x_i^t)_{t\in -T' }  )$ and $y_j=( (x_i^t)_{t\in T'},   (x_j^t)_{t\in -T' }   )$.
\medskip

\noindent \textbf{$T'$-types pairwise efficiency}: For each $R\in \mathcal{R}^N$, $f(R)$ is \textsl{$T'$-types pairwise efficient}.
\begin{remark}{\textbf{Restriction on $|T'|$}}\ \\
By the definition of \textsl{$T'$-types pairwise efficiency}, it is easy to see that \textsl{$T'$-types pairwise efficiency} is stronger than \textsl{pairwise coordinatewise efficiency}. If we do not assume that $T'\subsetneq T$, i.e., $|T'|=m$ is also possible, then this new property is also stronger than \textsl{pairwise efficiency}. By Observation~\ref{observation}, we know that no \textsl{individually rational} and \textsl{strategy-proof} mechanism satisfies it. 

Given $T'\subsetneq T$, if there are only two types, i.e., $|T|=m=2$, then \textsl{$T'$-types pairwise efficiency} coincides with \textsl{pairwise coordinatewise efficiency}.
\hfill~$\diamond$  
\end{remark}

The following result reveals the strength of \textsl{$T'$-types pairwise efficiency}.\medskip 

\begin{theorem}
    If $|T|=m>2$, then even for multiple-type housing markets with separable preferences, no mechanism satisfies \textsl{individual rationality}, \textsl{strategy-proofness}, and \textsl{$T'$-types pairwise efficiency}. \label{thm:mtpe}
\end{theorem}

We prove it by a counterexample in Appendix~\ref{appendix:thm:mtpe}.\smallskip

We conclude this section with an important remark.

\begin{remark}{\textbf{Constrained efficiency improvements}}\ \\
Based on Result~\ref{result2}~(a), we know that 
for an efficiency property based on efficiency improvements, constraints on efficiency improvements must be made in order to ensure compatibility with \textsl{individual rationality} and \textsl{strategy-proofness}.\footnote{Recall that \textsl{Pareto efficiency} is an efficiency property based on efficiency improvements without any constraints over improvements.}
In this section, we examine two categories of constraints, those regarding the constrained improvements with a certain number of object types, and those pertaining to the constrained improvements with a certain number of agents, to guarantee compatibility with \textsl{individual rationality} and \textsl{strategy-proofness}.
On the one hand, when it comes to restrictions over object types, only two constraints prove to be useful: (i) improvements for one type only and (ii) improvements for all types together. On the other hand, restrictions over agents may be unnecessary. 
To be more precise, even if we only consider improvements for two agents, we can only achieve impossibility as long as there are no restrictions on object types. \label{remakr:improvments}
\hfill~$\diamond$  
\end{remark}

{\renewcommand{\arraystretch}{1.5}
\begin{table}[htb]
\centering
\caption{efficiency properties and their implicit restriction conditions}
{\begin{threeparttable}
\begin{tabular}{|l|l|l|l|l|l|}
  \hline
   & \textbf{CE} & \textbf{pE} & \textbf{pCE} & \textbf{cE} & \textbf{$T'$-pE} \\ \hline
constrained improvements for one type & +  & $-$  & +   & $-$  &$-$   \\ \hline
constrained improvements for all types & $-$  & +  & $-$   & +  & $-$   \\ \hline
constrained improvements for two agents & $-$  & +  & +   & $-$  & +   \\ \hline
compatibility with IR and SP & + (T1) & + (T3)  & +  (T5)  & + (T6)  & $-$ (T7)  \\ \hline
\end{tabular}
\medskip

\begin{tablenotes}
\small 
We present a summary of our results in the table above.
The first row describes five efficiency properties that we consider in this section. In the first column, the first three items represent three restriction conditions that we discussed in Remark~\ref{remakr:improvments}, and the last item means the compatibility with \textsl{individual rationality} and \textsl{strategy-proofness}.

The notation ``+'' (``$-$'') in a cell means that the property satisfies (violates) the condition. The notation for the first three rows is determined by the definition of our efficiency properties.
The compatibility with \textsl{individual rationality} and \textsl{strategy-proofness} for the last row is obtained from Theorems~\ref{thm:cTTC},~\ref{thm:bttc},~\ref{thm:pce},~\ref{thm:coale}, and~\ref{thm:mtpe}, respectively.

Abbreviations in the first row respectively refer to:

\textbf{CE} stands for \textsl{coordinatewise efficiency}, 

\textbf{pE} stands for \textsl{pairwise efficiency}, 

\textbf{pCE} stands for \textsl{pairwise coordinatewise efficiency}, 

\textbf{cE} stands for \textsl{coalitional efficiency}, and

\textbf{$T'$-pE} stands for \textsl{$T'$-types pairwise efficiency}.
\end{tablenotes}
\end{threeparttable}}
\label{table:HAproperties}
\end{table}}
\newpage

\section{Discussion}\label{discussion}

\subsection{Significance}\label{discussion1}
What can we learn from our results? We discuss their significance from two perspectives. On the positive side, a straightforward implication of our two characterizations is that if the social planner cares about \textsl{individual rationality} and \textsl{strategy-proofness}, then cTTC is an outstanding candidate (for separable preferences), and if he also wants to maintain \textsl{group strategy-proofness}, then he might consider using bTTC (regardless of preference domains), as both cTTC and bTTC preserve some efficiency. In other words, if the social planner is concerned with the properties discussed in this paper, our characterizations suggest that cTTC and bTTC are the only candidates to satisfy them.\smallskip

Furthermore, our characterizations together with the impossibility results, reveal that cTTC and bTTC can be used as benchmarks in two ways. First, as mentioned in the introduction, cTTC and bTTC can primarily be interpreted as a compatibility test. This means that apart from efficiency properties that we consider, if you are interested in other reasonable efficiency properties, you can determine if they are compatible with \textsl{individual rationality} and \textsl{strategy-proofness} by checking if they are satisfied by cTTC or bTTC. This testing approach helps us uncover other interesting efficiency properties. 
Second, as mentioned in Remarks~\ref{remark:secondbest} and \ref{remark:interpretation}, aside from the objective of characterizing our two particular mechanisms, if you want to discover another attractive mechanism satisfying \textsl{individual rationality} and \textsl{strategy-proofness}, then this mechanism should not perform worse than cTTC and bTTC; otherwise, you should use cTTC or bTTC instead of this inefficient mechanism. In this respect, cTTC and bTTC help us identify the Pareto frontier of \textsl{individually rational} and \textsl{strategy-proof} mechanisms.\smallskip

On the negative side, our impossibility results also reveal two unfavorable things. First, as mentioned in Remark~\ref{remakr:improvments}, we demonstrate that, aside from the two constrained improvements we consider, it is difficult to obtain a reasonable efficiency property compatible with \textsl{individual rationality} and \textsl{strategy-proofness}. Second, as Observation~\ref{observation} suggests, in the presence of \textsl{individual rationality} and \textsl{strategy-proofness}, it is impossible to achieve an efficiency property based on constrained improvements for one type (e.g., \textsl{coordinatewise efficiency} and \textsl{pairwise coordinatewise efficiency}) and an efficiency property based on constrained improvements for all types  (e.g., \textsl{pairwise efficiency} and \textsl{coalitional efficiency})  simultaneously; a choice must be made.\smallskip

We want to emphasize that the purpose of this paper is not to characterize specific mechanisms. Instead, as the title of this paper suggests, we are interested in determining the types of efficiency properties that are compatible with \textsl{individual rationality} and \textsl{strategy-proofness}, i.e., our primary focus lies in examining the compatibility between efficiency and the other properties. However, we have made some unexpected discoveries: we find that only two mechanisms, cTTC and bTTC, satisfy our properties.  

\subsection{Related characterizations of TTC extensions}\label{discussion2}
Here we make a discussion of other characterizations of TTC extensions. It is important to clearly and concisely summarize our key findings and to relate them to other studies.\footnote{In our companion paper \citep{fengesp}, by replacing efficiency properties with incentive properties,
we provide other characterizations of cTTC and bTTC, the relation between these two papers are discussed in \citet{fengesp}. } \medskip

\noindent \textbf{Another characterization of cTTC}.
Here, we discuss the relation between our characterization of cTTC (Theorem~\ref{thm:cTTC}) and \citet[Theorem~2]{feng2022}, which characterizes cTTC on the basis of \textsl{individual rationality}, \textsl{strategy-proofness}, \textsl{non-bossiness}, and \textsl{unanimity}\footnote{In \citet{feng2022}, they refer to cTTC as the typewise top-trading-cycles mechanism. }
To distinguish between \citet{feng2022}'s result and ours, note that their result is established by weakening \textsl{Pareto efficiency} to \textsl{unanimity} and strengthening \textsl{strategy-proofness} to the combination of \textsl{strategy-proofness} and \textsl{non-bossiness}, whereas ours does not require \textsl{non-bossiness}. On the other hand, we use \textsl{coordinatewise efficiency}, which is stronger than \textsl{unanimity}.
Hence, the incentive property in our characterization is weaker while the efficiency property is stronger than in \citet[Theorem~2]{feng2022} and our result is logically independent.

\noindent \textbf{Another characterization of bTTC}.
Here, we discuss the relation between our characterization of bTTC (Theorems~\ref{thm:bttc} and \ref{thm:bttc2}) and \citet{fkkgsp}'s, who essentially characterize bTTC by means of \textsl{individual rationality}, \textsl{group strategy-proofness} and \textsl{anonymity}.\footnote{\textit{Anonymity} says that the mechanism is defined independently of the names of the agents. And they show that for separable preferences and for strict preferences,
only the class of hybrid mechanisms between the no-trade mechanism and bTTC, satisfies all of their properties.} 
Since (i) there is no logical relation between \textsl{anonymity} and \textsl{pairwise efficiency}, and (ii) the incentive property in our characterization for separable preferences
is weaker, our results are logically independent.

\subsection*{Object allocation problems with multi-demands and with ownership}
Next, we compare our results to \citet{tamura2021} and \citet{biro2022}. Each considers an extension of Shapley-Scarf housing markets.

\citet{tamura2021} consider the multiple reallocation problems model, which is a general model for allocating objects to agents who can consume any set of objects. In this model, each object is owned by an agent, but now each agent has strict preferences over all objects and his preferences over sets of objects are monotonically responsive to these ``objects-preferences.''\footnote{\citet{papai2003} also studies this model, and in fact our cTTC is is a specific case of her segmented trading cycle mechanisms. See \citet[Section~4]{feng2022} for a detailed discussion.} 
In our model, we impose more structure by assuming that (i) the set of objects is partitioned into exogenously given types and (ii) each agent owns and wishes to consume one object of each type. 
For this more general model, \citet{tamura2021} consider another TTC extension: the ``generalized top trading cycles mechanism  (gTTC),'' which satisfies \textsl{individual rationality} and \textsl{Pareto efficiency} but violates \textsl{strategy-proofness}. By strengthening \textsl{individual rationality} and weakening \textsl{strategy-proofness}, they provide a characterization of gTTC for lexicographic preferences. Thus, their results complement ours: if we exclude \textsl{strategy-proofness}, then there exists another TTC extension, which performs better than our mechanisms in terms of efficiency.\smallskip

\citet{biro2022} consider another extension where each agent owns a set of homogeneous and agent-specific objects, and they consider a modification of bTTC to their model with the ``endowments quota constraint.'' This constraint means that for each agent, the number of objects he can consume is the same as the number of objects he is endowed with. They show that this modification is neither \textsl{Pareto efficient} nor \textsl{strategy-proof}. Thus, their results show the limitation of the modified bTTC in their model, while in our model, bTTC is \textsl{group strategy-proof}.

\subsection{External validity of compatibility results}\label{discussion3}
While our compatibility results of efficiency with \textsl{individual rationality} and \textsl{strategy-proofness} are initially based on the multiple-type housing market model, they can be extended to more general environments. 
For instance, consider the multiple reallocation problem, where objects are not categorized by different types, and agents' endowments may be unevenly distributed.

Firstly, \textsl{pairwise efficiency} and \textsl{coalitional efficiency} are still well-defined, and it is evident that they are compatible with other properties, as bTTC in that model is still well-defined and satisfies all the properties. Note that in this context, we only claim the compatibility of efficiency with \textsl{individual rationality} and \textsl{strategy-proofness}: we do not know whether a characterization of bTTC similar to ours remains valid or not.
Secondly, although we cannot directly apply \textsl{coordinatewise efficiency} and \textsl{pairwise coordinatewise efficiency} in this context, suitable modifications can be made. To illustrate this, let us consider the segmented trading cycle mechanisms introduced by \citet{papai2003}. 
In these mechanisms, objects are partitioned into different segments in advance, with each segment containing at most one object from each agent; the TTC algorithm is then executed separately for each segment.
We can then assess efficiency for each segment independently, which represents an adaptation of \textsl{coordinatewise efficiency}. It is evident that segmented trading cycle mechanisms satisfy this efficiency property. Moreover, since segmented trading cycle mechanisms are \textsl{individually rational} and \textsl{strategy-proof}, we conclude that this ``segment efficiecny'' notion is compatible with \textsl{individual rationality} and \textsl{strategy-proofness}.

However, it is worth noting that these efficiency properties may not always be suitable for general models. Nevertheless, the key insight here is that ensuring compatibility still requires constraints on improvements. Two constraints, which draw inspiration from \textsl{coordinatewise efficiency} (and \textsl{pairwise coordinatewise efficiency}) and \textsl{pairwise efficiency} (and \textsl{coalitional efficiency}), respectively, are possible as discussed above.
This observation aligns with the following finding in \citet{papai2007}: trade restrictions are necessary to yield positive results.
Similarly, it is apparent that the example we used to demonstrate the impossibility results (Theorems~\ref{thm:imposs} and~\ref{thm:mtpe}) can be extended to general environments. This further emphasizes the challenge of introducing another constraints on improvements to ensure compatibility.

\section{Conclusion}
\label{sec:conclusion}
In this paper, we provide several characterizations of two extensions of Gale's TTC for multiple-type housing markets. Specifically, we present the first characterizations for strict preferences.

Our analysis sheds light on the trade-offs involved in the object allocation problem with multi-unit demands \citep{sonmez1999}.\footnote{\citet{sonmez1999} shows that if at least one agent owns more than one object, then \textsl{individual rationality}, \textsl{strategy-proofness}, and \textsl{Pareto efficiency} are incompatible.}
Given the inherent incompatibility among individual rationality, strategic robustness, and efficiency, our two TTC extensions do not meet the most stringent efficiency criterion, i.e., \textsl{Pareto efficiency}. However, each of them satisfies one attainable efficiency property. Thus, they perform remarkably well according to all three objectives.
Our work (1) suggests a direction for how to reconcile the trade-off between strategic robustness and efficiency; and (2) highlights the potential of TTC in settings where agents have multi-unit demands.

\appendix
\section{Appendix: proofs}\label{section:Appendix1}

\subsection{Lexicographic preferences} \label{appendix:unclear}
Here we introduce a new preferences domain, \textsl{lexicographic preferences}. We use this domain as stepping stone for our results. That is, we first show some results for lexicographic preferences, then extend these results to separable and strict preferences.

Before defining our next preference domain, we introduce some notation. We use a bijective function $\pi_i:T\to T$  to order types according to agent $i$'s ``(subjective) importance,'' with $\pi_i(1)$ being the most important and $\pi_i(m)$ being the least important object type. We denote $\pi_i$ as an ordered list of types, e.g., by $\pi_i=(2,3,1)$, we mean that $\pi_i(1)=2$, $\pi_i(2)=3$, and $\pi_i(3)=1$. So for each agent $i\in N$ and each allotment $x_i=(x_i^1,\ldots,x_i^m)$, by  $x_i^{\pi_i}=(x_i^{\pi_i(1)},\ldots,x_i^{\pi_i(m)})$ we denote the allotment after rearranging it with respect to the \textit{object-type importance order} $\pi_i$.\medskip

\noindent\textbf{(Separably) Lexicographic preferences.} Agent~$i$'s preferences $R_i\in\mathcal{R}$ are \textit{(separably) lexicographic} if they are separable with type-$t$ marginal preferences $(R_i^t)_{t\in T}$ and there exists an object-type importance order $\pi_i:T\to T$ such that for any two allotments $x_i$ and $y_i$,
\begin{equation*}
\begin{aligned}
& \mbox{if }x_i^{\pi_i(1)} P_i^{\pi_i(1)} y_i^{\pi_i(1)}\mbox{ or }\\
&\mbox{if there exists a positive integer }k\leq m-1 \mbox{ such that } \\
& x_i^{\pi_i(1)}=y_i^{\pi_i(1)},\ \ldots,\ x_i^{\pi_i(k)} =y_i^{\pi_i(k)},\mbox{ and }
x_i^{\pi_i(k+1)} P_i^{\pi_i(k+1)} y_i^{\pi_i(k+1)},\\
& \mbox{then } x_i \mathbin{P_i} y_i.
\end{aligned}
\end{equation*}
$\mathcal{R}_l$ denotes the \textit{domain of lexicographic preferences}.\medskip

Note that $R_i\in \mathcal{R}_l$ can be restated as a $m+1$-tuple $R_i=(R_i^1,\ldots,R_i^m,\pi_i)=((R_i^t)_{t\in T},\pi_i)$, or a strict ordering of all objects,\footnote{See \citet[Remark 1]{feng2020} for details.} 
i.e., $R_i$ lists first all $\pi(1)$ objects (according to $R_i^{\pi(1)}$), then all $\pi(2)$ objects (according to $R_i^{\pi(2)}$), and so on. We provide a simple illustration in Example~\ref{example:bTTC}.  \medskip

Note that when $|T|=m>1$, \[\mathcal{R}_l\subsetneq \mathcal{R}_s \subsetneq \mathcal{R}.\] \smallskip

\subsection{Auxiliary properties and results}\label{appendix:auxiliary-results}

We introduce the well-known property of \textit{(Maskin) monotonicity}, i.e., the invariance under monotonic transformations of preferences at a selected allocation.
We formulate \textsl{monotonicity} as well as our first auxiliary result for markets with strict preferences; however, we could use markets with separable preferences and with lexicographic preferences instead.\medskip

Let $i\in N$. Given preferences $R_i\in \mathcal{R}$ and an allotment $x_i$, let $L(x_i,R_i)=\{y_i\in \Pi_{t\in T}O^t\mid x_i\mathbin{R_i}y_i \}$ be the \textit{lower contour set of $R_i$ at $x_i$}. Preference relation $R'_i$ is a \textit{monotonic transformation of $R_i$ at $x_i$} if $L(x_i,R_i)\subseteq L(x_i,R'_i)$.
Similarly, given a preference profile $R\in \mathcal{R}^N$ and an allocation $x$, a preference profile $R'\in \mathcal{R}^N$ is a \textit{monotonic transformation of $R$ at $x$} if for each $i\in N$, $R'_i$ is a monotonic transformation of $R_i$ at $x_i$.\medskip

\noindent \textbf{Monotonicity}: For each $R\in \mathcal{R}$ and for each monotonic transformation $R'\in \mathcal{R}^N$ of $R$ at $f(R)$, $f(R')=f(R)$.\medskip

\textsl{Strategy-proofness} and \textsl{non-bossiness} imply \textsl{monotonicity}.

\begin{lemma}[Lemma~3 in \citet{feng2022}] \ \\
If a mechanism is \textsl{strategy-proof} and \textsl{non-bossy}, then it is \textsl{monotonic}. \label{lemma:mon}
\end{lemma}

Next, we list some useful results based on \textsl{strategy-proofness}, \textsl{non-bossiness}, and \textsl{monotonicity}.

\begin{fact}[Fact~1 in \citet{feng2022}] \ \\
Let $x_i$ be an allotment. Let $R_i,\hat{R}_i$ be lexicographic preferences such that (1) $\pi_i=\hat{\pi}_i$ and (2) for each $t\in T$, $\hat{R}^t_i$ is a monotonic transformation of $R^t_i$ at $x_i^t$. Then,
	$\hat{R_i}$ is a monotonic transformation of $R_i$ at $x_i$.
	\label{fact1}
\end{fact}

\begin{fact}
Let $f$ be a \textsl{strategy-proof} and \textsl{non-bossy} mechanism. Let $R\in \mathcal{R}_l^N$, $x\equiv f(R)$, $i\in N$, and $R^*_i\in \mathcal{R}_l$ be preferences that only differ with $R_i$ in the marginal preference of the most important type (type-$t$), i.e., (1) $\pi_i=\pi_i^*$ where $\pi_i^*(t)=1$, and (2) for each $\tau\neq t$, $R_i^{\tau}=R^{*\tau}_i$. 
		
If $f_i^t(R^*_i,R_{-i})=x_i^t$, then $f(R^*_i,R_{-i})=x$. \label{fact2}
	\end{fact}

\begin{proof}[\textbf{Proof}]
	It is without loss of generality to assume that $t=1$ and $\pi_i:1,\ldots,m$.
	Let $y\equiv f(R^*_i,R_{-i})$ and assume $y^1_i=x_i^1 $. By \textsl{strategy-proofness} of $f$, $x_i \mathbin{R_i} y_i$ and $y_i \mathbin{R^*_i}x_i$.
	Since $R_i$ are lexicographic preferences, $x_i \mathbin{R_i} y_i$ implies $x_i^2 \mathbin{R_i^2} y_i^2$.
	Similarly, since $R^*_i$ are lexicographic preferences, $y_i \mathbin{R^*_i}x_i$ implies $y_i^2 \mathbin{R^{*2}_i} x_i^2$. Since $R_i^2=R^{*2}_i$, we find that $x_i^2=y_i^2$. Applying the same argument sequentially for type-$\tau$ marginal preferences with $\tau=3,\ldots,m$ yields $x_i=y_i$. By \textsl{non-bossiness} of $f$, $x=y$.
\end{proof}

\subsection{Proof of Theorem~\ref{thm:cTTC}}\label{appendix:cTTC}
Here we only show the uniqueness.

Let $f:\mathcal{R}^N_l\to X$ be a mechanism satisfying \textsl{individual rationality}, \textsl{strategy-proofness}, and \textsl{coordinatewise efficiency}.

\subsubsection*{A result for restricted preferences}\  
We first consider a restricted domain $\mathcal{R}^N_{\pi}\subsetneq \mathcal{R}^N_{l}$ such that all agents share the same importance order $\pi$. It is without loss of generality to assume that $\pi:1,\ldots,m$.

\begin{proposition}
For each $R\in \mathcal{R}^N_{\pi}$, $f(R)=cTTC(R)$. \label{proposition:cttc1}
\end{proposition}

Note that Proposition~\ref{proposition:cttc1} is not a characterization result on the restricted domain. Proposition~\ref{proposition:cttc1} states that if a mechanism is defined on the lexicographic domain and satisfies the three properties, then for each profile in 
$\mathcal{R}^N_{\pi}$, this mechanism always selects the cTTC allocation. 

The proof of Proposition~\ref{proposition:cttc1} consists of three claims.\smallskip

First, we show that for each market with restricted preferences, $f$ assigns the cTTC allocation of type-$1$.
\begin{claim}
For each $R\in \mathcal{R}^N_{\pi}$, $f^1(R)=cTTC^1(R)$. \label{claim:cttc1}
\end{claim}
\begin{proof}[\textbf{Proof}]
 
Let $C$ be a first step top trading cycle under $TTC^1$ at $R$ involving a set of agents $S_C\subseteq N$. We first show that $C$ is executed at $f(R)$, i.e., for each $i\in S_C$, $f_i^1(R)=cTTC^1_i(R)$, by induction on $|S_C|$.

\noindent\textbf{\textit{Induction basis.}} $|S_C|=1$. In this case, agent $i\in S_C$ points to his type-$1$ endowed object, i.e., $C=(i\to o_i^1\to i)$. Since preferences are lexicographic, agent $i$ will be strictly worse off if he received any other type-$1$ object.
Thus, by \textsl{individual rationality} of $f$, $C$ must be executed.\smallskip

\noindent\textbf{\textit{Induction hypothesis.}} Let $K\in \{2,\ldots,n\}$. Suppose that $C$ is executed when $|S_C|<K$.\smallskip

\noindent\textbf{\textit{Induction step.}} Let $|S_C|=K$. Without loss of generality, assume that $S_C=\{1,\ldots,K\}$ and $C=(1\to o_2^1\to 2\to o_3^1 \to \ldots \to K\to o_1^1\to 1)$. 

By contradiction, assume that $C$ is not executed. Thus, there is an agent $i\in S_C$ who does not receive $o_{i+1}^1$, i.e., $f_i^{1}(R)\neq o_{i+1}^1$. 
It is without loss of generality to assume that $i=2$. We proceed by contradiction in two steps.

\noindent\textbf{Step~1.}
Let $\hat{R}_2\in \mathcal{R}_{\pi}$ be 
such that for agent $2$ and type-$1$ objects, only $o_{3}^1(=cTTC_2^1(R))$ is acceptable (apart from his type-$1$ endowment), i.e.,

$$\hat{R}^{1}_2:o_{3}^1,o_2^1,\ldots,$$
$$\text{ for each }t\in T\setminus \{1\}:\hat{R}_2^t=R_2^t, \text { and }$$
$$\hat{\pi}_2=\pi:1,\ldots,m.$$

Let $$\hat{R}\equiv (\hat{R}_2,R_{-2}).$$ 
Since $\hat{R}_2\in \mathcal{R}_l$ and $f$ is \textsl{individually rational}, $f_2^{1}(\hat{R})\in \{o_{3}^1,o_2^1\}$. By \textsl{strategy-proofness} of $f$, $f_2^1(R)\neq o_{3}^1$ implies that $f_2^{1}(\hat{R})\neq o_{3}^1$, otherwise instead of $R_2$, agent $2$ has an incentive to misreport $\hat{R}_2$ at $R$. Thus, $f_2^{1}(\hat{R})=o_2^{1}$. 
Thus, agent $1$ cannot receive $o_{2}^1(=cTTC^1_1(R))$ from agent $2$ because it is assigned to agent $2$. 
Overall, we find that 
    \begin{equation}
        f_2^{1}(\hat{R})=o_2^{1}\neq o_3^1 \text{ and } f_{1}^{1}(\hat{R})\neq o_{2}^1. \label{oneesp}
    \end{equation}

\noindent\textbf{Step~2.}    
Let $\hat{\hat{R}}_1\in \mathcal{R}_{\pi}$ be 
such that for agent $1$ and type-$1$ objects, only $o_{2}^1$ and $o_{3}^1$ are acceptable (apart from his type-$1$ endowment), i.e.,

$$\hat{\hat{R}}^{1}_1:o_{2}^1,o_{3}^1,o_1^1,\ldots, \text{ (if $K=2$ then here we have }\hat{\hat{R}}^{1}_1:o_{2}^1,o_1^1, \ldots)$$
$$\text{ for each }t\in T\setminus \{1\}:\hat{\hat{R}}_1^t=R_1^t, \text { and }$$
$$\hat{\hat{\pi}}_1=\pi:1,\ldots,m.$$

Let $$\hat{\hat{R}}\equiv (\hat{\hat{R}}_1,\hat{R}_{-1})=(\hat{\hat{R}}_1,\hat{R}_2,R_3,\ldots,R_n).$$
By \textsl{individual rationality} of $f$, $f_1^{1}(\hat{\hat{R}})\in \{o_{3}^1,o_2^1,o_1^1\}$. By \textsl{strategy-proofness} of $f$, $f_1^1(\hat{R})\neq o_{2}^1$ (see (\ref{oneesp})) implies that $f_1^{1}(\hat{\hat{R}})\neq o_{2}^1$, otherwise instead of $\hat{R}_1(=R_1)$, agent $1$ has an incentive to misreport $\hat{\hat{R}}_1$ at $\hat{R}$. 

We then show that $f_1^{1}(\hat{\hat{R}})=o_3^1$. Let $\tilde{R}^1_1:o_{3}^1,o_1^1,\ldots$ and consider $\tilde{R}_1$ that is obtained from $R_1$ by replacing type-1 marginal preferences $R_1^1$ with type-1 marginal preferences $\tilde{R}_1^1$. That is, $\tilde{R}_1=(\tilde{R}_1^1,R_1^2,\ldots,R_1^m,\pi)$. At $(\tilde{R}_1,\hat{\hat{R}}_{-1}) $, there is a top trading cycle $C'=(1\to o_3^1\to 3\to \ldots \to K\to o_1^1\to 1)$ that only involves $K-1$ agents. Thus, by the induction hypothesis, $C'$ is executed and $f_1^1(\tilde{R}_1,\hat{\hat{R}}_{-1})=o_3^1$. See the figure below for the graphical explanation.

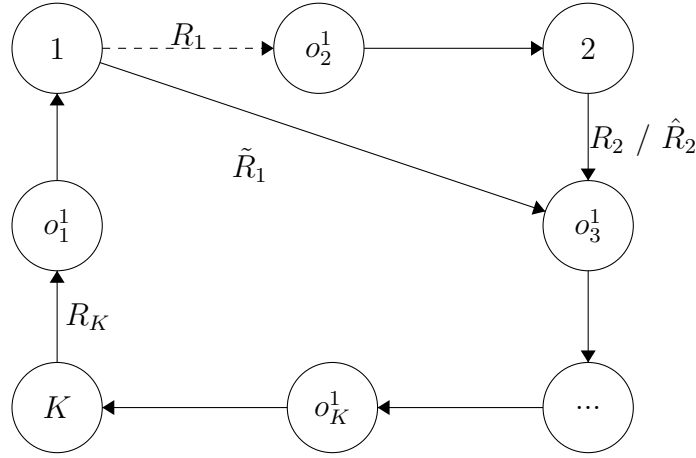
\begin{center}
    \begin{tikzpicture}[scale=0.2]
    \tikzstyle{every node}+=[inner sep=0pt]
    \draw [black] (11.4,-13.6) circle (3);
    \draw (11.4,-13.6) node {$1$};
    \draw [black] (28.8,-13.6) circle (3);
    \draw (28.8,-13.6) node {$o_2^1$};
    \draw [black] (11.4,-37.5) circle (3);
    \draw (11.4,-37.5) node {$K$};
    \draw [black] (29.7,-37.5) circle (3);
    \draw (29.7,-37.5) node {$o_K^1$};
    \draw [black] (46.7,-37.5) circle (3);
    \draw (46.7,-37.5) node {$...$};
    \draw [black] (46.7,-25.4) circle (3);
    \draw (46.7,-25.4) node {$o_3^1$};
    \draw [black] (11.4,-25.4) circle (3);
    \draw (11.4,-25.4) node {$o_1^1$};
    \draw [black] (46.7,-13.6) circle (3);
    \draw (46.7,-13.6) node {$2$};
    \draw [black] (31.8,-13.6) -- (43.7,-13.6);
    \fill [black] (43.7,-13.6) -- (42.9,-13.1) -- (42.9,-14.1);
    \draw [black] (46.7,-16.6) -- (46.7,-22.4);
    \fill [black] (46.7,-22.4) -- (47.2,-21.6) -- (46.2,-21.6);
    \draw (46.7,-19.5) node [right] {$R_2~/~\hat{R}_2$};
    \draw [black] (46.7,-28.4) -- (46.7,-34.5);
    \fill [black] (46.7,-34.5) -- (47.2,-33.7) -- (46.2,-33.7);
    \draw [black] (43.7,-37.5) -- (32.7,-37.5);
    \fill [black] (32.7,-37.5) -- (33.5,-38) -- (33.5,-37);
    \draw [black] (26.7,-37.5) -- (14.4,-37.5);
    \fill [black] (14.4,-37.5) -- (15.2,-38) -- (15.2,-37);
    \draw [black] (11.4,-34.5) -- (11.4,-28.4);
    \fill [black] (11.4,-28.4) -- (10.9,-29.2) -- (11.9,-29.2);
    \draw (11.9,-31.45) node [right] {$R_K$};
    \draw [black] (11.4,-22.4) -- (11.4,-16.6);
    \fill [black] (11.4,-16.6) -- (10.9,-17.4) -- (11.9,-17.4);
    \draw [black] (14.25,-14.55) -- (43.85,-24.45);
    \fill [black] (43.85,-24.45) -- (43.25,-23.72) -- (42.94,-24.67);
    \draw (24.21,-20.2) node [below] {$\tilde{R}_1$};
    \draw [dashed] (14.4,-13.6) -- (25.8,-13.6);
    \fill [black] (25.8,-13.6) -- (25,-13.1) -- (25,-14.1);
    \draw (20.1,-13.7) node [above] {$R_1$};
    \end{tikzpicture}
    \end{center}

Therefore, by \textsl{strategy-proofness} of $f$, $f_1^{1}(\hat{\hat{R}})=o_3^1$, otherwise instead of $\hat{\hat{R}}_1$, agent $1$ has an incentive to misreport $\tilde{R}_1$ at $\hat{\hat{R}}$. 

Moreover, $f_1^{1}(\hat{\hat{R}})=o_3^1$ implies that $f_2^{1}(\hat{\hat{R}})\neq o_3^1$. By \textsl{individual rationality} of $f$, $f_2^{1}(\hat{\hat{R}})=o_2^1$. Overall, we find that 
\begin{equation}
    f_1^{1}(\hat{\hat{R}})=o_3^1 \text{ and } f_2^{1}(\hat{\hat{R}})=o_2^1. \label{queccc}
\end{equation}

However, this equation implies that $f$ is not \textsl{coordinatewise efficient} since agents $1$ and $2$ can be better off by swapping their type-$1$ allotments. That is, for $y\equiv f(\hat{\hat{R}})$, there is a coordinatewise improvement $z$ such that (i) $z_1^1=y_2^1(=o_2^1),z_2^1=y_1^1(=o_3^1)$, and (ii) all others are the same as $y$. Thus, we conclude that $C$ is executed when $|S_C|=K$.\medskip

It suffices to show that $C$ is executed at $f(R)$ because once we have shown that agents who trade at the first step of TTC (of type-$1$) always receive their TTC allotments of type-$1$ under $f$, we can consider agents who trade at the second step of TTC (of type-$1$) by following the same proof arguments as in the first step trading cycles, and so on.
Thus, the proof of Claim~\ref{claim:cttc1} is completed.
\end{proof}
 
Note that at step~1 and step~2 of the proof of Claim~\ref{claim:cttc1}, we only require that agents in $S_C$ have restricted preferences in $\mathcal{R}_{\pi}$, i.e., if $R_{S_C}\in \mathcal{R}_{\pi}^{S_C}$ then for any $R_{-S_C}\in \mathcal{R}_l^{-S_C}$, $f^1_{S_C}(R_{S_C},R_{-S_C})=cTTC^1_{S_C}(R_{S_C},R_{-S_C})$. Therefore, Claim~\ref{claim:cttc1} implies the following fact.

\begin{fact}[Restricted preferences]\ \\
    For each $R\in \mathcal{R}^N_l$, let $\mathbb{C}\equiv\{C_1,C_2,\ldots,C_I\}$ be the set of top trading cycles that are obtained via the TTC algorithm of type-$1$ at $R^1$. Moreover, for each top trading cycle $C_i\in \mathbb{C}$, assume that $C_i$ is executed at step~$s_i$, and without loss of generality, assume that if $i<i'$ then $s_i\leq s_{i'}$. 
    
For each $C_i\in \mathbb{C}$, if all agents in $S_{C_1},S_{C_2},\ldots,S_{C_{i-1}},S_{C_i}$ have restricted preferences, then $C_1,\ldots, C_{i-1},C_i$ are executed, regardless of the preferences of other agents in $C_{i+1},\ldots,C_I$. That is, 
for each $C_i\in \mathbb{C}$, let $S'\equiv \cup_{k=1}^{i} S_{C_k}$. If $R$ is such that for each $j\in S'$, $R_j\in \mathcal{R}_{\pi}$, then 
 $f^1_{S'}(R)=cTTC^1_{S'}(R)$. \label{fact:restricted}
\end{fact}
\smallskip

Next, we show that $f$ is ``coordinatewise individually rational'' at type-$2$.
\begin{claim}
For each $R\in \mathcal{R}^N_{\pi}$ and each $i\in N$, $f_i^2(R) \mathbin{R_i^2} o_i^2$. \label{claim:cttc2}
\end{claim}
\begin{proof}[\textbf{Proof}]
By contradiction, assume that there exists $R\in \mathcal{R}^N_{\pi}$ and an agent $i\in N$ such that $o_i^2 \mathbin{P_i^2} f_i^2(R)$.

Let $y\equiv f(R)$. Recall that by Claim~\ref{claim:cttc1}, $y^1=cTTC^1(R)=TTC^1(R^1)$. 
 It is without loss of generality to assume that $i=1$. Since $R_1\in \mathcal{R}_l$ and $f$ is \textsl{individually rational},

\begin{equation}
    y_1^1\neq o_1^1. \label{equ:cttcneq}
\end{equation}

Let $\hat{R}_1\in \mathcal{R}_{\pi}$ be such that

$$\hat{R}_1^2:o_1^2,y_1^2,\ldots,$$
$$\text{ for each }t\in T\setminus\{2\},\hat{R}_1^t:y_1^t,o_1^t,\ldots,\text { and }$$
$$\hat{\pi}_1=\pi:1,\ldots,m.$$

By \textsl{strategy-proofness} of $f$, $f_1(\hat{R}_1,R_{-1})=y_1$. Note that $\hat{\pi}_1=\pi:1,\ldots,m$ and $(\hat{R}_1,R_{-1})\in \mathcal{R}^N_{\pi}$.
By Claim~\ref{claim:cttc1}, $f^1(\hat{R}_1,R_{-1})=cTTC^1(\hat{R}_1,R_{-1})=y^1$.

Let $\bar{R}_1\in \mathcal{R}_l$ be such that 
$\bar{R}_1$ and $\hat{R}_1$ only differ in the importance order, where the orders of type-$1$ and type-$2$ are switched, i.e., 
$$\text{for each }t\in T,\bar{R}_1^t= \hat{R}_1^t,\text{ and}$$
$$\bar{\pi}_1:2,1,3,\ldots,m.$$


By \textsl{individual rationality} of $f$, $f_1^2(\bar{R}_1,R_{-1})=o_1^2$ and hence $f_1^1(\bar{R}_1,R_{-1})\in \{y_1^1,o_1^1\}$. 
Since $R_1$ is lexicographic, any allotment $z_1$ with $z_1^1=y_1^1$ and $z_1^2=o_1^2$ is strictly better than $y_1$ at $R_1$.

By \textsl{strategy-proofness} of $f$, $f_1^1(\bar{R}_1,R_{-1})\neq f_1^1(\hat{R}_1,R_{-1})=y_1^1$; otherwise agent $1$ has an incentive to misreport $\bar{R}_1$ at $(\hat{R}_1,R_{-1})$. Thus,
\begin{equation}
    f_1^1(\bar{R}_1,R_{-1})=o_1^1. \label{equ:cttc2}
\end{equation}

Next, we show that (\ref{equ:cttc2}) contradicts with \textsl{coordinatewise efficiency} of $f$.

let $\ell$ be the step of the TTC algorithm at which agent $1$ receives his type-$1$ object $y_1^1$. Let $C$ be the corresponding top trading cycle that involves agent $1$, i.e., $1\in S_C$.

Note that by Claim~\ref{claim:cttc1} and Fact~\ref{fact:restricted}, all top trading cycles that are obtained before step~$\ell$ are executed at $f(\bar{R}_1,R_{-1})$. Thus, by the definition of $TTC$, we know that for each agent in $S_C$, the object that he pointed at in $C$ is his most preferred type-$1$ object among the unassigned type-$1$ objects, i.e., for each $i\in S_C$, all better type-$1$ objects for him, are assigned to someone else via some top trading cycles that are obtained before step~$\ell$.

Since $y_1^1\neq o_1^1$ (see (\ref{equ:cttcneq})), $|S_C|>1$. We show a contradiction by induction on $|S_C|$. \smallskip

\noindent\textbf{\textit{Induction basis.}} $|S_C|=2$. Without loss of generality, let $C=(1\to o_2^1\to 2\to o_1^1\to 1)$. Since $f_1^1(\bar{R}_1,R_{-1})=o_1^1$ (see (\ref{equ:cttc2})), agent $2$ does not receive his most (feasible) preferred object $o_1^1$. 

Let $R'_2$ be such that only $o_1^1$ is acceptable (apart from his type-$2$ endowment), 
By \textsl{strategy-proofness} of $f$, agent $2$ still does not receive $o_1^1$ at $f(\bar{R}_1,R'_2,R_3,\ldots,R_n)$. However, this allocation is not \textsl{coordinatewise efficient} since there is a \textsl{coordinatewise improvement} such that agent $1$ receives $o_2^1$, agent $2$ receives $o_1^1$, and all others are the same.

The following induction arguments for $K>2$ are similar to the proof of Claim~\ref{claim:cttc1}.

\noindent\textbf{\textit{Induction hypothesis.}} Let $K\in \{2,\ldots,n\}$. Suppose that $C$ is executed when $|S_C|<K$.\smallskip

\noindent\textbf{\textit{Induction step.}} Let $|S_C|=K$. Without loss of generality, assume that $S_C=\{1,\ldots,K\}$ and hence $C=(1\to o_2^1\to 2\to\ldots \to K-1 \to o_K^1\to K\to o_1^1\to 1)$. \smallskip

Recall that by (\ref{equ:cttc2}), agent $1$ receives his type-$1$ endowment and hence agent $K$ does not receive his most (feasible) preferred object $o_1^1$.
Let $\hat{\hat{R}}_K\in \mathcal{R}_{\pi}$ be 
such that for agent $K$ and type-$1$ objects, only $o_{1}^1$ and $o_{2}^1$ are acceptable (apart from his type-$1$ endowment), i.e.,

$$\hat{\hat{R}}_{K}^1:o_{1}^1,o_{2}^1,o_K^1,\ldots,$$
$$\text{ for each }t\in T\setminus \{1\}:\hat{\hat{R}}_K^t=R_K^t, \text { and }$$
$$\hat{\hat{\pi}}_K=\pi:1,\ldots,m.$$

Let $\hat{\hat{R}}\equiv (\hat{\hat{R}}_K,\bar{R}_1,{R}_{N\setminus\{1,K\}})$. By \textsl{individual rationality} of $f$, $f_K^{1}(\hat{\hat{R}})\in \{o_{1}^1,o_2^1,o_K^1\}$. By \textsl{strategy-proofness} of $f$, $f_K^1(\hat{R})\neq o_{1}^1$ (see (\ref{equ:cttc2})) implies that $f_K^{1}(\hat{\hat{R}})\neq o_{1}^1$, otherwise instead of $R_K$, agent $K$ has an incentive to misreport $\hat{\hat{R}}_K$ at $(\bar{R}_1,R_{-1})$. 

We then show that $f_K^{1}(\hat{\hat{R}})=o_2^1$. Let $\tilde{R}^1_K:o_{2}^1,o_K^1,\ldots$ and
consider $\tilde{R}_K$ that is obtained from $R_K$ by replacing type-1 marginal preferences $R_K^1$ with type-$1$ marginal preferences $\tilde{R}_K^1$. That is, $\tilde{R}_K=(\tilde{R}_K^1,R_K^2,\ldots,R_K^m,\pi)$. At $(\tilde{R}_K,\hat{\hat{R}}_{-K})$, there is a top trading cycle $C'=(2\to \ldots \to K\to o_2^1\to 2)$ that only involves $K-1$ agents. Thus, by Fact~\ref{fact:restricted} and the induction hypothesis, $C'$ is executed and $f_K^1(\tilde{R}_K,\hat{\hat{R}}_{-K})=o_2^1$. See the figure below for the graphical explanation.

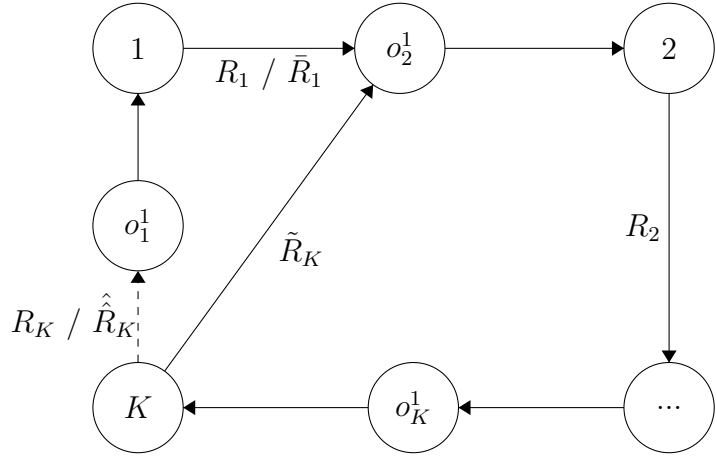
\begin{center}
    \begin{tikzpicture}[scale=0.2]
    \tikzstyle{every node}+=[inner sep=0pt]
    \draw [black] (11.4,-13.6) circle (3);
    \draw (11.4,-13.6) node {$1$};
    \draw [black] (28.8,-13.6) circle (3);
    \draw (28.8,-13.6) node {$o_2^1$};
    \draw [black] (11.4,-37.5) circle (3);
    \draw (11.4,-37.5) node {$K$};
    \draw [black] (29.7,-37.5) circle (3);
    \draw (29.7,-37.5) node {$o_K^1$};
    \draw [black] (46.7,-37.5) circle (3);
    \draw (46.7,-37.5) node {$...$};
    \draw [black] (11.4,-25.4) circle (3);
    \draw (11.4,-25.4) node {$o_1^1$};
    \draw [black] (46.7,-13.6) circle (3);
    \draw (46.7,-13.6) node {$2$};
    \draw [black] (31.8,-13.6) -- (43.7,-13.6);
    \fill [black] (43.7,-13.6) -- (42.9,-13.1) -- (42.9,-14.1);
    \draw [black] (43.7,-37.5) -- (32.7,-37.5);
    \fill [black] (32.7,-37.5) -- (33.5,-38) -- (33.5,-37);
    \draw [black] (26.7,-37.5) -- (14.4,-37.5);
    \fill [black] (14.4,-37.5) -- (15.2,-38) -- (15.2,-37);
    \draw [black] (11.4,-22.4) -- (11.4,-16.6);
    \fill [black] (11.4,-16.6) -- (10.9,-17.4) -- (11.9,-17.4);
    \draw [black] (14.4,-13.6) -- (25.8,-13.6);
    \fill [black] (25.8,-13.6) -- (25,-13.1) -- (25,-14.1);
    \draw (20.1,-14.1) node [below] {$R_1~/~\bar{R}_1$};
    \draw [black] (13.17,-35.07) -- (27.03,-16.03);
    \fill [black] (27.03,-16.03) -- (26.16,-16.38) -- (26.97,-16.97);
    \draw (20.69,-26.93) node [right] {$\tilde{R}_K$};
    \draw [dashed] (11.4,-34.5) -- (11.4,-28.4);
    \fill [black] (11.4,-28.4) -- (10.9,-29.2) -- (11.9,-29.2);
    \draw (11.4,-31.45) node [left] {$R_K~/~\hat{\hat{R}}_K$};
    \draw [black] (46.7,-16.6) -- (46.7,-34.5);
    \fill [black] (46.7,-34.5) -- (47.2,-33.7) -- (46.2,-33.7);
    \draw (46.2,-25.55) node [left] {$R_2$};
    \end{tikzpicture}
    \end{center}

Therefore, by \textsl{strategy-proofness} of $f$, $f_K^{1}(\hat{\hat{R}})=o_2^1$, otherwise instead of $\hat{\hat{R}}_K$, agent $K$ has an incentive to misreport $\tilde{R}_K$ at $\hat{\hat{R}}$. Moreover, $f_K^{1}(\hat{\hat{R}})=o_2^1$ implies that $f_1^{1}(\hat{\hat{R}})\neq o_2^1$. By \textsl{individual rationality} of $f$, $f_1^{1}(\hat{\hat{R}})=o_1^1$. Overall, we find that 
\begin{equation}
    f_K^{1}(\hat{\hat{R}})=o_2^1 \text{ and } f_1^{1}(\hat{\hat{R}})=o_1^1. 
\end{equation}

However, this equation above implies that $f$ is not \textsl{coordinatewise efficient} since agents $1$ and $K$ can be better off by swapping their type-$1$ allotments. 
\end{proof}

\smallskip

Next, we show that $f$ also selects the cTTC allocation of type-$2$.
\begin{claim}
For each $R\in \mathcal{R}^N_{\pi}$, $f^2(R) =cTTC^2(R)$. \label{claim:cttc3}
\end{claim}
\begin{proof}[\textbf{Proof}]
The proof is similar to Claim~\ref{claim:cttc1}, the main difference is that instead of \textsl{individual rationality}, we use Claim~\ref{claim:cttc2}.

Let $C$ be a first step top trading cycle at $cTTC^2(R)$ which consists of a set of agents $S_C\subseteq N$. 

Similar to Claim~\ref{claim:cttc1}, we only show that $C$ is executed at $f(R)$ by induction on $|S_C|$. It suffices to show that $C$ is executed at $f(R)$ because once we have shown that agents who trade at the first step of the TTC algorithm (of type-$2$) always receive their TTC allotments of type-$2$ under $f$, we can consider agents who trade at the second step of the TTC (of type-$2$) by following the same proof arguments as for first step trading cycles, and so on.

\noindent\textbf{\textit{Induction basis.}} $|S_C|=1$. In this case, agent $i\in S_C$ points to his type-$2$ endowed object, i.e., $C=(i\to o_i^2\to i)$. Since preferences are lexicographic, agent $i$ will be strictly worse off if he receives any other type-$2$ object.
Thus, by Claim~\ref{claim:cttc2}, $C$ must be executed.\smallskip

\noindent\textbf{\textit{Induction hypothesis.}} Let $K\in \{2,\ldots,n\}$. Suppose that $C$ is executed when $|S_C|<K$.\smallskip

\noindent\textbf{\textit{Induction step.}} Let $|S_C|=K$. Without loss of generality, assume that $S_C=\{1,\ldots,K\}$ and $C=(1\to o_2^2\to 2\to o_3^2 \to \ldots \to K\to o_1^2\to 1)$. 

By contradiction, assume that $C$ is not executed. Thus, there is an agent $i\in S_C$ who does not receive $o_{i+1}^2$, i.e., $f_i^{2}(R)\neq o_{i+1}^2$. 
It is without loss of generality to assume that $i=2$. We proceed by contradiction in two steps.

\noindent\textbf{Step~1.}
Let $\hat{R}_2\in \mathcal{R}_{\pi}$ be 
such that for agent $2$ and type-$2$ objects, only $o_{3}^2$ is acceptable (apart from his type-$2$ endowment), i.e.,

$$\hat{R}^{2}_2:o_{3}^2,o_2^2,\ldots,$$
$$\text{ for each }t\in T\setminus \{2\}:\hat{R}_2^t=R_2^t, \text { and }$$
$$\hat{\pi}_2=\pi:1,\ldots,m.$$

By Claim~\ref{claim:cttc2}, $f_2^{2}(\hat{R})\in \{o_{3}^2,o_2^2\}$. By \textsl{strategy-proofness} of $f$, $f_2^2(R)\neq o_{3}^1$ implies that $f_2^{2}(\hat{R})\neq o_{3}^2$, otherwise instead of $R_2$, agent $2$ has an incentive to misreport $\hat{R}_2$ at $R$. Thus, $f_2^{2}(\hat{R})=o_2^{2}$. 
Thus, agent $1$ cannot receive $o_{2}^2$ from agent $2$ because it is assigned to agent $2$. 
Overall, we find that 
    \begin{equation}
        f_2^{2}(\hat{R})=o_2^{2}\neq o_3^2 \text{ and } f_{1}^{2}(\hat{R})\neq o_{2}^2. \label{oneespt2}
    \end{equation}

\noindent\textbf{Step~2.}    
Let $\hat{\hat{R}}_1\in \mathcal{R}_{\pi}$ be 
such that for agent $1$ and type-$2$ objects, only $o_{2}^2$ and $o_{3}^2$ are acceptable (apart from his type-$2$ endowment), i.e.,

$$\hat{\hat{R}}^{2}_1:o_{2}^2,o_{3}^2,o_1^2,\ldots, \text{ (if $K=2$ then here we have }\hat{\hat{R}}^{2}_1:o_{2}^2,o_1^2, \ldots)$$
$$\text{ for each }t\in T\setminus \{1\}:\hat{\hat{R}}_1^t=R_1^t, \text { and }$$
$$\hat{\hat{\pi}}_1=\pi:1,\ldots,m.$$

Let $\hat{\hat{R}}\equiv (\hat{\hat{R}}_1,\hat{R}_{-1})=(\hat{\hat{R}}_1,\hat{R}_2,R_3,\ldots,R_n)$. By Claim~\ref{claim:cttc2}, $f_1^{2}(\hat{\hat{R}})\in \{o_{3}^2,o_2^2,o_1^2\}$. By \textsl{strategy-proofness} of $f$, $f_1^2(\hat{R})\neq o_{2}^2$ (see (\ref{oneespt2})) implies that $f_1^{2}(\hat{\hat{R}})\neq o_{2}^2$, otherwise instead of $\hat{R}_1(=R_1)$, agent $1$ has an incentive to misreport $\hat{\hat{R}}_1$ at $\hat{R}$. 

We then show that $f_1^{2}(\hat{\hat{R}})=o_3^2$. To see it, consider $\tilde{R}_1^2:o_3^2,o_1^2,\ldots$ and $\tilde{R}_1=(R_1^1,\tilde{R}_1^2,R_1^3,\ldots,R_1^m,\pi)$. At $(\tilde{R}_1,\hat{\hat{R}}_{-1})$, there is a top trading cycle $C'=(1\to o_3^2\to 3\to \ldots \to K\to o_1^2\to 1)$ that only involves $K-1$ agents. Thus, by the induction hypothesis, $C'$ is executed and $f_1^2(\tilde{R}_1,\hat{\hat{R}}_{-1})=o_3^2$. Therefore, by \textsl{strategy-proofness} of $f$, $f_1^{2}(\hat{\hat{R}})=o_3^2$, otherwise instead of $\hat{\hat{R}}_1$, agent $1$ has an incentive to misreport $\tilde{R}_1$ at $\hat{\hat{R}}$. Moreover, $f_1^{2}(\hat{\hat{R}})=o_3^2$ implies that $f_2^{2}(\hat{\hat{R}})\neq o_3^2$. By 
Claim~\ref{claim:cttc2}, $f_2^{2}(\hat{\hat{R}})=o_2^2$. Overall, we find that 
\begin{equation}
    f_1^{2}(\hat{\hat{R}})=o_3^2 \text{ and } f_2^{2}(\hat{\hat{R}})=o_2^2. 
\end{equation}

However, this equation above implies that $f$ is not \textsl{coordinatewise efficient} since agents $1$ and $2$ can be better off by swapping their type-$2$ allotments.

Thus, the proof of Claim~\ref{claim:cttc3} is completed.
\end{proof}
 
By Claim~\ref{claim:cttc1} and Claim~\ref{claim:cttc3}, we know that for each $R\in \mathcal{R}^N_{\pi}$, the allocations of type-$1$ and type-$2$ under $f$ are the same as $cTTC$ allocation, i.e., $f^1(R)=cTTC^1(R)$ and $f^2(R)=cTTC^2(R)$. By applying similar arguments, we can also show that $f^3(R)=cTTC^3(R)$ and so on. Thus, we conclude that for each $R\in \mathcal{R}^N_{\pi}$, and each $t\in T$, $f^t(R)=cTTC^t(R)$, which completes the proof of Proposition~\ref{proposition:cttc1}.

\subsubsection*{Proof of Theorem~\ref{thm:cTTC}}\ \\
The proof of Theorem~\ref{thm:cTTC} will be shown by extending Proposition~\ref{proposition:cttc1} to the domain of separable preference profiles.

Let $\bar{S}\subseteq N$ and $R $ be such that only agents in $\bar{S}$ do no have restricted (but separable) preferences, i.e., $R_{-\bar{S}}\in \mathcal{R}_{\pi}^{-\bar{S}}$ and for each $i\in\bar{S}$, $R_i\in \mathcal{R}_s\setminus \mathcal{R}_{\pi}$.
We show that $f(R)=cTTC(R)$ by induction on $|\bar{S}|$. \smallskip

We first consider the case that only one agent $i$ does not have restricted (but separable) preferences, i.e., $\bar{S}=\{i\}$. We will show that $f$ still selects the cTTC allocation.

\begin{lemma}
    For each $R\in \mathcal{R}^N_s$, each $i\in N$ with $R_i\in\mathcal{R}_s\setminus \mathcal{R}_{\pi}$ and $R_{-i} \in \mathcal{R}_{\pi}^{-i}$, $f(R)=cTTC(R)$. \label{lemma:basis}
\end{lemma}

The proof of Lemma~\ref{lemma:basis} consists of four claims. \smallskip

It is without loss of generality to assume that $i=1$. Thus, $R_1\in \mathcal{R}_s\setminus \mathcal{R}_{\pi}$.
 Let $y\equiv f(R)$ and $x\equiv cTTC(R)$. 

We first show that agent $1$ still receives his cTTC allocation at $R$, i.e., $y_1=x_1$. 

\begin{claim}
$y_1=x_1$. \label{claim:cttc4}
\end{claim}
\begin{proof}[\textbf{Proof}]
By contradiction, suppose that $y_1\neq x_1$.

Let $\bar{R}_1\in \mathcal{R}_{\pi}$ be such that $\bar{R}_1$ and $R_1$ share the same marginal preferences, i.e., for each $t\in T$, $\bar{R}_1^t=R_1^t$. Note that $(\bar{R}_1,R_{-1})\in \mathcal{R}_{\pi}^N $ and hence by Proposition~\ref{proposition:cttc1}, $f(\bar{R}_1,R_{-1})=cTTC(R)=x$.

Note that if for each $t\in T$, $x_1^t \mathbin{R_1^t}y_1^t$, then
$x_1 \mathbin{P_1}y_1$ as $x_1\neq y_1$. However, this implies that
agent $1$ has an incentive to misreport $\bar{R}_1$ at $R$.
Thus, by \textsl{strategy-proofness} of $f$, there exists one type $\tau\in T$ such that $y_1^\tau\mathbin{P_1^\tau}x_1^\tau$. 

By the definition of $cTTC$, $x_1^\tau \mathbin{R_1^\tau}o_1^\tau$ and hence $y_1^\tau \neq o_1^t$. Overall, we have
\begin{equation}
    y_1^\tau\mathbin{P_1^\tau}x_1^\tau \mathbin{R_1^\tau}o_1^\tau.  \label{samewhy}
\end{equation}

Let $\hat{R}_1\in \mathcal{R}_{\pi}$ be such that for each type $t\in T$, agent $1$ positions $y_1^t$ first and $o_1^t$ second, i.e.,

$$\text{ for each }t\in T:\hat{R}_1^t:y_1^t,o_1^t,\ldots,\text{ and}$$
$$\hat{\pi}_1=\pi:1,\ldots,m.$$

Let $$\hat{R}\equiv (\hat{R}_1,R_{-1}).$$

By \textsl{strategy-proofness} of $f$, $f_1(\hat{R})=f_1(R)=y_1$; otherwise agent $1$ has an incentive to misreport $R_1$ at $\hat{R}$. Since $\hat{R}\in \mathcal{R}_{\pi}^N$, by Proposition~\ref{proposition:cttc1}, $f(\hat{R})=cTTC(\hat{R})$. In particular, $y_1^\tau=cTTC_1^\tau(\hat{R} )$.

Next, to obtain the contradiction, we show that $cTTC_1^\tau(\hat{R} )=o_1^\tau$. By the definition of $cTTC$, 
\begin{equation}
    cTTC_1^\tau(\hat{R} )=TTC_1^\tau(\hat{R}^\tau)\in \{y_1^\tau,o_1^\tau\}. \label{equation:claim4}
\end{equation}

Recall that $cTTC^\tau(R)=TTC^\tau(R^\tau)=x^\tau$ and $y_1^\tau \mathbin{P_1^\tau} x_1^\tau  \mathbin{R_1^\tau} o_1^\tau$ (see (\ref{samewhy})). 
Thus, by \textsl{strategy-proofness} of $TTC$, $x_1^\tau =TTC_1^\tau(R^\tau )\mathbin{R_1^\tau} TTC_1^\tau(\hat{R}^\tau )$. Together with (\ref{equation:claim4}), we conclude that $ TTC_1^\tau(\hat{R}^\tau )  =o_1^\tau$. It
implies that $cTTC_1^\tau(\hat{R}_1,R_{-1} )=o_1^\tau\neq y_1^\tau$.
\end{proof} 

Note that Claim~\ref{claim:cttc4} implies that for each $R=(R_1,R_{-1})\in (\mathcal{R}_s\setminus{\mathcal{R}_{\pi}}) \times \mathcal{R}_{\pi}^{N\setminus\{1\}}$, $f_1(R)=cTTC_1(R)$. 

Next, we show that $y=x$ by applying similar arguments in Claims~\ref{claim:cttc1}, \ref{claim:cttc2}, and \ref{claim:cttc3}.

\begin{claim}
For each $R=(R_1,R_{-1})\in (\mathcal{R}_s\setminus{\mathcal{R}_{\pi}}) \times \mathcal{R}_{\pi}^{N\setminus\{1\}}$, $f^1(R)=cTTC^1(R)$.  \label{claim:cttc5}
\end{claim}
\begin{proof}[\textbf{Proof}]
Let $\ell$ be the step of the TTC algorithm at which agent $1$ receives type-$1$ object $y_1^1(=x_1^1=cTTC_1^1(R))$. Let $C$ be the corresponding top trading cycle that involves agent $1$, i.e., $1\in S_C$. 

Note that by Claim~\ref{claim:cttc1} and Fact~\ref{fact:restricted}, all top trading cycles that are obtained before step~$\ell$ are executed at $f(R)$. 
Moreover, if $C$ is executed, then again by Claim~\ref{claim:cttc1} and Fact~\ref{fact:restricted}, all remaining top trading cycles are also executed. 
Thus, it suffices to show that $C$ is executed, i.e., for each $i\in S_C$, $f^1_i(R)=cTTC^1_i(R)$.

Since all top trading cycles that are obtained before step~$\ell$ are executed, by the definition of $TTC$, we know that for each agent in $S_C$, the object that he pointed at in $C$ is his most preferred type-$1$ object among the unassigned type-$1$ objects, i.e., for each $i\in S_C$, all better type-$1$ objects for him, are assigned to someone else via top trading cycles that are obtained before step~$\ell$.

Similar to Claim~\ref{claim:cttc1}, we show that $C$ is executed at $f(R)$ by induction on $|S_C|$.

\noindent\textbf{\textit{Induction basis.}} $|S_C|=1$. In this case, $S_C=\{1\}$. By Claim~\ref{claim:cttc4}, $f^1_1(R)=cTTC^1_1(R)$. \smallskip

\noindent\textbf{\textit{Induction hypothesis.}} Let $K\in \{2,\ldots,n\}$. Suppose that $C$ is executed when $|S_C|<K$.\smallskip

\noindent\textbf{\textit{Induction step.}} Let $|S_C|=K$. Without loss of generality, assume that $S_C=\{1,\ldots,K\}$ and $C=(1\to o_2^1\to 2\to\ldots \to K\to o_1^1\to 1)$. \smallskip

By contradiction, assume that $C$ is not executed. Thus, there is an agent $i\in S_C\setminus\{1\}$ who does not receive $o_{i+1}^1$, i.e., $f_i^{1}(R)\neq o_{i+1}^1$. We proceed by contradiction in two steps.

\noindent\textbf{Step~1.}
Let $\hat{R}_i\in \mathcal{R}_{\pi}$ be 
such that for agent $i$ and type-$1$ objects, only $o_{i+1}^1$ is acceptable (apart from his type-$1$ endowment), i.e.,

$$\hat{R}^{1}_i:o_{i+1}^1,o_i^1,\ldots,$$
$$\text{ for each }t\in T\setminus \{1\}:\hat{R}_i^t=R_i^t, \text { and }$$
$$\hat{\pi}_1=\pi:1,\ldots,m.$$

Note that at $\hat{R}_i$, if $i$ does not receive $o_{i+1}^1$, then by \textsl{individual rationality} of $f$, he must receive his type-$1$ endowment $o_i^1$.

Let $\hat{R}\equiv (\hat{R}_i,R_{-i})$.
Since $\hat{R}_i\in \mathcal{R}_l$ and $f$ is \textsl{individually rational}, $f_i^{1}(\hat{R})\in \{o_{i+1}^1,o_i^1\}$. By \textsl{strategy-proofness} of $f$, $f_i^1(R)\neq o_{i+1}^1$ implies that $f_2^{i}(\hat{R})\neq o_{i+1}^1$, otherwise instead of $R_i$, agent $i$ has an incentive to misreport $\hat{R}_i$ at $R$. Thus, $f_i^{1}(\hat{R})=o_i^{1}$. 
Thus, agent $i-1$ cannot receive $o_{i}^1$ from agent $i$ because it is assigned to agent $i$. 
Overall, we find that 
    \begin{equation}
        f_i^{1}(\hat{R})=o_i^{1}\neq o_{i+1}^1 \text{ and } f_{i-1}^{1}(\hat{R})\neq o_{i}^1. \label{oneclaim}
    \end{equation}

\noindent\textbf{Step~2.}    
Let $\hat{\hat{R}}_{i-1}\in \mathcal{R}_{\pi}$ be 
such that for agent $i-1$ and type-$1$ objects, only $o_{i}^1$ and $o_{i+1}^1$ are acceptable (apart from his type-$1$ endowment), i.e.,

$$\hat{\hat{R}}^{1}_{i-1}:o_{i}^1,o_{i+1}^1,o_{i-1}^1,\ldots, \text{ (if $K=2$ then here we have }\hat{\hat{R}}^{1}_{i-1}:o_{i}^1,o_{i-1}^1, \ldots)$$
$$\text{ for each }t\in T\setminus \{1\}:\hat{\hat{R}}_{i-1}^t=R_{i-1}^t, \text { and }$$
$$\hat{\hat{\pi}}_{i-1}=\pi:1,\ldots,m.$$

Let $\hat{\hat{R}}\equiv (\hat{\hat{R}}_{i-1},\hat{R}_{-1})$. By \textsl{individual rationality} of $f$, $f_{i-1}^{1}(\hat{\hat{R}})\in \{o_{i+1}^1,o_i^1,o_{i-1}^1\}$. By \textsl{strategy-proofness} of $f$, $f_{i-1}^1(\hat{R})\neq o_{i}^1$ (see (\ref{oneclaim})) implies that $f_{i-1}^{1}(\hat{\hat{R}})\neq o_{i}^1$, otherwise instead of $\hat{R}_{i-1}$, agent $i-1$ has an incentive to misreport $\hat{\hat{R}}_{i-1}$ at $\hat{R}$. 

We then show that $f_{i-1}^{1}(\hat{\hat{R}})=o_{i+1}^1$. To see it, consider $\tilde{R}^1_{i-1}:o_{i+1}^1,o_{i-1}^1,\ldots$ and $\tilde{R}_{i-1}=(\tilde{R}_{i-1}^1,R_{i-1}^2,\ldots,R_{i-1}^m,\pi)$. At $(\tilde{R}_{i-1},\hat{\hat{R}}_{-(i-1)})$, there is a top trading cycle $C'=(i-1\to o_{i+1}^1\to i+1\to o_{i+2}^1\to\ldots \to i-2\to o_{i-1}^1\to i-1)$ that only involves $K-1$ agents. Thus, by the induction hypothesis, $C'$ is executed and $f_{i-1}^1(\tilde{R}_1,\hat{\hat{R}}_{-1})=o_{i+1}^1$. Therefore, by \textsl{strategy-proofness} of $f$, $f_{i-1}^{1}(\hat{\hat{R}})=o_{i+1}^1$, otherwise instead of $\hat{\hat{R}}_{i-1}$, agent $i-1$ has an incentive to misreport $\tilde{R}_{i-1}$ at $\hat{\hat{R}}$. Moreover, $f_{i-1}^{1}(\hat{\hat{R}})=o_{i+1}^1$ implies that $f_i^{1}(\hat{\hat{R}})\neq o_{i+1}^1$. By \textsl{individual rationality} of $f$, $f_i^{1}(\hat{\hat{R}})=o_i^1$. Overall, we find that 
\begin{equation}
    f_{i-1}^{1}(\hat{\hat{R}})=o_{i+1}^1 \text{ and } f_i^{1}(\hat{\hat{R}})=o_i^1. 
\end{equation}

However, this equation implies that $f$ is not \textsl{coordinatewise efficient} since agents $i-1$ and $i$ can be better off by swapping their type-$1$ allotments. 
\end{proof}

The next two claims, Claims~\ref{claim:cttc6} and \ref{claim:cttc7}, can be proven by a similar way to Claims~\ref{claim:cttc2} and \ref{claim:cttc3}, respectively. Thus, we omit the proofs. Note that the key point is that since agent $1$ still receives his cTTC allocation, we only need to show that agents who still have restricted preferences, will also receive their cTTC allocation. Thus the proofs of Claim~\ref{claim:cttc2} and \ref{claim:cttc3} are still valid for the case where only agent $1$ does not have restricted preferences.

\begin{claim}
For each $R=(R_1,R_{-1})\in (\mathcal{R}_s\setminus{\mathcal{R}_{\pi}}) \times \mathcal{R}_{\pi}^{N\setminus\{1\}}$, and each $i\in N$, $f_i^2(R) \mathbin{R_i^2} o_i^2$. \label{claim:cttc6}
\end{claim}

\begin{claim}
For each $R=(R_1,R_{-1})\in (\mathcal{R}_s\setminus{\mathcal{R}_{\pi}}) \times \mathcal{R}_{\pi}^{N\setminus\{1\}}$,  $f^2(R) =cTTC^2(R)$. \label{claim:cttc7}
\end{claim}

Thus, similar to the proof of Proposition~\ref{proposition:cttc1},
by Claims~\ref{claim:cttc4}, \ref{claim:cttc5}, \ref{claim:cttc6}, and \ref{claim:cttc7},
we conclude that for each $R\in \mathcal{R}^N_s$, and each $\bar{S}\subseteq N$, such that $|\bar{S}|=1$ and $R_{-\bar{S}}\in \mathcal{R}_{\pi}^{-\bar{S}}$, $f(R)=cTTC(R)$. Thus, the proof of Lemma~\ref{lemma:basis} is completed. 
\smallskip

Now, we are ready to prove Theorem~\ref{thm:cTTC}. Let $R\in \mathcal{R}^N_l$ and $\bar{S}\subseteq N$ be such that exactly only agents in $\bar{S}$ have non restricted (but separable) preferences, we show that $f(R)=cTTC(R)$. 

\begin{lemma}
    For each $R\in \mathcal{R}^N_s$ and each $\bar{S}\subseteq N$
    such that $R_{-\bar{S}}\in \mathcal{R}_{\pi}^{-\bar{S}}$, $f(R)=cTTC(R)$. 
    \label{lemma:step}
\end{lemma}

The proof of Lemma~\ref{lemma:step} is showing by induction on $|\bar{S}|$.\smallskip

\noindent\textbf{\textit{Induction basis.}} $|\bar{S}|=1$. This is done by Lemma~\ref{lemma:basis}.\smallskip

\noindent\textbf{\textit{Induction hypothesis.}} Let $K\in \{2,\ldots,n\}$. Suppose that $f(R)=cTTC(R)$ when $|\bar{S}|<K$.\smallskip

\noindent\textbf{\textit{Induction step.}} Let $|\bar{S}|=K$. Similar to Lemma~\ref{lemma:basis}, the proof of this part consists of four claims.\smallskip

We first show that agents in $\bar{S}$ still receive their cTTC allotments.
\begin{claim}
    For each $R=(R_{\bar{S}},R_{-\bar{S}})\in (\mathcal{R}_s\setminus{\mathcal{R}_{\pi}})^{\bar{S}} \times \mathcal{R}_{\pi}^{-\bar{S}}$,
    and each $i\in \bar{S}$, $f_i(R)=cTTC_i(R)$. 
\end{claim}

\begin{proof}[\textbf{Proof}]
Let $y\equiv f(R)$ and $x=cTTC(R)$.
By contradiction, assume that there is an agent $i\in \bar{S}$ who does not receive his cTTC allotment $x_i$. Without loss of generality, assume that $i=1$.

Let $\bar{R}_1\in \mathcal{R}_{\pi}$ be such that $\bar{R}_1$ and $R_1$ share the same marginal preferences, i.e., for each $t\in T$, $\bar{R}_1^t=R_1^t$. Let $\bar{R}\equiv (\bar{R}_1,R_{-1})$. 

Note that at $\bar{R} $, there are only $K-1$ agents (in $\bar{S}\setminus\{1\}$) who have non restricted (but separable) preferences. Thus, by the induction hypothesis and the definition of $cTTC$, $f(\bar{R})=cTTC(\bar{R})=cTTC(R)=x$. Then, the remaining proof is exactly the same as the proof of Claim~\ref{claim:cttc4} and hence we omit it.
\end{proof}

The following three claims can be proven by a similar way to Claims~\ref{claim:cttc5}, \ref{claim:cttc6}, and \ref{claim:cttc7}, respectively. Thus, we omit the proofs. Note that the key point is that since agents in $\bar{S}$ still receive their cTTC allocation, we only need to show that agents who still have restricted preferences, will also receive their cTTC allocation. Thus the proofs are still valid for the case where only agents in $\bar{S}$ do not have restricted preferences.

\begin{claim}
    For each $R=(R_{\bar{S}},R_{-\bar{S}})\in (\mathcal{R}_s\setminus{\mathcal{R}_{\pi}})^{\bar{S}} \times \mathcal{R}_{\pi}^{-\bar{S}}$,
     $f^1(R)=cTTC^1(R)$.
\end{claim}

\begin{claim}
    For each $R=(R_{\bar{S}},R_{-\bar{S}})\in (\mathcal{R}_s\setminus{\mathcal{R}_{\pi}})^{\bar{S}} \times \mathcal{R}_{\pi}^{-\bar{S}}$, and each $i\in N$,
     $f^2_i(R)\mathbin{R^2_i} o_i^2$.
\end{claim}

\begin{claim}
    For each $R=(R_{\bar{S}},R_{-\bar{S}})\in (\mathcal{R}_s\setminus{\mathcal{R}_{\pi}})^{\bar{S}} \times \mathcal{R}_{\pi}^{-\bar{S}}$,
     $f^2(R)=cTTC^2(R)$.
\end{claim}

  Hence, we conclude that for each $R\in \mathcal{R}^N_s$, and each $\bar{S}\subseteq N$, such that $|\bar{S}|=K$ and $R_{-\bar{S}}\in \mathcal{R}_{\pi}^{-\bar{S}}$, $f(R)=cTTC(R)$. Thus, the proof of Lemma~\ref{lemma:step} is completed.
  Therefore, Theorem~\ref{thm:cTTC} is proven by applying Lemma~\ref{lemma:step} with $\bar{S}=N$.

\subsection{Proofs in Subsection~\ref{subs:Pairwise efficiency}}\label{appendix:bttc}
We only show the uniqueness here.
Before doing it, we redefine bTTC for lexicographic preferences as follows.

\subsubsection*{Alternative definition of bTTC}\ \\
We restate bTTC for lexicographic preferences by adjusting the multiple-type top trading cycles (mTTC) algorithm from \citet{feng2020}.\medskip

\noindent\textbf{The bundle top trading cycles (bTTC) algorithm / mechanism.} \smallskip

\noindent\textbf{Input.} A multiple-type housing market $(N,e,R)$ with $R\in\mathcal{R}^N_l$.\smallskip

\noindent\textbf{Step~1.} \textbf{\textit{Building step.}} Let $N(1)=N$ and $U(1)=O$. We construct a directed graph $G(1)$ with the set of nodes $N(1)\cup U(1)$. For each $o\in U(1)$, we add an edge from the object to its owner and for each $i\in N(1)$, we add an edge from the agent to his most preferred object in $O$ (according to the linear representation of $R_i$). For each edge  $(i,o)\in N\times O$  we say that agent $i$ points to object $o$.\smallskip

\noindent\textbf{\textit{Implementation step.}} A \textit{trading cycle} is a directed cycle in graph $G(1)$. Given the finite number of nodes, at least one trading cycle exists. We assign to each agent $i$ in a trading cycle the object that he pointed to, and denote the object assigned to him in this step by $a_i(1)$. Moreover, let $e_i(1)$ be the whole endowment of object $a_i(1)$'s owner, and assign the allotment $x_i(1)=\{e_i(1)\}$ to agent $i$. If agent $i\in N$  was not part of a trading cycle, then $x_i(1)=\emptyset$.\smallskip

\noindent\textbf{\textit{Removal step.}} We remove all agents and objects that were assigned in the implementation step, let $N(2)$ and $U(2)$ be the remaining agents and objects, respectively. Go to Step~$2$.\smallskip

In general, at Step $q\ (\geq 2)$ we have the following:\smallskip

\noindent\textbf{Step~$\bm{q}$.} If $U(q)$ (or equivalently $N(q)$) is empty, then stop; otherwise do the following.\smallskip

\noindent\textbf{\textit{Building step.}} We construct a directed graph $G(q)$ with the set of nodes $N(q)\cup U(q)$. For each $o\in U(q)$, we add an edge from the object to its owner and for each $i\in N$, we add an edge from the agent to his most preferred feasible continuation object in $U_i(q)$  (according to the linear representation of $R_i$).\smallskip

\noindent\textbf{\textit{Implementation step.}} A \textit{trading cycle} is a directed cycle in graph $G(q)$. Given the finite number of nodes, at least one trading cycle exists. We assign to each agent $i$ in a trading cycle the object that he pointed to, and denote the object assigned to him in this step by $a_i(q)$.
Moreover, let $e_i(q)$ be the whole endowment of object $a_i(q)$'s owner, and assign the allotment $x_i(q)=\{e_i(q)\}$ to agent $i$. If agent $i\in N$  was not part of a trading cycle, then $x_i(q)=\emptyset$.\smallskip

\noindent\textbf{\textit{Removal step.}} We remove all agents and objects that were assigned in the implementation step, let $N(q+1)$ and $U(q+1)$ be the remaining agents and objects, respectively. Go to Step~$q+1$.\smallskip

\noindent\textbf{Output.} The bTTC algorithm terminates when all objects in $O$ are assigned  (it takes at most $n$ steps). Assume that the final step is Step~$q^*$. Then, the final allocation is $x(q^*)=\{x_1(q^*),\ldots,x_n(q^*)\}$.\smallskip

The \textit{bundle top trading cycles mechanism (bTTC)}, $bTTC$, assigns to each market $R\in \mathcal{R}^N_l$ the allocation $x(q^*)$ obtained by the bTTC algorithm.\medskip

\begin{example}[\textbf{bTTC}]\ \\
    Consider $R \in\mathcal{R}^N_l$  with $N=\{1,2,3\}$, $T= \{H(ouse),C(ar)\}$, $O=\{H_1,H_2,H_3,C_1,C_2,C_3\}$, and
    $$\bm{R_1}:H_2,H_3,\bm{H_1},C_3,C_2,\bm{C_1},$$
    $$\bm{R_2}:C_1,\bm{C_2},C_3, H_3,\bm{H_2},H_1,$$
    $$\bm{R_3}:H_2,H_1,\bm{H_3},C_1,\bm{C_3},C_2.$$	
    
    The bTTC allocation at $R$ is obtained as follows.
    
    \noindent \textbf{Step~$\bm{1}$.} \textbf{\textit{Building step.}} $G(1)=(N\cup O,E(1))$ with set of directed edges \newline $E(1)=\{(H_1,1),(H_2,2),(H_3,3),(C_1,1),(C_2,2),(C_3,3), (1,H_2),(2,C_1),(3,H_2)\}$.\medskip
    
    \noindent\textbf{\textit{Implementation step.}} The trading cycle $1\to H_2\to 2\to C_1\to 1$ forms. Then, $a_1(1)=H_2$, $a_2(1)=C_1$, and $e_1(1)=\{H_2,C_2\}$, $e_2(1)=\{H_1,C_1\}$;  thus, $x_1(1)=\{H_2,C_2\}$, $x_2(1)=\{H_1,C_1\}$, and $x_3(1)=\emptyset$.\medskip
    
    \noindent\textbf{\textit{Removal step.}}
    $N(2)={3}$, $U(2)=\{H_3,C_3\}$.\medskip
    
    \noindent \textbf{Step~$\bm{2}$.} \textbf{\textit{Building step.}} $G(2)=(N(2)\cup U(2),E(2))$ with set of directed edges $E(2)=\{(H_3,3),(C_3,3),(3,H_3)\}$.\medskip
    
    \noindent\textbf{\textit{Implementation step.}} The trading cycle $3\to H_3\to 3$ forms. Then,  $a_3(2)=H_3$ and $e_3(2)=\{H_3,C_3\}$; $x_1(2)=\{H_2,C_2\}$, $x_2(2)=\{H_1,C_1\}$, and $x_3(2)=\{H_3,C_3\}$.\medskip
    
    \noindent\textbf{\textit{Removal step.}} $N(3)=\emptyset$ and $U(3)=\emptyset$.\medskip
    
    Thus, the bTTC algorithm computes the allocation $x=((H_2,C_2) ,(H_1,C_1),(H_3,C_3))$.\hfill~$\diamond$ \label{example:bTTC}
\end{example}

Let $R\in \mathcal{R}^N_l$, let $\mathcal{C}(R)$ be a set of top trading cycles that are obtained at step~1 of the re-defined bTTC above at $R$. We say that a trading cycle $C$ is a \textit{first step top trading cycle} if $C\in \mathcal{C}(R)$.
For each first step top trading cycle $C$, let $S_C\subseteq N$ be the set of agents who are involved in $C$, and for each $i\in S_C$, let $c_i$ be the object that agent $i$ points at in $C$, and $t_i$ be the type of object $c_i$, i.e., $c_i\in O^{t_i}$.
We say that a trading cycle $C$ is \textit{executed} at $f(R)$ if 
for each $i\in S_C$, agent $i$ receives $c_i$ at $f(R)$. 
Moreover, for each $i\in S_C$ and $c_i\in O$, let $i'$ be the owner of $c_i$. Since $i$ and $i'$ are involved in $C$, $i' \in S_C$. We say that a trading cycle C is \textit{fully executed} at $f( R)$ if for each $i\in S_C$, agent $i$ receives $e_{i'}$ at $f( R)$, i.e.,
$f_i(R) = e_{i'}$.\medskip

Next, we show the characterization of bTTC for lexicographic preferences. 

\begin{theorem}
    A mechanism $f:\mathcal{R}^N_l\to X$ is \textsl{individually rational}, \textsl{strategy-proof}, \textsl{non-bossy}, and \textsl{pairwise efficient} if and only if it is bTTC. \label{thm:bttc3}
\end{theorem}

\subsubsection*{Proof of Theorem~\ref{thm:bttc3}}\ \\

Let $f:\mathcal{R}^N_l\to X$ be \textsl{individually rational}, \textsl{strategy-proof}, \textsl{non-bossy}, and \textsl{pairwise efficient}. Note that by Lemma~\ref{lemma:mon}, $f$ is \textsl{monotonic}. 

We first explain the intuition of the proof. 
Consider a first step top trading cycle that forms at the first step of bTTC. First, we show that if this first-step top trading cycle is formed by only one or two agents, it is fully executed under $f$ (Lemma~\ref{lemma:bttceff1}) . Then, we extend this result to any number of agents under $f$ (Lemma~\ref{lemma:bttceff2}). 
Once we have shown that agents who trade at the first step of bTTC
always receive their bTTC allotments under $f$, we can consider agents who trade at later steps of bTTC. The full execution of 
second step top trading cycles can be shown by following the same proof arguments as for first step top trading cycles; etc.
The formal proof for first step top trading cycles now follows.

\begin{lemma}
    If a mechanism $f:\mathcal{R}^N_l\to X$ is \textsl{individually rational}, \textsl{strategy-proof}, \textsl{non-bossy}, and \textsl{pairwise efficient}, then for each $R\in \mathcal{R}^N_l$, each first step top trading cycle $C(\in \mathcal{C}(R))$ with $|S_C|\leq 2$, $C$ is fully executed at $f(R)$.  \label{lemma:bttceff1}
   \end{lemma}
 \begin{proof}[\textbf{Proof}]
 Let $C\in\mathcal{C}(R)$ be a first step top trading cycle that consists of agents $S_C$ with $|S_C|\leq 2$. 
 
We show it by two steps. First, we show that $C$ is executed.
    \begin{claim}
   $C$ is executed. \label{claim:bttceff1}
   \end{claim}
   \begin{proof}

   When $|S_C|=1$. In this case, agent $i\in S_C$ points to one of his endowed objects, i.e., $c_i=o_i^{t_i}$ and hence $C=(i\to c_i\to i)$. 
   Since preferences are lexicographic, i.e., $R_i\in \mathcal{R}_l$, 
   agent $i$ will be strictly worse off if he receives any other type-$t_i$ objects. Thus, $C$ must be executed by \textsl{individual rationality} of $f$.

   When $|S_C|=2$. Without loss of generality, assume that $S_C=\{1,2\}$. By contradiction, assume that $C$ is not executed. Thus, there is an agent $i\in S_C$ does not receive his most preferred object $c_i$. Without loss of generality, let $i=2$. 
   
   Let $\hat{R}_2$ be such that agent $2$ only wants to receive type-$t_2$ object $c_2$ and no other objects, i.e.,
     $$\hat{R}^{t_2}_2:c_2(=o_1^{t_2}),o^{t_2}_2,\ldots,$$
   $$\text{ for each }t\in T\setminus \{t_2\}:\hat{R}_2^t:o_2^t,\ldots, \text{ and }$$
   $$\hat{\pi}_2=\pi_2:t_2,\ldots.$$
   
    Note that at $\hat{R}_2$, if $2$ does not receive $c_2$, then from individual rationality of $f$, he must receive his full endowment $e_2=(o_2^1,\ldots,o_2^m)$. 
    Let 
    $$\hat{R}\equiv (\hat{R}_2,R_{-2}).$$
    
    By \textsl{individual rationality} of $f$, $f_2^{t_2}(\hat{R})\in \{c_2,o_2^{t_2}\}$. By \textsl{strategy-proofness} of $f$, $f_2^{t_2}(R)\neq c_2$ implies that $f_2^{t_2}(\hat{R})\neq c_2$, otherwise instead of $R_2$, agent $2$ has an incentive to misreport $\hat{R}_2$ at $R$. Thus, $f_2^{t_2}(\hat{R})=o_2^{t_2}$. Then, by \textsl{individual rationality} of $f$, $f_2(\hat{R})=e_2$. Thus, agent $1$ cannot receive $c_{1}(\in e_2)$ from agent $2$ because it is assigned to agent $2$. 
    
    Let $\bar{R}_1$ be such that be such that agent $1$ only wants to receive type-$t_1$ object $c_1$ and no other objects, i.e.,

$$\bar{R}^{t_1}_1:c_1(=o_2^{t_1}),o_1^{t_1},\ldots,$$  
$$\text{ for each }t\in T\setminus \{t_1\}:\bar{R}_1^t:o_1^{t},\ldots \text{ and }$$
   $$\bar{\pi}_1=\hat{\pi}_1=\pi_1:t_1,\ldots.$$     

Note that at $\bar{R}_1$, if agent $1$ does not receive $c_1$, then from individual rationality of $f$, he must receive his full endowment $e_1$. 
Let $$\bar{R}\equiv (\bar{R}_1,\hat{R}_2,\hat{R}_3,\ldots,\hat{R}_n)=(\bar{R}_1,\hat{R}_2,R_3,\ldots,R_n).$$


By \textsl{individual rationality} of $f$, $f_1^{t_1}(\bar{R})\in \{c_1,o_1^{t_1}\}$. By \textsl{strategy-proofness} of $f$, $f_1^{t_1}(\hat{R})\neq c_1$ implies that $f_1^{t_1}(\bar{R})\neq c_1$, otherwise instead of $\hat{R}_1$, agent $1$ has an incentive to misreport $\bar{R}_1$ at $\hat{R}$. Thus, $f_1^{t_1}(\bar{R})=o_1^{t_1}$. Then, by \textsl{individual rationality} of $f$, $f_1(\bar{R})=e_1$, and in particular, $f_1^{t_2}(\bar{R})=o_1^{t_2}=c_2$.
Moreover, by \textsl{individual rationality} of $f$, $f_2(\bar{R})=e_2$, and in particular, $f_2^{t_1}(\bar{R})=o_2^{t_1}=c_1$. This implies that $f_2^{t_1}(\bar{R}) \mathbin{P_1^{t_1}} f^{t_1}_1(\bar{R}) $ and $f_1^{t_2}(\bar{R}) \mathbin{P_2^{t_2}} f^{t_2}_2(\bar{R}) $ and hence $f_2(\bar{R})  \mathbin{P_1}f_1(\bar{R}) $ and $f_1(\bar{R})  \mathbin{P_2}f_2(\bar{R}) $, in which contradicts with \textsl{pairwise efficiency} of $f$.

Overall, by the contradiction above we show that at $\bar{R}$, agent $1$ receives $c_1$. Thus, by \textsl{strategy-proofness} of $f$, he also receives $c_1$ at $\hat{R}$; otherwise he has an incentive to misreport $\bar{R}_1$ at $\hat{R}$. Together with \textsl{individual rationality} of $f$, it implies that agent $2$ receives $c_2$ at $\hat{R}$. Therefore, by \textsl{strategy-proofness} of $f$, agent $2$ also receives $c_2$ at $R$; otherwise he has an incentive to misreport $\hat{R}_1$ at $R$.  \medskip
   \end{proof}

Next, we show that $C$ is fully executed. There are two cases.

    \noindent\textbf{Case~1.} $|S_C|=1$. 
   In this case, agent $i\in S_C$ points to one of his endowed objects, i.e., $c_i\in e_i$. Without loss of generality, assume that $S_C=\{1\}$ and $\pi_1:t_1,\ldots$. Thus, $C=(1\to o_1^{t_1}\to 1)$.

   Let $y\equiv f(R)$. By contradiction, suppose that $y_1\neq e_1$. 
   Note that by Claim~\ref{claim:bttceff1}, $y_1^{t_1}=o_1^{t_1}$.  Let $t\in T\setminus\{t_1\}$ be such that $y_1^t\neq o_1^t$. Without loss of generality, assume that agent $1$ receives agent $2$'s endowment of type-$t$ at $y$, i.e., $y_1^t=o_2^t$.
   
   Let $\hat{R}\in \mathcal{R}^N_l$ be such that each agent $j$ positions $y_j$ at the top and changes his importance order as $\pi_1$, i.e., for each agent $j\in N$,  (i) $\hat{\pi}_j=\pi_1:t_1,\ldots$, and (ii) for each $\tau \in T$, $\hat{R}_{j}^\tau:y_{j}^\tau,\ldots$
   
   It is easy to see that $\hat{R}$ is a monotonic transformation of $R$ at $y$. Thus, by monotonicity of $f$, $f(\hat{R})=y$. 
   
   Let $\bar{R}_2$ be such that
   $$\bar{\pi}_2=\hat{\pi}_2(=\pi_1),$$
   $$\bar{R}_2^{t_1}:o_1^{t_1},y_2^{t_1},\ldots, \text{ and }$$
   $$\text{For each }\tau \in T\setminus\{t_1\}, \bar{R}_2^{\tau}=\hat{R}_2^{\tau}.$$
   
   Let $$\bar{R}\equiv (\bar{R}_2,\hat{R}_{-2}).$$
   
   Note that by strategy-proofness of $f$, for type-$t_1$, agent $2$ either receives $o_1^{t_1}$ or $y_2^{t_1}$; otherwise he has an incentive to misreport $\hat{R}_2$ at $\bar{R}$.
   Moreover, $C$ is still a first top trading cycle at $\bar{R}$, i.e., $C\in \mathcal{C}(\bar{R})$. Thus, by Claim~\ref{claim:bttceff1}, $C$ is executed and hence agent $1$ receives $o_1^{t_1}$ at $f(\bar{R})$. See the figure below for the graphical explanation.

   
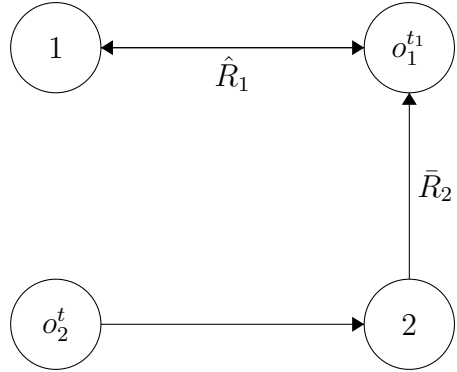
\begin{center}
       \begin{tikzpicture}[scale=0.2]
       \tikzstyle{every node}+=[inner sep=0pt]
\draw [black] (23.4,-20.3) circle (3);
\draw (23.4,-20.3) node {$1$};
\draw [black] (46.9,-20.3) circle (3);
\draw (46.9,-20.3) node {$o_1^{t_1}$};
\draw [black] (46.9,-38.7) circle (3);
\draw (46.9,-38.7) node {$2$};
\draw [black] (23.4,-38.7) circle (3);
\draw (23.4,-38.7) node {$o_2^t$};
\draw [black] (43.9,-20.3) -- (26.4,-20.3);
\fill [black] (26.4,-20.3) -- (27.2,-20.8) -- (27.2,-19.8);
\draw [black] (26.4,-20.3) -- (43.9,-20.3);
\fill [black] (43.9,-20.3) -- (43.1,-19.8) -- (43.1,-20.8);
\draw (35.15,-20.8) node [below] {$\hat{R}_1$};
\draw [black] (26.4,-38.7) -- (43.9,-38.7);
\fill [black] (43.9,-38.7) -- (43.1,-38.2) -- (43.1,-39.2);
\draw [black] (46.9,-35.7) -- (46.9,-23.3);
\fill [black] (46.9,-23.3) -- (46.4,-24.1) -- (47.4,-24.1);
\draw (47.4,-29.5) node [right] {$\bar{R}_2$};
       \end{tikzpicture}
\end{center}

   Thus, agent $2$ still receives $y_2^{t_1}$, and hence by Fact~\ref{fact2}, 
   \begin{equation}
       f(\bar{R})=f(\hat{R})=y \text{ and particularly, } f_1^{t_1}(\tilde{R})=y_1^{t_1}=o_1^{t_1}. \label{equ111}
   \end{equation}

    
   Let $\tilde{R}_1$ be such that agent $1$ only changes his importance order as $t$ is the most important, i.e., $\tilde{\pi}_1:t,\ldots$ and $\tilde{R}_1=(\hat{R}_1^1,\ldots,\hat{R}_1^m,\tilde{\pi}_1)$.
   
   Let $$\tilde{R}\equiv (\tilde{R}_1,\bar{R}_{-1}).$$
   
   By monotonicity of $f$, $f(\tilde{R})=f(\bar{R})=y$. However, at $\tilde{R}$, there is a first step top trading cycle $C'\in \mathcal{C}(\tilde{R})$ consisting of agents $1$ and $2$, i.e., $C'=(1\to o_2^t(=y_1^t)\to 2\to o_1^{t_1}\to 1)$. See the figure below for the graphical explanation.

   
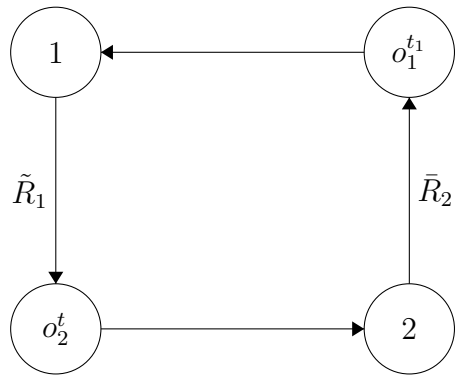
\begin{center}
       \begin{tikzpicture}[scale=0.2]
       \tikzstyle{every node}+=[inner sep=0pt]
\draw [black] (23.4,-20.3) circle (3);
\draw (23.4,-20.3) node {$1$};
\draw [black] (46.9,-20.3) circle (3);
\draw (46.9,-20.3) node {$o_1^{t_1}$};
\draw [black] (46.9,-38.7) circle (3);
\draw (46.9,-38.7) node {$2$};
\draw [black] (23.4,-38.7) circle (3);
\draw (23.4,-38.7) node {$o_2^t$};
\draw [black] (26.4,-38.7) -- (43.9,-38.7);
\fill [black] (43.9,-38.7) -- (43.1,-38.2) -- (43.1,-39.2);
\draw [black] (46.9,-35.7) -- (46.9,-23.3);
\fill [black] (46.9,-23.3) -- (46.4,-24.1) -- (47.4,-24.1);
\draw (47.4,-29.5) node [right] {$\bar{R}_2$};
\draw [black] (43.9,-20.3) -- (26.4,-20.3);
\fill [black] (26.4,-20.3) -- (27.2,-20.8) -- (27.2,-19.8);
\draw [black] (23.4,-23.3) -- (23.4,-35.7);
\fill [black] (23.4,-35.7) -- (23.9,-34.9) -- (22.9,-34.9);
\draw (22.9,-29.5) node [left] {$\tilde{R}_1$};
       \end{tikzpicture}
\end{center}
   
   By Claim~\ref{claim:bttceff1}, $C'$ is executed at $f(\tilde{R})$. Thus, $f_2^{t_1}(\tilde{R})=o_1^{t_1}$, which contradicts with the fact that $f_1^{t_1}(\tilde{R})=y_1^{t_1}=o_1^{t_1}$ (see (\ref{equ111})).

    \noindent\textbf{Case~2.} $|S_C|=2$. Without loss of generality, assume that $S_C=\{1,2\}$. Thus, $C=(1\to c_1(=o_2^{t_1})\to 2\to c_2(=o_1^{t_2})\to 1)$.
    By contradiction, assume that $C$ is not fully executed. Without loss of generality, assume that agent $1$ does not receive agent $2$'s full endowments, i.e., $f_1(R)\neq e_2$. Note that by Claim~\ref{claim:bttceff1}, $f_1^{t_1}(R)=c_1=o_2^{t_1}$. Thus, there is a type $t\in T\setminus\{t_1\}$ such that $f_1^{t}(R)\neq o_2^t$. Without loss of generality, assume that agent $1$ receives agent $i$'s endowment of type-$t$, i.e., $f_1^{t}(R)=o_i^t$. 
   Let $y\equiv f(R)$. There are two sub-cases.
   
     \noindent{Sub-case~1.} $i=1$. Let $\hat{R}_1$ be such that 
     $$\text{for each } t\in T, \hat{R}_1^t:y_1^t,\ldots, \text{ and}$$
     $$\hat{\pi}_1:t,\ldots$$
     
     By \textsl{monotonicity} of $f$, $f(\hat{R}_1,R_{-1})=f(R)=y$. Then, we are back to Case~1. 
     
      \noindent{Sub-case~2.} $i\neq 1$. Without loss of generality, assume that $i=3$. Thus, $y_2^{t_2}=o_1^{t_2}$, $y_1^{t_1}=o_2^{t_1}$, and $y_1^t=o_3^t$. We will obtain a contradiction to complete the proof of this sub-case. Let $\hat{R}_3$ be such that 
     $$R_3^{t_2}:o_1^{t_2},y_3^{t_2},\ldots $$ 
     $$\text{for each } t\in T\setminus\{t_2\}, \hat{R}_3^t:y_3^t,\ldots, \text{ and}$$
     $$\hat{\pi}_3:t_2,\ldots$$
     
     Let $$\hat{R}\equiv (\hat{R}_3,R_{-3}).$$
     
     Note that $C$ is still a first step top trading cycle at $\hat{R}$, i.e., $C\in\mathcal{C}(\hat{R})$. Thus, by Claim~\ref{claim:bttceff1}, $C$ is executed. See the figure below for the graphical explanation.

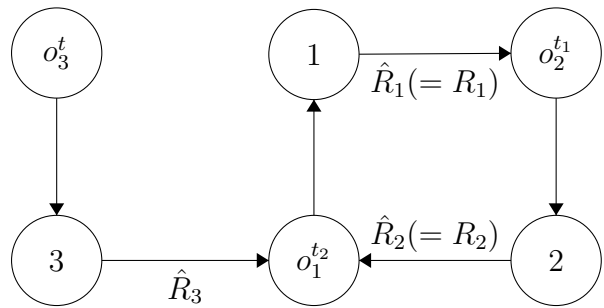
\begin{center}
       \begin{tikzpicture}[scale=0.2]
       \tikzstyle{every node}+=[inner sep=0pt]
       \draw [black] (36,-30.9) circle (3);
       \draw (36,-30.9) node {$o_1^{t_2}$};
       \draw [black] (36,-17.2) circle (3);
       \draw (36,-17.2) node {$1$};
       \draw [black] (52.1,-17.1) circle (3);
       \draw (52.1,-17.1) node {$o_2^{t_1}$};
       \draw [black] (52.1,-30.9) circle (3);
       \draw (52.1,-30.9) node {$2$};
       \draw [black] (18.9,-30.9) circle (3);
       \draw (18.9,-30.9) node {$3$};
       \draw [black] (18.9,-17.1) circle (3);
       \draw (18.9,-17.1) node {$o_3^t$};
       \draw [black] (36,-27.9) -- (36,-20.2);
       \fill [black] (36,-20.2) -- (35.5,-21) -- (36.5,-21);
       \draw [black] (39,-17.18) -- (49.1,-17.12);
       \fill [black] (49.1,-17.12) -- (48.3,-16.62) -- (48.3,-17.62);
       \draw (44.05,-17.66) node [below] {$\hat{R}_1(=R_1)$};
       \draw [black] (52.1,-20.1) -- (52.1,-27.9);
       \fill [black] (52.1,-27.9) -- (52.6,-27.1) -- (51.6,-27.1);
       \draw [black] (49.1,-30.9) -- (39,-30.9);
       \fill [black] (39,-30.9) -- (39.8,-31.4) -- (39.8,-30.4);
       \draw (44.05,-30.4) node [above] {$\hat{R}_2(=R_2)$};
       \draw [black] (18.9,-20.1) -- (18.9,-27.9);
       \fill [black] (18.9,-27.9) -- (19.4,-27.1) -- (18.4,-27.1);
       \draw [black] (21.9,-30.9) -- (33,-30.9);
       \fill [black] (33,-30.9) -- (32.2,-30.4) -- (32.2,-31.4);
       \draw (27.45,-31.4) node [below] {$\hat{R}_3$};
       \end{tikzpicture}
       \end{center}
      Hence, agent $3$ cannot receive $o_1^{t_2}(=c_2)$. Thus, by \textsl{strategy-proofness} of $f$, he still receives $y_3$, i.e., $f_3(\hat{R})=y_3$. Therefore, by \textsl{non-bossiness} of $f$, $f(\hat{R})=y$.     
     
     Let $\bar{R}_1$ be such that 
     $$\text{for each }t\in T, \bar{R}_1^t: y_1^t,\ldots, \text{ and }$$
     $$\bar{\pi}_1:t,\ldots$$
Let $$\bar{R}\equiv (\bar{R}_1,\hat{R}_{-1}).$$

Then, since $\bar{R}_1$ is a monotonic transformation of $\hat{R}_1$ at $y$, $f(\bar{R})=y$. In particular,
\begin{equation}
   f_2^{t_1}(\bar{R})=o_1^{t_2}(=c_2). \label{eeeee}
\end{equation}

Note that at $\bar{R}$, there is a first step top trading cycle $C'=(1\to y_1^t(=o_3^t)\to 3\to o_1^{t_2}\to 1)$ that involves two agents. Thus, by Claim~\ref{claim:bttceff1}, $C'$ is executed. See the figure below for the graphical explanation.

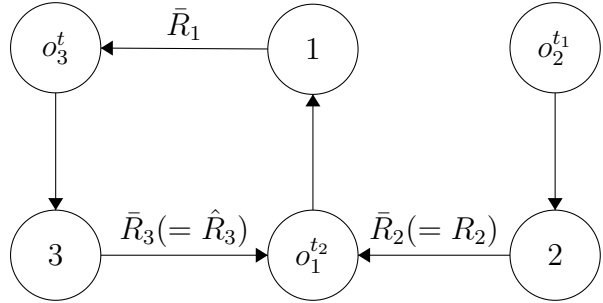
\begin{center}
       \begin{tikzpicture}[scale=0.2]
       \tikzstyle{every node}+=[inner sep=0pt]
       \draw [black] (36,-30.9) circle (3);
       \draw (36,-30.9) node {$o_1^{t_2}$};
       \draw [black] (36,-17.2) circle (3);
       \draw (36,-17.2) node {$1$};
       \draw [black] (52.1,-17.1) circle (3);
       \draw (52.1,-17.1) node {$o_2^{t_1}$};
       \draw [black] (52.1,-30.9) circle (3);
       \draw (52.1,-30.9) node {$2$};
       \draw [black] (18.9,-30.9) circle (3);
       \draw (18.9,-30.9) node {$3$};
       \draw [black] (18.9,-17.1) circle (3);
       \draw (18.9,-17.1) node {$o_3^t$};
       \draw [black] (36,-27.9) -- (36,-20.2);
       \fill [black] (36,-20.2) -- (35.5,-21) -- (36.5,-21);
       \draw [black] (52.1,-20.1) -- (52.1,-27.9);
       \fill [black] (52.1,-27.9) -- (52.6,-27.1) -- (51.6,-27.1);
       \draw [black] (49.1,-30.9) -- (39,-30.9);
       \fill [black] (39,-30.9) -- (39.8,-31.4) -- (39.8,-30.4);
       \draw (44.05,-30.4) node [above] {$\bar{R}_2(=R_2)$};
       \draw [black] (18.9,-20.1) -- (18.9,-27.9);
       \fill [black] (18.9,-27.9) -- (19.4,-27.1) -- (18.4,-27.1);
       \draw [black] (21.9,-30.9) -- (33,-30.9);
       \fill [black] (33,-30.9) -- (32.2,-30.4) -- (32.2,-31.4);
       \draw (27.45,-30.4) node [above] {$\bar{R}_3(=\hat{R}_3)$};
       \draw [black] (33,-17.18) -- (21.9,-17.12);
       \fill [black] (21.9,-17.12) -- (22.7,-17.62) -- (22.7,-16.62);
       \draw (27.45,-16.62) node [above] {$\bar{R}_1$};
       \end{tikzpicture}
       \end{center}

It implies that 
$f_3^{t_2}(\bar{R})=o_1^{t_2}$, which contradicts with $f(\bar{R})=y=f(\hat{R})$ and (\ref{eeeee}). 
\end{proof}

\begin{lemma}
    If a mechanism $f:\mathcal{R}^N_l \to X$ is \textsl{individually rational}, \textsl{strategy-proof}, \textsl{non-bossy}, and \textsl{pairwise efficient}, then for each $R\in \mathcal{R}^N_l$, each first step top trading cycle $C(\in \mathcal{C}(R) )$ is fully executed under $f$ at $R$. \label{lemma:bttceff2}
   \end{lemma}
   \begin{proof}[\textbf{Proof}]
   Let $C\in\mathcal{C}(R)$ be a first step top trading cycle that consists of agents $S_C\subseteq N$. We prove this lemma by induction on $|S_C|$.
   
    \noindent\textbf{\textit{Induction Basis.}} $|S_C|\leq 2$. This is done by Lemma~\ref{lemma:bttceff1}. \smallskip
    
   \noindent\textbf{\textit{Induction hypothesis.}} Let $K\in \{3,\ldots,n\}$. Suppose that $C$ is fully executed when $|S_C|<K$. 
   
   \noindent\textbf{\textit{Induction step.}} Let $|S_C|=K$. Without loss of generality, assume that $S_C=\{1,\ldots,K\}$ and $C=(1\to c_1\to 2\to c_2\to\ldots \to K\to c_K\to 1)$. 
   
   Similar to Lemma~\ref{lemma:bttceff1}, we first show that $C$ is executed.
   \begin{claim}
   $C$ is executed.\label{claim:bttc3}
   \end{claim}
   \begin{proof}

    By contradiction, assume that $C$ is not executed. Thus, there is an agent $i\in S_C$ who does not receive $c_i$, i.e., $f_i^{t_i}(R)\neq c_i$. Without loss of generality, let $i=2$.
    
      Let $\hat{R}_2$ be such that agent $2$ only wants to receive type-$t_2$ object $c_2$ and no other objects, i.e.,
     $$\hat{R}^{t_2}_2:c_2(=o_3^{t_2}),o^{t_2}_2,\ldots,$$
   $$\text{ for each }t\in T\setminus \{t_2\}:\hat{R}_2^t:o_2^t,\ldots, \text{ and }$$
   $$\hat{\pi}_2=\pi_2:t_2,\ldots.$$
   
    Note that at $\hat{R}_2$, if agent $2$ does not receive $c_2$, then from individual rationality of $f$, he must receive his full endowment $e_2$.
   
   Let $\hat{R}\equiv (\hat{R}_2,R_{-2})$. We proceed in two steps.
   
   \noindent\textbf{Step~1.} We show that agent $2$ receives $c_2$ under $f$ at $\hat{R}$, i.e., $f_2^{t_2}(\hat{R})=c_2$.
   
   By \textsl{individual rationality} of $f$, $f_2^{t_2}(\hat{R})\in \{c_2,o_2^{t_2}\}$. By \textsl{strategy-proofness} of $f$, $f_2^{t_2}(R)\neq c_2$ implies that $f_2^{t_2}(\hat{R})\neq c_2$, otherwise instead of $R_2$, agent $2$ has an incentive to misreport $\hat{R}_2$ at $R$. Thus, $f_2^{t_2}(\hat{R})=o_2^{t_2}$. Then, by \textsl{individual rationality} of $f$, $f_2(\hat{R})=e_2$. Thus, agent $1$ cannot receive $c_{1}(\in e_2)$ from agent $2$ because it is assigned to agent $2$. 
   
   Let $y\equiv f(\hat{R})$. Overall, we find that
   \begin{equation}
      y_2=e_2\text{ and }y_1^{t_1}\neq c_1(=o_2^{t_1}). \label{equ1}
   \end{equation}

Let $\bar{R}_1$ be such that

$$\bar{R}^{t_1}_1:c_1(=o_2^{t_1}),o_3^{t_1},o_1^{t_1},\ldots,$$  
$$\text{ for each }t\in T\setminus \{t_1\}:\bar{R}_1^t:=\hat{R}_1^t(=R_1^t) \text{ and }$$
   $$\bar{\pi}_1=\hat{\pi}_1=\pi_1:t_1,\ldots$$     

Note that $\bar{R}_1$ and $\hat{R}_1$ only differ in type-$t_1$ marginal preferences. Let $$\bar{R}\equiv (\bar{R}_1,\hat{R}_2,\hat{R}_3,\ldots,\hat{R}_n).$$

To obtain the contradiction, we want to show that at $\bar{R}$, agent $1$ receives $c_1$ and agent $2$ receives $c_2$, i.e., $f_1^{t_1}(\bar{R})=c_1=o_2^{t_1}$ and $f_2^{t_2}(\bar{R})=c_2=o_3^{t_2}$. \medskip

By \textsl{individual rationality} of $f$, $f_1^{t_1}(\bar{R})\in \{o_2^{t_1},o_3^{t_1},o_1^{t_1}\}$. By \textsl{strategy-proofness} of $f$, $f_1^{t_1}(\hat{R})\neq o_2^{t_1}$ implies that $f_1^{t_1}(\bar{R})\neq o_2^{t_1}$, otherwise instead of $\hat{R}_1$, agent $1$ has an incentive to misreport $\bar{R}_1$ at $\hat{R}$.

Thus, $f_1^{t_1}(\bar{R})\in \{o_3^{t_1},o_1^{t_1}\}$. Next, we show that $f_1(\bar{R})=e_3$ and hence  $f_1^{t_2}(\bar{R})=o_3^{t_2}=c_2$. 

Let $\tilde{R}_1$ be such that  

$$\tilde{R}^{t_1}_1:o_3^{t_1},o_1^{t_1},\ldots,$$  
$$\text{ for each }t\in T\setminus \{t_1\}:\tilde{R}_1^t:=\bar{R}_1^t(=R_1^t) \text{ and }$$
$$\tilde{\pi}_1=\bar{\pi}_1=\hat{\pi}_1=\pi_1:t_1,\ldots$$   

Since $f_1^{t_1}(\bar{R})\neq c_1$, $\tilde{R}_1$ is a monotonic transformation of $\bar{R}_1$ at $f(\bar{R})$.  Thus, 
\begin{equation}
    f(\bar{R})=f(\tilde{R}_1,\bar{R}_{-1}). \label{equ12345}
\end{equation}

Note that at $(\tilde{R}_1,\bar{R}_{-1})$, there is a first step top trading cycle $C'=(1\to o_3^{t_1}\to 3\to c_3\to\ldots \to K\to c_K\to 1)$. Since $C'\in \mathcal{C}(\tilde{R}_1,\bar{R}_{-1})$ and $|S_{C'}|=K-1$, by the Induction hypothesis, $C'$ is fully executed. Thus, $f_1(\tilde{R}_1,\bar{R}_{-1})=e_3$, and in particular, $f_1^{t_2}(\tilde{R}_1,\bar{R}_{-1})=o_3^{t_2}=c_2$. Together with (\ref{equ12345}), we conclude that $f_1^{t_2}(\bar{R})=o_3^{t_2}=c_2$ and $f_1(\bar{R})=e_3$. Therefore, $f_2^{t_2}(\bar{R})\neq o_3^{t_1}=c_2$. Hence, by \textsl{individual rationality} of $f$, $f_2(\bar{R})=e_2$, and in particular, $f_2^{t_1}(\bar{R})=o_2^{t_1}=c_1$. 

This implies that $c_1=f_2^{t_1}(\bar{R}) \mathbin{\bar{P}_1^{t_1}} f^{t_1}_1(\bar{R}) $ and $c_2=f_1^{t_2}(\bar{R}) \mathbin{\bar{P}_2^{t_2}} f^{t_2}_2(\bar{R}) $. Hence, $f_2(\bar{R})  \mathbin{\bar{P}_1}f_1(\bar{R}) $ and $f_1(\bar{R})  \mathbin{\bar{P}_2}f_2(\bar{R}) $, in which contradicts with \textsl{pairwise efficiency} of $f$.

Overall, by contradiction we show that at $\bar{R}$, agent $1$ receives $c_1$.  Together with \textsl{individual rationality} of $f$, it implies that agent $2$ receives $c_2$ at $\bar{R}$. 
Subsequently, by \textsl{strategy-proofness} of $f$, agent $1$ also receives $c_1$ at $\hat{R}$; otherwise he has an incentive to misreport $\bar{R}_1$ at $\hat{R}$. Again, together with \textsl{individual rationality} of $f$, it implies that agent $2$ receives $c_2$ at $\hat{R}$.

   \noindent\textbf{Step~2.} We show that agent $2$ receives $c_2$ under $f$ at $R$, i.e., $f_2^{t_2}(R)=c_2$.
   
   Note that $c_2$ is agent $2$'s most preferred type-$t_2$ object at $R_2$. By \textsl{strategy-proofness} of $f$, $f_2(R)\mathbin{R_2}f_2(\hat{R})$. Hence, $f_2^{t_2}(R) \mathbin{R^{t_2}_2} f_2^{t_2}(\hat{R}) $, which implies that $f_2^{t_2}(R)=c_2$.   
   \end{proof}

   Next, we show that $C$ is fully executed at $f(R)$. Let $x\equiv bTTC(R)$, $y\equiv f(R)$. Note that if $C$ is fully executed, then for each $i\in S_C$, $y_i=f_i(R)=x_i$.
   
   By contradiction, suppose that there is an agent $i\in S_C$ such that $y_i \neq x_i$. Without loss of generality, let $i=1$. By Claim~\ref{claim:bttc3}, $C$ is executed under $f$ at $R$. In particular,
   \begin{equation}\label{y1t1}
       y_1^{t_1}=o_2^{t_1}=x_1^{t_1} \mbox{ and } y_K^{t_K}=o_1^{t_K}=x_K^{t_K}.
   \end{equation}
    Since $y_1 \neq x_1(=e_2)$, there is a type $t \in T\setminus \{t_1\}  $
   and an agent $j \neq 2$ such that  $y_1^t =o_j^t$.  There are two cases.
   \smallskip
   
   \noindent \textbf{Case 1: $j\in S_C$}.
   Let $\hat{R}_1$ such that agent $1$ positions $y_1$ at the top and moves $t$ to  the most important, i.e.,
   (i) $\hat{\pi}_1:t,\ldots$; and (ii) for each $\tau\in T$, $\hat{R}_1^\tau: y_1^\tau,\ldots$. Since $\hat{R}_1$ is a monotonic transformation of $R_1$ at $y$, we have
   \begin{equation}\label{whywhy}
       f(\hat{R}_1,R_{-1})=f(R)=y \text{ and particularly, }f_1^{t_1}(\hat{R}_1,R_{-1})=o_2^{t_1} .
   \end{equation}
   Note that there is a first step top trading cycle $C'\equiv (1\to o_j^{t}\to j\to o_{j+1}^{t_j}\to j+1\to \cdots \to K\to o_1^{t_K}\to 1)$ at $(\hat{R}_1,R_{-1})$. i.e., $C'\in \mathcal{C}(\hat{R}_1,R_{-1})$. Since $j\neq 2$, $C'$ contains less than $K$ agents. Thus, by the induction hypothesis, $C'$ is fully executed at $f(\hat{R}_1,R_{-1})$. Therefore, $f_1(\hat{R}_1,R_{-1})=e_j$ and hence
   $f_1^{t_1}(\hat{R}_1,R_{-1})=o_j^{t_1}$, which contradicts with the fact that $f_1^{t_1}(\hat{R}_1,R_{-1})=o_2^{t_1}$ (see (\ref{whywhy})).
   \bigskip
   
   \noindent \textbf{Case 2: $j\not \in S_C$}.
   Let $\hat{R}_j$ be such that
   $$\hat{R}_j^{t_K}: o_1^{t_K},y_j^{t_K},\ldots$$
   $$\text{for each } \tau \in T\backslash \{t_K\},\hat{R}_j^\tau: y_j^\tau,\ldots, \text{ and}$$
   $$\hat{\pi}_j:t_K,\ldots$$

   Let $$\hat{R}\equiv (\hat{R}_j,R_{-j}).$$
   Note that $C$ is still a first step top trading cycle at $\hat{R}$, and hence, by Claim~\ref{claim:bttc3}, $C$ is executed. In particular, with (\ref{y1t1}), we have
   \begin{equation}\label{o1t1}
       f_K^{t_K}(\hat{R})=y_K^{t_K}=o_1^{t_K}.
   \end{equation}
   
   So, agent $j$ does not receive $o_1^{t_K}$ at $f(\hat{R})$. So, by strategy-proofness of $f$, $f_j(\hat{R})=y_j$; otherwise he has an incentive to misreport $R_j$ at $\hat{R}$. So, by non-bossiness of $f$, $f(\hat{R})=y$.
   
   Let $\bar{R}_1$ be such that agent $1$ positions $y_1$ at the top and moves $t$ to  the most important, i.e.,
   $$\text{for each }\tau\in T,\bar{R}_1^\tau: y_1^\tau,\ldots \text{ and}$$
   $$\bar{\pi}_1:t,\ldots$$
   Let
   $$\bar{R}\equiv (\bar{R}_1,\hat{R}_{-1} ).$$
   Since $\bar{R}_1$ is a monotonic transformation of $ \hat{R}_1(=R_1)$ at $y$, we have
   \begin{equation}\label{bary}
   f(\bar{R} )=y\text{ and particularly, }f_K^{t_K}(\bar{R})=o_1^{t_K}.
   \end{equation}

   Note that $C'\equiv (1\to y_1^t(=o_j^t) \to j \to o_1^{t_K}\to 1)$ is a first step top trading cycle at $\bar{R}$, i.e., $C' \in \mathcal{C}(\bar{R})$. See the figure below for the graphical explanation.
   
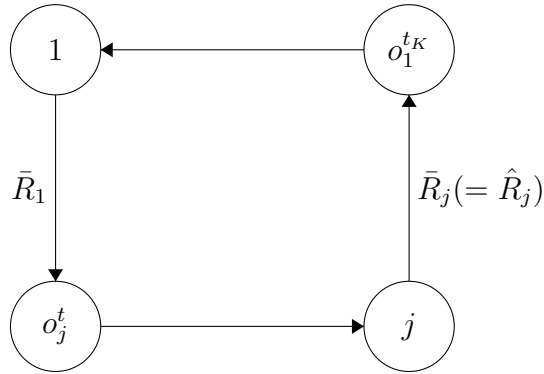
\begin{center}
       \begin{tikzpicture}[scale=0.2]
       \tikzstyle{every node}+=[inner sep=0pt]
\draw [black] (23.4,-20.3) circle (3);
\draw (23.4,-20.3) node {$1$};
\draw [black] (46.9,-20.3) circle (3);
\draw (46.9,-20.3) node {$o_1^{t_K}$};
\draw [black] (46.9,-38.7) circle (3);
\draw (46.9,-38.7) node {$j$};
\draw [black] (23.4,-38.7) circle (3);
\draw (23.4,-38.7) node {$o_j^t$};
\draw [black] (26.4,-38.7) -- (43.9,-38.7);
\fill [black] (43.9,-38.7) -- (43.1,-38.2) -- (43.1,-39.2);
\draw [black] (46.9,-35.7) -- (46.9,-23.3);
\fill [black] (46.9,-23.3) -- (46.4,-24.1) -- (47.4,-24.1);
\draw (47.4,-29.5) node [right] {$\bar{R}_j(=\hat{R}_j)$};
\draw [black] (43.9,-20.3) -- (26.4,-20.3);
\fill [black] (26.4,-20.3) -- (27.2,-20.8) -- (27.2,-19.8);
\draw [black] (23.4,-23.3) -- (23.4,-35.7);
\fill [black] (23.4,-35.7) -- (23.9,-34.9) -- (22.9,-34.9);
\draw (22.9,-29.5) node [left] {$\bar{R}_1$};
       \end{tikzpicture}
\end{center}
Thus, by Claim~\ref{claim:bttc3}, cycle $C'$ is executed at $f(\bar{R})$. Therefore, $f_j^{t_K}(\bar{R})=o_1^{t_K}$. Since $j\not \in S_C$, this
   contradicts with the fact that $f_K^{t_K}(\hat{R})=o_1^{t_K}$ (see (\ref{bary})).    
\end{proof}
 By Lemma~\ref{lemma:bttceff2}, we have shown that agents who trade at step~1 of the bTTC algorithm always receive their bTTC allotments under $f$. Next, we can consider agents who trade at step~2 of the bTTC algorithm by
following the same proof arguments as for first step trading cycles, and so on. Thus, the proof of Theorem~\ref{thm:bttc3} is completed.

\subsubsection*{Proof of Theorems~\ref{thm:bttc} and \ref{thm:bttc2}}\ 
Next, we extend Theorem~\ref{thm:bttc3} from lexicographic preferences to separable preferences.\medskip

Let $f:\mathcal{R}^N_s\to X$ be \textsl{individually rational}, \textsl{strategy-proof}, \textsl{non-bossy}, and \textsl{pairwise efficient}. Note that by Lemma~\ref{lemma:mon}, $f$ is \textsl{monotonic}.
    
    Let $S\subseteq N$ and $R \in \mathcal{R}^N_s$ be such that only agents in $S$ do no have lexicographic preferences, i.e., $R_S\not \in \mathcal{R}_l^S$ and $R_{-S}\in \mathcal{R}_l^{-S}$. We show that $f(R)=bTTC(R)$ by induction on $|S|$.
    
    \noindent We first consider $S=\{i\}$, i.e., $|S|=1$ as the induction basis.
    
    Let $x\equiv f(R)$ and $y\equiv bTTC(R)$. 
    
    Let $\hat{R}_i\in \mathcal{R}_l$ be such that for each $t\in T$, $\hat{R}_i^t:x_i^t,\ldots$
    
    By \textsl{monotonicity} of $f$, $f(\hat{R}_i,R_{-i}) =x$. Note that $(\hat{R}_i,R_{-i}) \in \mathcal{R}^N_l$. Thus, by Theorem~\ref{thm:bttc3}, $f$ coincides with bTTC, i.e., $bTTC(\hat{R}_i,R_{-i}) =f(\hat{R}_i,R_{-i}) =x$.
    
    Let $\bar{R}_i\in \mathcal{R}_l$ be such that (a) $\bar{\pi}_i=\hat{\pi}_i$; and (b) for each $t\in T$, $\bar{R}_i^t:y_i^t,\ldots$
    
    By \textsl{monotonicity} of $bTTC$, $bTTC(\bar{R}_i,R_{-i})=y$. 
    Note that $(\bar{R}_i,R_{-i}) \in \mathcal{R}^N_l$. Thus, again by Theorem~\ref{thm:bttc3}, $f(\bar{R}_i,R_{-i})=bTTC(\bar{R}_i,R_{-i})=y$.
    
    By \textsl{strategy-proofness} of $bTTC$, $y_i =bTTC_i(R)\mathbin{R_i}bTTC_i (\hat{R}_i,R_{-i})= x_i$; by \textsl{strategy-proofness} of $f$, $x_i=f_i(R) \mathbin{R_i} f_i(\bar{R}_i,R_{-i})=y_i$.
    Thus, $x_i=y_i$. Subsequently, by \textsl{non-bossiness} of $bTTC$, $x=bTTC(\hat{R}_i,R_{-i})=bTTC(\bar{R}_i,R_{-i})=y$.
    
    We can apply repeatedly the same argument to obtain that for $|S|\in \{2,\ldots,n\}$, and for each profile $R\in \mathcal{R}^N_s$
    where exactly $|S|$ agents have non-lexicographic preferences, $f(R)=bTTC(R)$. 
    Thus, for each $R\in \mathcal{R}^N_s$, $f(R)=bTTC(R)$. Thus, the proof of Theorem~\ref{thm:bttc2} is completed.\medskip

Theorem~\ref{thm:bttc} can be proven by exactly the same way to Theorem~\ref{thm:bttc2} and hence we omit it.

\subsection{Proof of Theorem~\ref{thm:mtpe}}
\label{appendix:thm:mtpe}
Let $N=\{1,2\}$ and $T=\{1,2,3\}$. Let $R\in \mathcal{R}^N_l$ be such that 

$$R_1:o_2^1,\bm{o_1^1},\bm{o_1^3},o_2^3,o_2^2,\bm{o_1^2}.$$
$$R_2:o_1^1,\bm{o_2^1},o_1^3,\bm{o_2^3},o_1^2,\bm{o_2^2}.$$

So, agent $1$ would like to trade type-$1$ and type-$2$ but not type-$3$, and agent $2$ would like to trade all types.

Let $f$ be \textsl{individually rational} and \textsl{$T'$-types pairwise efficient}. Thus, by \textsl{$T'$-types pairwise efficiency} of $f$, two agents trade in type-$1$ and type-$2$, i.e.,
 $f^1_1(R)=o_2^1,f^1_2(R)=o_1^1$ and $f^2_1(R)=o_2^2,f^2_2(R)=o_1^2$. We only need to consider the allocation in type-$3$. There are two cases.\smallskip
 
 \noindent\textbf{Case~1.} Agents trade in type-$3$, i.e., $f^3_1(R)=o_2^3$ and $f^3_2(R)=o_1^3$. Then $f_1(R)= (o_2^1,o_2^2,o_2^3)$.

 Let $R'_1: \bm{o_1^3},o_2^3,o_2^1,\bm{o_1^1},o_2^2,\bm{o_1^2}.$ By \textsl{individual rationality} of $f$, agent $1$ receives his type-$3$ endowment, i.e., $f_1^3(R'_1,R_2)=o_1^3$. 
 By \textsl{$T'$-types pairwise efficiency} of $f$, agents $1$ and $2$ still trade in type-$1$ and type-$2$ at $(R'_1,R_2)$, i.e., 
 $f^1_1(R'_1,R_2)=o_2^1,f^1_2(R'_1,R_2)=o_1^1$ and $f^2_1(R'_1,R_2)=o_2^2,f^2_2(R'_1,R_2)=o_1^2$. Thus, $f_1(R'_1,R_2)=(o_2^1,o_2^2,o_1^3)\mathbin{P_1} (o_2^1,o_2^2,o_2^3)=f_1(R)$, which implies that $f$ is not \textsl{strategy-proof}.\smallskip

\noindent\textbf{Case~2.} Agents do not trade in type-$3$, i.e., $f^2_1(R)=o_1^2$ and $f^2_2(R)=o_2^2$. Then $f_2(R)=(o_1^1,o_1^2,o_2^3)$.

Let $y=(y_1,y_2)=((o_2^1,o_1^2,o_2^3),(o_1^1,o_2^2,o_1^3))$, i.e., agents only trade in type-$1$ and type-$3$.

Let $R'_2:o_1^3,\bm{o_2^3},\bm{o_2^1},o_1^1, \bm{o_2^2},o_1^2.$ 
If agents do not trade in type-$3$ at $(R_1,R'_2)$, then by \textsl{individual rationality} of $f$, $f_2(R_1,R'_2)=o_2$ and hence $f(R_1,R'_2)=e$. This contradicts \textsl{$T'$-types pairwise efficiency} of $f$, since $y_1\mathbin{P_1}o_1$ and $y_2\mathbin{P_2}o_2$. Thus, type-$3$ is traded at $(R_1,R'_2)$. Therefore, by \textsl{individual rationality} of $f$, type-$1$ is also traded, otherwise $o_1\mathbin{P_1} f_1(R_1,R'_2)$. Thus, $f_2^1(R_1,R'_2)=o_1^1$ and $f_2^3(R_1,R'_2)=o_1^3$. So, $f_2(R_1,R'_2)\mathbin{P_2}(o_1^1,o_1^2,o_2^3)=f_2(R)$, which implies that $f$ is not \textsl{strategy-proof}.

The extension for more than two agents and three types is easy and thus we omit it.

\subsection{An example for endnote~\ref{myfootnote}}\label{footnoteexample}
The following efficiency property is an adaptation of \citet{papai2007}'s \textsl{restricted efficiency}.\footnote{In \citet{papai2007}, \textsl{restricted efficiency} is defined originally as follows. Let $X'\subsetneq X$ and $f:\mathcal{R}^N\to X'$. $f$ is \textit{restricted efficient} if for each $R\in \mathcal{R}^N$, there does not exist $x\in X'$ such that $x$ Pareto dominates $f(R)$ at $R$.}

Let $y\in X$ and $Y\equiv \{y,e\}$.
$f:\mathcal{R}^N\to Y$ is $Y$-\textit{restricted efficient} if for each $R\in \mathcal{R}^N$, there does not exist $x\in Y$ such that $x$ Pareto dominates $f(R)$ at $R$.

This is a weak efficiency property since it rules out the extremely inefficient no-trade mechanism. Also, it is easy to see that cTTC and bTTC do not satisfy this efficiency property. However, due to its restriction (feasibility of only two allocations), this property is uninteresting.
Next, we show that this property is compatible with \textsl{individual rationality} and \textsl{group strategy-proofness}.

Let $f:\mathcal{R}^N\to Y$ be such that for each $R\in \mathcal{R}^N$, if for each $i\in N$, $y\mathbin{R_i} e$, then $f(R)=y$, otherwise $f(R)=e$. This mechanism resembles a form of unanimous voting, i.e., it always selects the status quo allocation ($e$) unless all agents unanimously prefer $y$ to $e$. By the definition of $f$, it is easy to see it is \textsl{individually rational}, \textsl{group strategy-proof}, and $Y$-\textit{restricted efficient}.

\subsection{An example for endnote~\ref{footnotece}}\label{footnoteceexample}
Consider the case  where there are only two agents and two types, and $O=\{H_1,H_2,C_1,C_2\}$. Consider the following preferences.
$$R_1:(H_2,C_2),(H_1,C_1),(H_2,C_1),(H_1,C_2);$$
$$R_2:(H_1,C_1),(H_2,C_2),(H_1,C_2),(H_2,C_1).$$

Note that $R$ is strict but not separable. The unanimously best allocation is $x=( (H_2,C_2),(H_1,C_1)  )$, and there is another coordinatewise efficient allocation 
$y=( (H_1,C_1),(H_2,C_2) )$.

\section{Appendix: independence of the properties}\label{section:Appendix2}
We provide several examples to establish the logical independence of the properties in our characterizations.\smallskip

\begin{tabular}{|l|l|l|l|l|l|l|}
    \hline
     & Example~3 (NT) & Example~4 (SD) & Example~5 (MSIR) & Example~9 & cTTC & bTTC \\\hline
    IR         & +         & $-$          & +         & +         & +    & +    \\\hline
    SP         & +         & +         & $-$          & +         & +    & +    \\\hline
    NB         & +         & +         & +         & $-$          & +    & +    \\\hline
    GSP         & +         & +         & $-$          & $-$          & $-$    & +    \\\hline
    PE         & $-$          & +         & +         & $-$          & $-$     & $-$     \\\hline
    CE         & $-$         & +         & +         & $-$          & +    & $-$     \\\hline
    pE         & $-$          & +         & +         & +         & $-$     & +   \\\hline
    \end{tabular}
\smallskip

     \textbf{Satisfaction of properties of mechanisms}. The notation ``+'' (``$-$'') in a cell means that the property is satisfied (violated) by the corresponding mechanism. 
     Abbreviations in the first column respectively refer to
     \textsl{individual rationality},   
     \textsl{strategy-proofness}, 
     \textsl{non-bossiness}, 
     \textsl{group strategy-proofness}, 
     \textsl{Pareto efficiency}, 
     \textsl{coordinatewise efficiency}, and
     \textsl{pairwise efficiency}.

     \subsection{Theorems~\ref{thm:cTTC}, \ref{thm:imposs}, \ref{thm:pce}, and \ref{thm:mtpe}}
     The following examples establish the logical independence of the properties in Theorem~\ref{thm:cTTC}. We label examples by the property that is not satisfied.

\begin{example}[\textbf{\textsl{Coordinatewise efficiency}}]\ \\
    The no-trade mechanism that always assigns the endowment allocation to each market is \textsl{individually rational}, \textsl{group strategy-proof} (and hence
    \textsl{strategy-proof} and \textsl{non-bossy}), but not \textsl{coordinatewise efficient}.\hfill~$\diamond$
    \label{example:notFESP}
\end{example}

\begin{example}[\textbf{\textsl{Individual rationality}}]\ \\
    By ignoring property rights that are established via the endowments, we can easily adjust the well-known mechanism of serial dictatorship to our setting: based on an ordering of agents, we let agents sequentially choose their allotments. Serial dictatorship mechanisms have been shown in various resource allocation models to satisfy \textsl{Pareto efficiency} (and hence \textsl{coordinatewise efficiency} and \textsl{pairwise efficiency}), \textsl{group strategy-proofness} (and hence
     \textsl{strategy-proofness} and \textsl{non-bossiness}); since property rights are ignored, they violate \textsl{individual rationality}.\hfill~$\diamond$
     \label{example:notIR}
    \end{example}

    \begin{example}[\textbf{\textsl{Strategy-proofness}}]\ \\
        \citet{biro2021} show that their Multiple-Serial-IR mechanisms are \textsl{individually rational} and \textsl{Pareto efficient} (and hence \textsl{coordinatewise efficient}), but not \textsl{strategy-proof}.
        
        A Multiple-Serial-IR mechanism is determined by a fixed order of the agents. At any preference profile and following the order, the rule lets each agent pick her most preferred allotment from the available objects such that this choice together with previous agents' choices is compatible with an \textsl{individually rational} allocation. Formally,\smallskip
	
\noindent \textbf{Input.}
	An order $\delta=(i_1,\ldots,i_n)$ of the agents and a multiple-type housing market $R\in \mathcal{R}^N$. \smallskip
	
\noindent \textbf{Step~0.} Let $Y(0)$ be the set of individually rational allocations in $X$.\smallskip
	
\noindent \textbf{Step~1.}
	Let $Y_1$ be the set of agent $i_1$'s allotments that are compatible with some allocation in $Y(0)$, i.e., $Y_1$ consists of all $y_{i_1}\in\Pi_{t\in T}O^t$ for which there exists an allocation $x\in Y(0)$ such that $x_{i_1}=y_{i_1}$.\smallskip
	
	\noindent Let $y^*_{i_1}$ be agent $i_1$'s most preferred allotment in $Y_1$, i.e., for each $y_{i_1}\in Y_1$, $y^*_{i_1} \mathbin{R_i}y_{i_1} $.\smallskip
	
	\noindent Let $Y(1)\subseteq Y(0)$ be the set of allocations in $Y(0)$ that are compatible with $y^*_{i_1}$, i.e., $Y(1)$ consists of all $x\in Y(0)$ with $x_{i_1}=y^*_{i_1}$.  \smallskip

\noindent \textbf{Step~$\bm{k=2,\ldots,n}$.}
	Let $Y_k$ be the set of agent $i_k$'s allotments that are compatible with some allocation in $Y(k-1)$.\smallskip
	
\noindent Let $y^*_{i_k}$ be agent $i_k$'s most preferred allotment in $Y_k$.\smallskip
	
\noindent Let $Y(k)\subseteq Y(k-1)$ be the set of allocations in $Y(k-1)$ that are compatible with $y^*_{i_k}$. \smallskip
	
\noindent \textbf{Output.}
	The allocation of the Multiple-Serial-IR mechanism associated with $\delta$ at $R$ is $MSIR(\delta,R)\equiv (y_1^*,y_2^*,\ldots,y_n^*)$.\smallskip
	
\noindent Given an order $\delta$, the associated Multiple-Serial-IR mechanism $\Delta$ assigns to each market $R$ the allocation $\Delta (R) \equiv MSIR(\delta,R)$.\smallskip

By the definition of Multiple-Serial-IR mechanisms, it is easy to see that these mechanisms are \textsl{individually rational} and \textsl{Pareto efficient}.\smallskip

We include a simple illustrative example for $n=2$ agents and $m=2$ types for completeness.\smallskip

Let $N=\{1,2\}$ and $T=\{H(ouse),C(ar)\}$. For each $i\in N$, let $(H_i,C_i)$ be agent $i$'s endowment. Let $R\in \mathcal{R}_l^N$ be given by
$$\bm{R_1}:H_2,\bm{H_1},C_2,\bm{C_1},$$
$$\bm{R_2}:H_1,\bm{H_2},\bm{C_2},C_1.$$
Consider the Multiple-Serial-IR mechanism $\Delta$ induced by $\delta=(1,2)$, i.e., agent $1$ moves first (note that since there are only two agents, when agent $1$ picks his allotment, the final allocation is completely determined). Since allocation $x\equiv ( (H_2,C_2),(H_1,C_1) )$ is \textsl{individually rational} at $R$ and $x_1=(H_2,C_2)$ is agent $1$'s most preferred allotment, $\Delta(R)=x$.

Next, consider $\bm{R'_2}:\bm{C_2},C_1,H_1,\bm{H_2}$. Note that at $(R_1,R'_2)$, only ${y\equiv  ((H_2,C_1),  (H_1,C_2)  )}$ and $e$ are \textsl{individually rational}. Thus, agent $1$ can only pick $y_1$ or $o_1$. Since $y_1\mathbin{R_1} o_1$, agent $1$ picks $y_1$ and hence $\Delta(R_1,R'_2)=y$. Finally, we see that $y_2\mathbin{R_2} x_2$, which implies that agent $2$ has an incentive to misreport $R'_2$ at $R$. Hence, the Multiple-Serial-IR mechanism induced by $\delta=(1,2)$ is not \textsl{strategy-proof}.
        \label{example:notSP}
 \hfill~$\diamond$
        \end{example}

Examples~\ref{example:notFESP}, \ref{example:notIR}, and \ref{example:notSP} are well defined for strict preferences and establish the logical independence of the properties in Theorems~\ref{thm:imposs}, \ref{thm:pce}, and \ref{thm:mtpe}.

\subsection{Theorems~\ref{thm:bttc}, \ref{thm:bttc2},  \ref{thm:coale}, and \ref{thm:bttc3}}
The following examples establish the logical independence of the properties in Theorem~\ref{thm:bttc3}.

\begin{example}[\textbf{\textsl{Pairwise efficiency}}]\ \\
    Same as example~\ref{example:notFESP}, the no-trade mechanism is \textsl{individually rational}, \textsl{strategy-proof}, and \textsl{non-bossy}, but not \textsl{pairwise efficient}.
    \hfill~$\diamond$\label{example:notBESP}
\end{example}

\begin{example}[\textbf{\textsl{Individual rationality}}]\ \\
    Same as example~\ref{example:notIR}, serial dictatorship mechanisms are \textsl{group strategy-proof} (hence
     \textsl{strategy-proof} and \textsl{non-bossy}), and \textsl{Pareto efficient} (hence \textsl{pairwise efficient}), but not \textsl{individually rational}.\hfill~$\diamond$ \label{example:notIR2}
    \end{example}

    \begin{example}[\textbf{\textsl{Strategy-proofness}}]\ \\
        Same as example~\ref{example:notSP}, Multiple-Serial-IR mechanisms satisfy \textsl{individual rationality} and \textsl{Pareto efficiency} (and hence \textsl{pairwise efficiency}), but violate \textsl{strategy-proofness}.

        We show that multiple-Serial-IR mechanisms are \textsl{non-bossy}.\medskip

   Let $\delta=(i_1,\ldots,i_n)$ be an order of the agents and let $\Delta$ denote the associated Multiple-Serial-IR mechanism. 
	
	Let $R\in \mathcal{R}_l^N$, $i\in N$, and $R'_i\in \mathcal{R}_l$. Let $R'\equiv (R'_i,R_{-i})$, $x\equiv \Delta(R)$, and $y \equiv \Delta (R')$. Assume $y_i=x_i$. We show that $y=x$.

Let $i_k \equiv i$. Since $y_i = x_i$ and for each $\ell=2,\ldots,k-1,k+1,\ldots,n$, $R'_{i_\ell}=R_{i_\ell}$, agent $i_1$'s choice at Step~$1$ under $R'$ is restricted in the same way as agent $i_1$'s choice at Step~$1$ under $R$. Thus, since $R'_{i_1}=R_{i_1}$, we have $y_{i_1}=x_{i_1}$. Similar arguments show that for each $\ell=2,\ldots,k-1,k+1,\ldots,n$, $y_{i_\ell}=x_{i_\ell}$.
Hence, $\Delta$ is \textsl{non-bossy}.
\hfill~$\diamond$
            \label{example:notSP2}
           \end{example}

           \begin{example}[\textbf{\textsl{Non-bossiness}}]\ \\ 
            \label{example:notNB} 
            Note that when there are only two agents, \textsl{non-bossiness} is trivially satisfied. Thus, at least we need three agents. So,
            consider markets with three agents and two types, i.e., $N=\{1,2,3\}$ and $T=\{1,2\}$. 
            
            Let $\mathcal{\hat{R}}\subsetneq \mathcal{R}^N_l$ be a set of markets such that for each $R\in \mathcal{\hat{R}}$, $R_1|^e$ is such that there is an agent $i\in \{2,3\}$, and agent $1$ positions agent $i$'s full endowment at the top, i.e., for each $t\in T$, $R_1^t:o_i^t,\ldots$.
            
            Let $y\in X$ be such that (i) $y_1=e_i$, (ii) $y_i=(e_1^1,e_j^2)$ and $y_j=(e_j^1,e_1^2)$, where $\{i,j\}=\{2,3\}$.
            
            Let $f$ be such that
            \begin{equation*}
            f(R)=\left\{
            \begin{aligned}
            y & , & R  \in \mathcal{\hat{R}} \text{ and } y \text{ Pareto dominates } bTTC(R),\\
            bTTC(R) & , & \text{ otherwise. }
            \end{aligned}
            \right.
            \end{equation*}
            
            Note that for $R\in \mathcal{\hat{R}}$, if $f_1(R)\neq e_i$, then there is an agent $k\in \{2,3\}$ (possibly $k=i$) who receives $e_i$ and prefers $e_i$ to $y_k$, i.e., $f_k(R)=bTTC_k(R)=e_i \mathbin{P_k}y_k$.
            
            It is easy to see that $f$ inherits \textsl{individually rational} and \textsl{pairwise efficiency} from bTTC. By the definition of $f$, one can verify that $f$ is \textsl{bossy}. We show that $f$ is \textsl{strategy-proof}.\medskip
            
            We first show that agent $1$ has no incentive to misreport. By the definition of $f$, it is easy to see that for each $R\in\mathcal{R}^N$, $f_1(R)=bTTC_1(R)$. Since $bTTC$ is \textsl{strategy-proof}, agent $1$ has no incentive to misreport.
            
            
            

            Next, we show that agents $2$ and $3$ have no incentive to misreport.  
            
            For $R\not \in \mathcal{\hat{R}}$, $f_2(R)=bTTC_2(R)$ and $f_3(R)=bTTC_3(R)$. Since $bTTC$ is \textsl{strategy-proof}, agents $2$ and $3$ have no incentive to misreport.
            
            For $R\in \mathcal{\hat{R}}$, there are two cases.
            
            \noindent\textbf{Case~1.} $f(R)=bTTC(R)$. By the definition of $f$, there is an agent $k\in \{2,3\}$ such that $bTTC_k(R)\mathbin{P_k} y_k$. Let $R'_k\neq R_k$ be a misreporting. Then, $f_k(R'_k,R_{-k})\in \{bTTC_k(R'_k,R_{-k}),y_k \}$. Since $bTTC$ is \textsl{strategy-proof}, $bTTC_k(R)\mathbin{R_k}bTTC_k(R'_k,R_{-k})$, and hence  $bTTC_k(R)=f_k(R)\mathbin{R_k}f_k(R'_k,R_{-k})$.
            
            \noindent\textbf{Case~2.} $f(R)=y$. By the definition of $f$, $y_2\mathbin{R_2} bTTC_2(R)$ and $y_3\mathbin{R_3} bTTC_3(R)$. Let $k\in \{2,3\}$ and $R'_k\neq R_k$ be a misreporting. Since $bTTC$ is \textsl{strategy-proof}, $y_k \mathbin{R_k} bTTC_k(R)\mathbin{R_k}bTTC_k(R'_k,R_{-k})$. By the definition of $f$, $f_k(R'_k,R_{-k})\in \{bTTC_k(R'_k,R_{-k}),y_k \}$. Thus, $f_k(R)=y_k \mathbin{R_k} f_k(R'_k,R_{-k})$.
             \hfill~$\diamond$
                
               \end{example}   

Examples~\ref{example:notBESP}, \ref{example:notIR2}, \ref{example:notSP2}, and \ref{example:notNB} are well defined on the domain of separable preference (strict preference) profiles and establish the logical independence of the properties in Theorem~\ref{thm:bttc2} (Theorem~\ref{thm:bttc}). These examples also 
establish the logical independence of the properties in Theorem~\ref{thm:coale} as well.


\bibliographystyle{chicago}
\bibliography{references}

\end{document}